\documentclass{article}
\usepackage{graphicx} % Required for inserting images
\usepackage[T1]{fontenc}
\usepackage{amsmath, amssymb}
\usepackage{amsthm}
\usepackage{thmtools, thm-restate}
\usepackage[dvipsnames]{xcolor}
\usepackage{fullpage}
\usepackage{hyperref}
\usepackage{cleveref}
\usepackage{thmtools}
\usepackage{xspace}
\usepackage{enumitem}
\usepackage{graphicx}
\usepackage[inkscapelatex=false]{svg}
\graphicspath{{./figures/}{./}}

\newtheorem{theorem}{Theorem}
\newtheorem{lemma}[theorem]{Lemma} 
\newtheorem{corollary}[theorem]{Corollary}
\newtheorem*{lemma*}{Lemma}

\theoremstyle{definition}
\newtheorem{definition}[theorem]{Definition}

\newcommand{\devilgame}{devil's game\xspace}
\newcommand{\Devilgame}{Devil's game\xspace}

\newcommand{\devil}{devil\xspace}
\newcommand{\human}{human\xspace}
\newcommand{\padding}{padding\xspace}

\newcommand{\semialgebraic}{semi-algebraic\xspace}
\newcommand{\Semialgebraic}{Semi-algebraic\xspace}

\newcommand{\R}{\ensuremath{\mathbb{R}}\xspace}

\newcommand{\QR}{\ensuremath{\textsc{Q}\mathbb{R}}\xspace}
\newcommand{\ER}{\ensuremath{\exists\mathbb{R}}\xspace}
\newcommand{\RH}{\ensuremath{\mathbb{R}\textsc{PH}}\xspace}
\newcommand{\PSPACE}{\ensuremath{\text{PSPACE}}\xspace}
\newcommand{\EXPSPACE}{\ensuremath{\text{EXPSPACE}}\xspace}

\newcommand{\etrinv}{\ensuremath{\textsc{ETRINV}}\xspace}
\newcommand{\fotr}{\ensuremath{\textsc{FOTR}}\xspace}
\newcommand{\quantifiedrealfeasibility}{\ensuremath{\textsc{Quantified} \textsc{Real} \textsc{Feasibility}}\xspace}
\newcommand{\quantifiedbooleanfeasibility}{\ensuremath{\textsc{Quantified Boolean Feasibility}}\xspace}

\newcommand{\fotrinv}{\ensuremath{\textsc{FOTRINV}}\xspace}

\newcommand{\rangedfotrinv}{\ensuremath{\textsc{Ranged-FOTRINV}}\xspace}
\newcommand{\planarfotrinv}{\ensuremath{\textsc{Planar-}\textsc{FOTRINV}}\xspace}
\newcommand{\planarfotrinvequal}{\ensuremath{\planarfotrinv(=)}\xspace}
\newcommand{\planarfotrinvinequal}{\ensuremath{\planarfotrinv(\leq)}\xspace}
\newcommand{\planaretrinv}{\ensuremath{\textsc{Planar-ETRINV}}\xspace}
\newcommand{\Arealhierarchy}[1]{\ensuremath{\Pi_{#1}\mathbb{R}}\xspace}
\newcommand{\Erealhierarchy}[1]{\ensuremath{\Sigma_{#1}\mathbb{R}}\xspace}
\newcommand{\Arealhierarchyinv}[1]{\ensuremath{\forall_{#1}\textsc{INV}}\xspace}
\newcommand{\Erealhierarchyinv}[1]{\ensuremath{\exists_{#1}\textsc{INV}}\xspace}
\newcommand{\GraphInPolygon}{\ensuremath{\textsc{Graph }\allowbreak\textsc{in }\allowbreak\textsc{Polygon}}\xspace}
\newcommand{\GraphInPolygonGame}{\ensuremath{\textsc{Graph }\allowbreak\textsc{in }\allowbreak\textsc{Polygon }\allowbreak\textsc{Game}}\xspace}
\newcommand{\PartialDrawingExtensibility}{\ensuremath{\textsc{Partial }\allowbreak\textsc{Drawing }\allowbreak\textsc{Extensibility}}\xspace}
\newcommand{\PlanarExtensionGame}{\ensuremath{\textsc{Planar }\allowbreak\textsc{Extension }\allowbreak\textsc{Game}}\xspace}

\newcommand{\StickGame}{\ensuremath{\textsc{Stick Game}}\xspace}
\newcommand{\PennyGame}{\ensuremath{\textsc{Penny Game}}\xspace}
\newcommand{\DiskGame}{\ensuremath{\textsc{Disk Game}}\xspace}

\newcommand{\PackingGame}{\ensuremath{\textsc{Packing Game}}\xspace}

\newcommand{\EuclideanVoronoiGame}{\ensuremath{\textsc{Euclidean Voronoi Game}}\xspace}

\newcommand{\OrderTypeGame}{\ensuremath{\textsc{Order Type Game}}\xspace}
\newcommand{\TOfotr}{\ensuremath{\textsc{Totally Ordered FOTR}}\xspace}
\newcommand{\BOfotr}{\ensuremath{\textsc{Bounded Open FOTR}}\xspace}
\newcommand{\II}{\ensuremath{\mathbf{\boldsymbol{\infty}}}\xspace}
\newcommand{\0}{\ensuremath{\mathbf{0}}\xspace}
\newcommand{\1}{\ensuremath{\mathbf{1}}\xspace}
\newcommand{\x}{\ensuremath{\mathbf{x}}\xspace}
\newcommand{\y}{\ensuremath{\mathbf{y}}\xspace}
\renewcommand{\a}{\ensuremath{\mathbf{a}}\xspace}
\renewcommand{\b}{\ensuremath{\mathbf{b}}\xspace}
\renewcommand{\c}{\ensuremath{\mathbf{c}}\xspace}
\renewcommand{\d}{\ensuremath{\mathbf{d}}\xspace}

\title{Devil's Games and \QR \\
Continuous Games complete for the \\ First-Order Theory of the Reals.}
\author{Lucas Meijer, Arnaud de Mesmay, Tillmann Miltzow, Marcus Schaefer, Jack Stade}
 \date{}

\begin{document}

\maketitle

\begin{abstract}
    We introduce the complexity class \textit{Quantified Reals} abbreviated as \QR.
    Let \fotr be the set of true sentences in the first-order theory of the reals.
    A language $L$ is in \QR, if there is a polynomial time reduction
    from $L$ to \fotr. 
    To the best of our knowledge, this is the first time this complexity class is studied.
    We show that \QR can also be defined using real Turing machines or the real RAM.
    Interestingly, it is known that deciding \fotr requires at least exponential time unconditionally~[Berman, \emph{Theoretical Computer Science, 1980}].
    
    We focus on so-called \devilgame{}s that have two defining properties: 
    \begin{itemize}
        \item Players alternate in taking turns and
        \item each turn gives an infinite continuum of possible options.
    \end{itemize}
    The two players are referred to as human and devil.
    Our paper presents four \QR-complete problems.
    
    First, we show that \fotrinv{} is \QR-complete. 
    \fotrinv has only inversion and addition constraints and all variables are contained in a bounded range.
    \fotrinv is a stepping stone for further reductions.
    
    Second, we show that the \PackingGame is \QR-complete. 
    In the \PackingGame we are given a container, as a polygonal domain and two sets of polygons, called pieces.
    One set of pieces for the human and one set for the devil.
    The human and the devil alternate by removing one of their pieces and placing it into the container.
    Both rotations and translations are allowed.
    The first player who cannot place a piece into the container loses.
    If all pieces can be placed the human wins.

    Third, we show that the \PlanarExtensionGame is \QR-complete. 
    In this game, we are given a partially drawn plane graph and the human and the devil alternate by placing
    vertices and the corresponding edges in a straight-line manner.
    The vertices and edges to be placed are prescribed beforehand and known to both players.
    The first player who cannot place a vertex while respecting planarity of the drawing loses.
    If all the vertices can be placed the human wins.

    Finally, we show that the \OrderTypeGame is \QR-complete.
    In this game, we are given an order type together with a linear order of the abstract points.
    The human and the devil alternate in placing a point in the Euclidean plane $\R^2$ following the linear order.
    The first player who cannot place a point respecting the order type loses.
    If all the points can be placed the human wins.        
\end{abstract}

\vfill

\paragraph{Acknowledgments.}
We thank Jeremy Kirn for some helpful discussions to the machine model.
T.M. and L.M. want to thank Netherlands Organisation for Scientific Research (NWO) Vidi grant VI.Vidi.213.150 for their generous support.
All authors thank Dagstuhl to give them the wonderful opportunity to meet in a relaxed and enthusiastic environment at the seminar \textit{Precision in Geometry} with seminar number 25372.
This research project originated at the workshop.
J.S. is supported by Independent Research Fund Denmark, grant 1054-00032B, and by the Carlsberg Foundation, grant CF24-1929.
\newpage

\tableofcontents

\newpage
\section{Introduction}
    We introduce the  complexity class \textit{Quantified Reals} denoted by \QR. 
    This complexity class can be defined as all problems that are equivalent to deciding the 
    First-Order Theory of the Reals under polynomial time reductions.
    We describe a framework to show \QR-completeness of \devilgame{}s.
    \Devilgame{}s have two key properties. 
    \begin{itemize}
        \item Players alternate in taking turns and
        \item each turn gives an infinite continuum of possible options.
    \end{itemize}

\subsection{Motivation}
    There is a classical puzzle, which goes as follows:

    \begin{center}
    \textit{You find yourself in hell, where you meet the devil. 
    They gesture to a round table and propose a simple game: 
    you and they will take turns placing identical coins on the tabletop, and coins may not overlap. 
    Whoever cannot place a new coin on the table loses. 
    Should you start the game?}    
    \end{center}

    \begin{center}
        \includegraphics[width=0.5\linewidth]{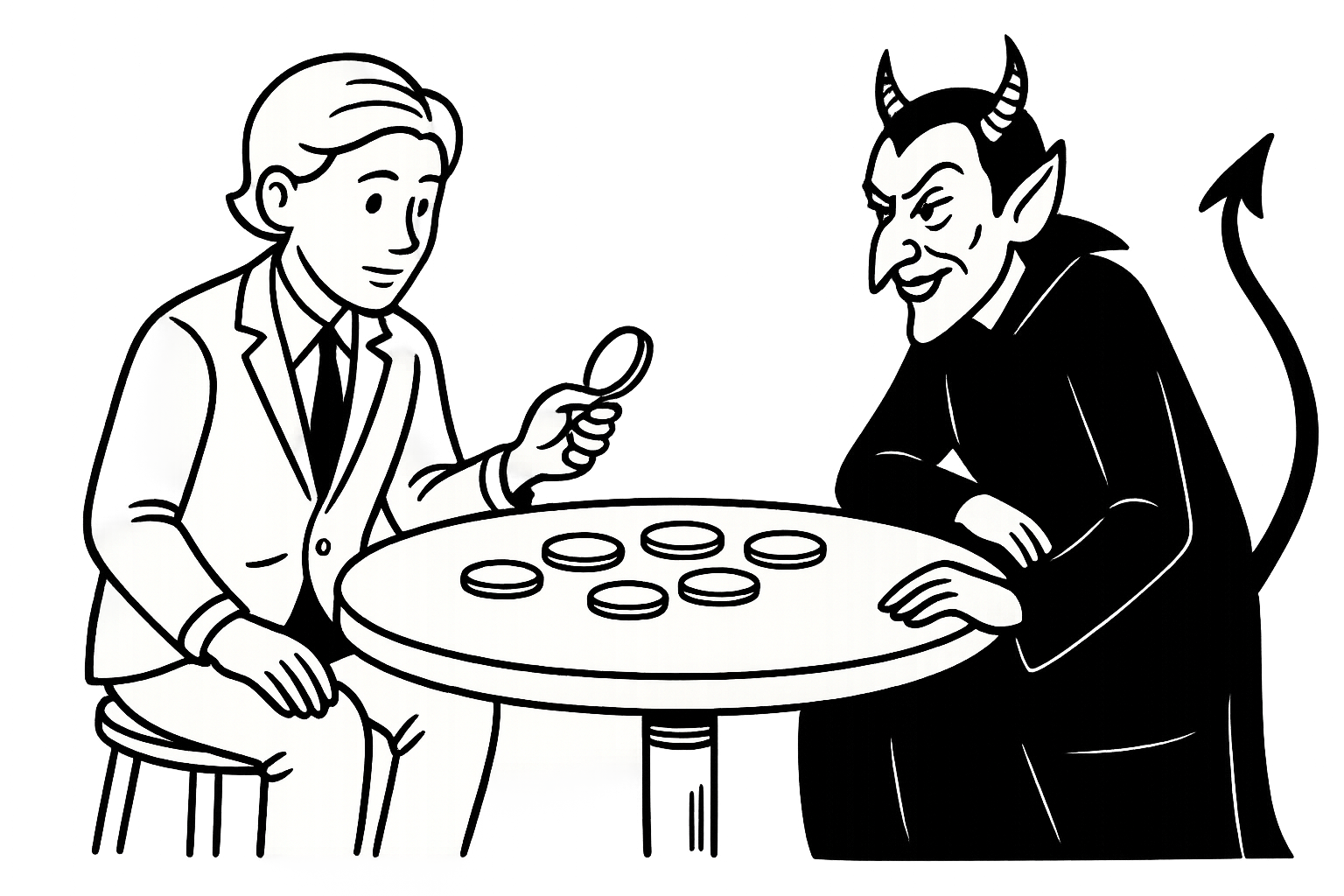}
    \end{center}
    
    (Spoiler Alert! Stop reading now, if you want to find the solution yourself.) Interestingly if the human places their first coin precisely at the center and then mirrors every move of the devil across that center, they can always respond and never be the one to run out of space first. 
    This idea uses symmetry and is a standard technique in combinatorial game theory.
    But note that if you had not placed the first coin in the middle, it seems impossible to analyze how to win this game, the reason being the two defining properties of \devilgame{}s.
    While there are plenty of results that show that combinatorial games are \PSPACE-complete,
    this seems not to capture the continuous nature of the \devilgame{}.
    And intuitively the continuous nature makes \devilgame{}s quite distinct from most other known 
    combinatorial games.
    This intuition motivates us to study \devilgame{}s more broadly.
    As a matter of fact, our results imply \textit{unconditionally} that \devilgame{}s require exponential time.
    We use the term \devilgame{}s for two reasons. 
    One is a reference to the old puzzle from above and the second is that they are devilishly difficult to analyze.

    In this article, both the human and the devil are referred to by the gender-neutral pronoun “they.”
    
\subsection{Results and Definition}

To state our results, we first need to define the first-order theory of the reals and its corresponding complexity class.
The first-order theory of the reals, which is also called the theory of real closed fields in the literature, is the problem of deciding the truth of (first-order) logical sentences with any number of quantifiers in any order. The variables can take any real number, and the formula is built using standard arithmetic operations. 
For example, the following is such a first-order logical sentence

\[ \forall x \exists y (x\cdot y = x+y) \lor (\exists z) (z > x+y) \rightarrow  ((x\cdot y)^2 > z),\]
and deciding its truth is an instance of the first-order theory of the reals. An in-depth introduction to the theory of the reals can be found in standard references such as Basu, Pollack and Roy~\cite{BasuPollackRoy2006}.

Next we define the class of the \textit{quantified reals}, \QR, of all problems equivalent to \fotr under polynomial time reductions.
Let us remark here that all problems in \QR take at least exponential time as was first shown by  Fischer and Rabin \cite{EXPTIMEhardness}.
To the best of our knowledge, no one has named the class before. 

\paragraph{\fotrinv{}.}
In order to show hardness of any problem it is convenient to have an appropriate intermediate problem.
Like, $3$-SAT for the class NP. In our case, this intermediate problem is called \fotrinv and defined as follows.

\begin{definition}[\fotrinv]
    In the problem $\fotrinv$, we are given a quantified formula $\exists x_1\forall x_2 \dots Q_n x_n : \Phi(x_1,\dots,x_n)$, where $\Phi$ consists of a conjunction between a set of equations of the form $x=1,\quad x+y=z,\quad x\cdot y=1,$
    for $x,y,z \in \{x_1, \ldots, x_n\}$.
    Each quantifier bounds exactly one variable and the quantifiers keep alternating between $\exists$ and $\forall$.
    The goal is to decide whether the system of equations has a solution when each existential variable is restricted to the range $[\tfrac12,2]$
    and each universal variable is restricted to the range $[\tfrac34,1]$.
\end{definition}

Note that the existentially-quantified variables are restricted to $[\frac12, 2]$ while the universally-quantified variables are restricted to $[\frac34, 1]$. Indeed, if $x=2$, then an addition constraint $x+y=z$ can never be satisfied. Similarly, if $z=\frac12$, then $x+y=z$ can never be satisfied. We want to be able to involve universally-quantified variables in addition constraints without immediately causing the formula to become false, so we restrict the universal variables to a smaller range.
We are now ready to state our first theorem.

\begin{restatable}{theorem}{FOTRINVTHM}
    \label{thm:FOTRINV}
\fotrinv is $Q\mathbb{R}$-complete.
\end{restatable}

We consider this our most technical result. 
The main difficulty lies in showing that we can bound (compactify to be more precise) the range of all the variables conveniently. 
The underlying idea is that if there is an assignment that makes a sentence true then we never had to pick very large values
for any of the variables.
If we have a
bound on the values of all variables involved, we can replace each variable by a scaled variable giving us small range bounds for \fotrinv.
The issue is that there is no simple bound on the value of the variables.
To see this consider the sentence
\[\forall x\neq 0 \exists y : x\cdot y = 1.\]
The sentence is clearly correct, but we cannot give an upper bound on $y$ without having some lower bound on the distance of $x$ to $0$.

The main benefit of \Cref{thm:FOTRINV} is that it allows us to use \fotrinv as a starting point for future reductions. 
Note that \etrinv played an important role to show \ER-hardness of a host of problems, for example~\cite{AAM22, BHJMW22, HKvnBV20}, so we expect that \fotrinv can play a similar role to show \QR-hardness for other interesting problems.

\paragraph{Real Polynomial Hierarchy.}
\begin{figure}[tbph]
    \centering
    \includegraphics{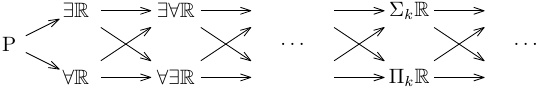}
    \caption{The different levels of the real polynomial hierarchy with inclusions indicated by arrows. }
    \label{fig:RealPolynomialHierachy}
\end{figure}
Of independent interest are fragments of the first-order theory of the reals where the number of quantifier alternations is bounded by a constant. 
This leads to the real polynomial hierarchy, a class that is similar to the polynomial hierarchy, but defined via real Turing machines instead of ordinary Turing machines~\cite{KirnMeijerMiltzowBodlaender2025, OracleSeparation, SS25}.
It is interesting as a natural complexity class.
There is a long compendium listing problems complete for the 
first level~\cite{ERCompendium}, a handful of problems known to be complete for the second level~\cite{JKM22,BC09,ERCompendium}, but as far as we are aware there are no problems known to be complete for the third level.
One of the main challenges to show hardness for the second level was the lack of an equivalent of \etrinv for higher levels of the hierarchy. 
We close this gap.
To state our results, we start with a formal definition.
\begin{definition}[Real Polynomial Hierarchy]
For integers $k\ge 0$, the real hierarchy has levels $\Arealhierarchy{k}$ and $\Erealhierarchy{k}$. The level $\Arealhierarchy{k}$ consists of problems that are polynomial-time reducible to the decision problem for first-order real formulas of the form:

\[\forall x_1\in \mathbb{R}^{n} \ \exists x_2\in \mathbb{R}^{n}\dots \ Q_k x_k\in \mathbb{R}^{n}:\Phi(x_1, \dots, x_k)\]

\noindent where $\Phi$ is quantifier-free and $Q_k$ is $\exists$ if $k$ is even or $\forall$ if $k$ is odd. 
The level $\Erealhierarchy{k}$ consists of problems that are polynomial-time reducible to the decision problem for first-order real formulas of form:

\[\exists x_1\in \mathbb{R}^{n} \ \forall x_2\in \mathbb{R}^{n}\dots \ Q_k x_k\in \mathbb{R}^{n}:\Phi(x_1, \dots, x_k)\]

\noindent where $\Phi$ is quantifier-free and $Q_k$ is $\exists$ if $k$ is odd or $\forall$ if $k$ is even.
\end{definition}

We define \RH to be the union of all the levels of a the real hierarchy, see also~\cite{OracleSeparation}. 
By a result of Renegar~\cite{R88} and Canny~\cite{C88,C88b}, we know that $\RH\subseteq \PSPACE$.

\begin{definition}[\Arealhierarchyinv{k}] 
In the problem $\Arealhierarchyinv{k}$, we are asked to decide the truth of a formula:

\[\forall x_1\in [\tfrac34, 1]^{n}\exists x_2\in [\tfrac12, 2]^{n}\dots Q_kx_k\in I_k^{n}:\Phi(x_1, \dots, x_k)\]

If $k$ is even, then $Q_k=\exists$ and $I_k=[\frac12, 2]$. If $k$ is odd, then $Q_k=\forall$ and $I_k=[\frac34, 1]$. $\Phi$ is a quantifier-free formula that is a conjunction of terms of the form $x+y=z$, $xy=1$, or $x=1$. 
\end{definition}

We also define $\Erealhierarchyinv{k}$ in an analogous way. It is well-known that $\etrinv=\Erealhierarchyinv{1}$ is $\ER$-hard, but until now, attempts to generalize this result to higher levels of the hierarchy have met with limited success~\cite{SS23,SS25}. Our techniques can be used to show that the problems $\Erealhierarchyinv{k}$ and $\Arealhierarchyinv{k}$ are complete for the respective levels of the real hierarchy so long as the innermost block of quantifiers is $\exists$.

\begin{restatable}{theorem}{HIERARCHYINVTHM}
    \label{thm:hierarchyinv}
For each $k$, $\Arealhierarchyinv{2k}$ is $\Arealhierarchy{2k}$-complete and $\Erealhierarchyinv{2k+1}$ is $\Erealhierarchy{2k+1}$-complete. 
\end{restatable}

\paragraph{Machine Model.}
\begin{figure}[t]
    \centering
    \includegraphics{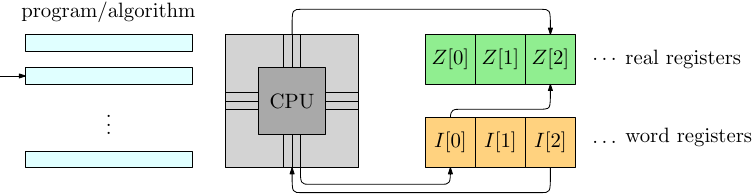}
    \caption{An illustration of a real RAM machine with a program/algorithm.}
    \label{fig:realRAM}
\end{figure}
Our next result is again a structural result.
While we have defined \QR using the first-order theory of the reals,
it is equally plausible to define it using a machine model.
Specifically, we will rely on the definition of a real RAM.
The real RAM is commonly used in computational geometry and defined explicitly for example in~\cite{EvdHM20}.
It is similar to the word RAM, but has additional real memory cells that can be manipulated using simple arithmetic operations.
Importantly, it is not allowed to round the real numbers or use them addresses.
Given an algorithm $A$ that runs on input $w$, we denote by $A(w)$ the result of the computation.
Note that the real RAM is polynomial time equivalent to a BSS-machine over the real numbers~\cite{BSS89} or a real Turing machine~\cite{OracleSeparation}. 
The main difference is that most algorithms and precise runtime analysis are done with respect to the real RAM.
The following theorem makes this machine model characterization precise.

\begin{restatable}{theorem}{MachineModelTHM}
\label{thm:MachineModel}
    A language $L$ is in \QR if and only if there is a polynomial time algorithm $A$ on a  real RAM and for each instance $w$ it holds 
    \[w \in L \Leftrightarrow \exists x_1 \in \R^n \ \forall y_1 \in \R^n \ \ldots \ \exists x_k \in \R^n \ \forall y_k \in \R^n : A(x_1,y_1,\ldots, x_k,y_k,w) = 1,\]
    for some $n,k$ polynomial in the length of $w$.
\end{restatable}
The theorem is very useful to show \QR-membership, but is also conceptually useful.

\paragraph{Packing.}
In our next result we apply \Cref{thm:FOTRINV} to show that the \PackingGame is \QR-hard and \Cref{thm:MachineModel} to show that \PackingGame is in \QR.
Before we state the theorem formally, we define the \PackingGame. See also \Cref{fig:Packing-Intro} for an illustration.
\begin{definition}[\PackingGame]
We are given two set of simple polygons $H$ and $D$, called pieces, and a container $C$ which is a polygonal domain. (A polygonal domain is a polygon that is allowed to have polygonal holes.)
The human and the devil alternate in selecting one piece from their set and placing it in the container 
$C$ without overlapping with any of the previously placed pieces.
It is allowed that vertices and edges overlap. Just no two interior points should overlap.
Furthermore, when placing the pieces it is allowed to rotate and translate the pieces.
The player that cannot place one of their pieces anymore (without repetition) loses.
If all pieces are successfully placed the human wins.
\end{definition}

We are now ready to present our third theorem.

\begin{restatable}{theorem}{PackingTHM}
    \label{thm:Packing}
The \PackingGame is $Q\mathbb{R}$-complete.
\end{restatable}

The theorem is established by a reduction from \fotrinv. Using mostly standard techniques from similar geometric reductions.
One technical strength is that both players can choose from a set of pieces and are not predetermined which piece they use.
We think the main value of \Cref{thm:Packing} is that this is a very natural game which very closely resembles the game mentioned in the motivational section.
One can easily imagine this being an actual computer game.

\begin{figure}[bthp]
    \centering
    \includegraphics[page=2]{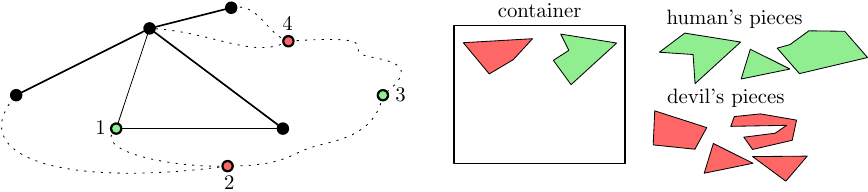}
    \caption{
     Left: The human and the devil place alternatingly vertices in the plane and the incident straight edges. 
    The black vertices and thick edges were part of the input. The vertices $2,3,4$ still need to be placed.
    Right: We see a square container and a set of pieces both for the human and the devil.
    Both players have already placed some pieces and the first player who cannot place a piece will loses. 
   }
    \label{fig:Packing-Intro}
\end{figure}

\paragraph{Graph Drawing.}
Another very natural game is the \PlanarExtensionGame. 
The idea is that two players have to draw a straight line planar graph together.
And part of the graph is already drawn.
Below is a more precise definition,
see \Cref{fig:Packing-Intro} for an illustration.

\begin{definition}[\PlanarExtensionGame]
    We are given a planar graph $G$, a planar straight-line drawing of a subgraph $G'$ in the plane, and
    a linear ordering on the vertices of $G$ not in $G'$.
    The human and the devil take turns in placing vertices according to the given linear order.
    All edges must be drawn as straight line segments without intersections.
    That is, no two edges are allowed to share a point in their interior, no two vertices are at the same location, and no vertex is allowed to lie in the interior of an edge.
    The game ends when the first player cannot place a vertex without creating an intersection.
    If the graph is drawn completely in the plane the human wins.
\end{definition}

We are now ready to state our fourth theorem.

\begin{restatable}{theorem}{Planar}
    \label{thm:Planar}
The \PlanarExtensionGame is $Q\mathbb{R}$-complete.
\end{restatable}

The main strength of the theorem is that it is a very natural game to consider.
\Cref{thm:Planar} follows the ideas of the reduction to show \ER-hardness of the \textsc{Graph in Polygon} problem~\cite{LMM}, combining it with new insights described below.

One highlight here is that we can show \QR-completeness also for the variant that requires the drawing to be simple.
That is, no edge has a vertex in its interior or overlaps with another edge. 
For the non-game variant only NP-hardness is known~\cite{P06} and \ER-hardness is an intriguing open problem~\cite{LMM}.

Notably, we develop a technique to enforce points to be placed collinearly with a line segment, despite the fact that we have no closed constraints geometric constraints.
Intuitively, we let the \human and the \devil check each other whether points are placed collinearly.
We leverage the interaction between the \human and the \devil to prove that when a player does not place a vertex collinear to a specific intended line segment, the opponent
can construct an obstruction that forces the player to create a crossing later and thus lose the game.

\paragraph{Order Types.}
Finally, we consider the \OrderTypeGame. For completeness, we first define order types.
An \emph{(abstract) order type} $O=(E,\chi)$ of size $n$ is a combinatorial abstraction of the possible locations of $n$ points in $\mathbb{R}^2$. Here, $E$ is a set of $n$ elements called the ground set and $\chi$ is a map

\[ \chi: \binom{E}{3} \rightarrow \{-1,0,1\}\]
called a chirotope satisfying a set of constraints derived from the geometric interpretation. 
We say that a point set $P \subset \mathbb{R}^2$ \emph{represents} a given order type $O = (E,\chi)$ if 
there is a bijection $\varphi$ between $E$ and $P$ such that the following holds.
Let $p = \varphi(e)$,  $q = \varphi(f)$, and $r = \varphi(g)$ then the orientation of $p,q,r$ is given by $\chi(e,f,g)$.
If $\chi(e,f,g) = 1$ then $(p,q,r)$ form a triangle in counter-clockwise order. 
If $\chi(e,f,g) = -1$ then $(p,q,r)$ form a triangle in clockwise order. 
And if $\chi(e,f,g) = 0$ then $(p,q,r)$ are on a line. See \Cref{fig:order-type-definition} for an illustration. For $E'$ a subset of $E$, we denote by $O_{|E'}$ the restricted order type on $E'$, where the chirotope is only defined on $\binom{E'}{3}$.

\begin{figure}[tpb]
    \centering
    \includegraphics{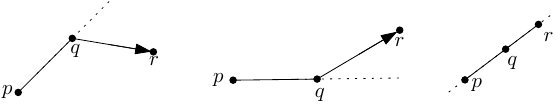}
    \caption{The orientation of a triple of points is either positive, negative, or zero. }
    \label{fig:order-type-definition}
\end{figure}
  The \textsc{Order Type Game} is defined as follows.

  \begin{definition}[\OrderTypeGame] We are given an order type $O$ and a linear ordering on the points of $O$. The human and the devil take turns in placing points in $\mathbb{R}^2$ according to the given order. At each step, the points must be placed so that the order type of the placed set of points $(p_1, \ldots , p_k)$ must represent $O_{|p_1, \ldots, p_k}$. The first player who cannot place a point loses, while the other player wins. If the entire order type is realized, the human wins.
\end{definition}

The following theorem summarizes our findings on the game.

\begin{restatable}{theorem}{OrderTypeGameTHM}
    \label{thm:OrderTypeGame}
    \OrderTypeGame is $Q\mathbb{R}$-complete.
\end{restatable}

We find \Cref{thm:OrderTypeGame} interesting for several reasons.
\begin{itemize}
    \item 
        For \ER-completeness, order type realizability is one of the most important problem 
both conceptually and as a tool to show hardness for a host of other problems~\cite{matousek2014intersection, C24, cardinal2018intersection, McDM10, HMWW24}. 
    This leads us to hope that \OrderTypeGame can also serve as a convenient starting point for reductions in the future.
    \item The techniques used to show \QR-completeness are distinct to the techniques used in the previous two theorems.
    \item We strengthen the claim that \devilgame{}s are usually \QR-complete.
\end{itemize}

In the next section, we reflect on our results
and answer questions the readers might have.

%%%%%%%%%%%%%%%%%%%%%%%%%%%%%%%%%%%%%%%%%%%
\subsection{Discussion and Context} 
%%%%%%%%%%%%%%%%%%%%%%%%%%%%%%%%%%%%%%%%%%%

\paragraph{Triple Exponentially Large Solutions.}
Readers familiar with the existential theory of the reals are aware that the solutions might have a double exponential magnitude.
A simple example is 
\[x_0 = 2^{2^{0}} = 2, \quad x_{n} = x_{n-1}^2 = 2^{2^{n-1}} \cdot 2^{2^{n-1}} = 2^{2^{n-1} + 2^{n-1}} = 2^{2^n}.\]
Thus, the formula $\exists x_0,\ldots,x_n  : x_0 = 2 \ \bigwedge_{i=1,\ldots,n} x_i = x_{i-1}^2$
has exactly one solution and this solution needs exponentially many bits if we want to describe it in binary.
On the positive side there is a matching double exponential upper bound.

Perhaps less well known is the fact that quantifier alternations allow us to create formulas that require triply-exponentially-large coordinates (see e.g. Theorem 11.18 in \cite{BasuPollackRoy2006}).
Let $\Psi_0(x,y) \equiv x = y^{2^{2^{0}}}$.
Our goal is $\Psi_n \equiv x = y^{2^{2^{n}}}$.
Assume that we have already constructed $\Psi_n$ and we want to construct $\Psi_{n+1}$.
One inefficient way is as follows:
\[ \Psi_{n+1}(x,y) \equiv \exists t \Psi_{n}(x,t) \land \Psi_{n}(t,y).\]
This is inefficient since the formula size of $\Psi_n$ would be exponential in $n$.
On the positive side, induction and a short calculation implies that 
\[x = t^{2^{(2^{n})}} \land t = y^{2^{(2^{n})}} \Rightarrow 
x = {\left(y^{2^{\left(2^{n}\right)}}\right)}^{{2^{\left(2^{n}\right)}}}
= y^{2^{\left(2^{n}\right)} \cdot 2^{\left(2^{n}\right)}} = y^{2^{\left(2^{n} + 2^n\right)}} = y^{2^{\left(2^{n+1}\right)}}.\]

And this is a triple exponentially large, as promised.
The question remains if we can define $\Psi_n$ in a more efficient way using a universal quantifier. The underlying idea is that we can reuse the formula $\Psi_n$ when constructing $\Psi_{n+1}$ by using a universal quantifier that acts as a logical ``and''.
As the two conditions $\Psi_n(x,t)$ and $\Psi_n(t,y)$ have exactly the same structure, we only need to replace the variables in a convenient way.
\[\Psi_{n+1}(x,y) \equiv  \exists t \ \forall i \in \{0,1\}  \ 
\exists a,b \  (i = 0 \Rightarrow a = x \land b = t) \land  (i = 1 \Rightarrow a = t \land b = y) \land \Psi_{n}(a,b).\]
It is not difficult to see that $\Psi_n$ has description complexity $O(n)$ and that it is equivalent to the old definition above.

This simple formula also illustrates why it is not trivial to place $\fotr$ even in \EXPSPACE{}, even if we were to know that there is an integer solution.
We use a known quantifier-elimination procedure (Theorem 14.16 in~\cite{BasuPollackRoy2006}) to obtain triply-exponential upper bounds for certain formulas (\Cref{lem:limitevaluation}). This will be used to compactify our formulas. 

\paragraph{Game Modeling.}
    In our definition we had to make a choice on how to define the games precisely.
    For example, we decided that each turn consists of placing a single piece/vertex/point.
    An alternative would be that each player has to play a set of pieces/vertices/points.
    We decided for the variant of placing a single object for two reasons.
    It seems more natural as a game and makes the \QR-hardness result stronger in a technical sense.
    However, the difference is purely cosmetic as we can just create moves where one of the players
    has to place an object, that do not affect anything of the part of the instance that we care about.

    Another choice that we made is that the first player to ``break the rules'' loses.
    An alternative definition would be to say that both players have to make a legal move if they can
    and the human only wins if the entire graph/packing/order type is fully realized.
    The reason we made the definition the way we did it is that the authors feel that a game being symmetric is more natural.
    Again, this is mainly a cosmetic difference as our \QR-hardness reductions 
    can be easily adopted.

    A third choice that we made is that we say that the human wins if all objects are placed correctly.
    An alternative would be to say that this is a tie.
    Again this is a subjective choice of the authors and the influence is purely cosmetic.
    It is easy to just give the devil in all our hardness reductions a final object that is impossible to place.
    This would have exactly the same effect.

\paragraph{Real \devilgame{}s.}
    We are aware of one game called the \textit{Magnet Game} or \textit{Kollide} which is a table top game that 
    people actually play and can be bought in a store. In this game every player has an equal number of magnets and there is a string to be placed on the table.
    The players take turns to place the magnets inside the string. The first player to get rid of all their magnets wins the game.
    However, if two magnets touch each other then the player whose turn it is has to pick up those magnets.
    Although this game is in theory extremely hard to analyse it is a fairly simple game in practice as your opponent is also
    not able to analyze it.
    Furthermore, physical dexterity seems to play a bigger role when playing the game in practice than mathematical analysis.
    (This is the author's subjective opinion.)

\paragraph{Upper Bounds for the First-Order Theory of the Reals.}
Upper bounds for the first-order theory of the reals have a long history going back to Tarski \cite{Tarski1951}, who was the first to show that the problem is even decidable. Collins \cite{Collins1975} later gave a doubly-exponential-time algorithm.
In 1986, Ben-Or, Kozen and Reif~\cite{BKR1986} claimed to have a proof that the first-order theory of the reals is in EXPSPACE. However, there was a gap in their proof (see \cite{C93}). 
Renegar~\cite{RENEGAR1992} later established that the theory can be decided in parallel exponential time, implying that it can be decided in exponential space.

\paragraph{Known Lower Bounds.}

It is actually possible to obtain \emph{unconditional} lower bounds for the first-order theory of the reals. Fischer and Rabin \cite{EXPTIMEhardness} showed that the problem is EXPTIME-hard, meaning that it cannot be decided in polynomial time. 
Berman~\cite{BermanEXPH} later extended this to show that the first-order theory of the reals is as hard as the exponential-time hierarchy. 

Specifically Berman defines the family of  complexity classes \text{Space Time Alternations}. 
For example, $\text{STA}(*,2^{n} , n)$ is the class of languages decidable in $2^{n}$ time and at most $n$ alternations on an alternating Turing Machine. 
(The $*$ indicates that there is no explicit limit on the space complexity. Since one can  never use more space than time there still is an implicit bound on the space complexity.)
He showed that there is a linear reduction from any language $L$ in 
$\text{STA}(*,2^{n} , n)$ 
to the language \textsc{RA}, the theory of the reals with addition, but without multiplication. This makes it a subtheory of \fotr.
Note that exponential space allows for doubly exponential time computations. 
And thus although $\text{STA}(*,2^{n} , n)$ allows for exponential space it does not capture all exponential space algorithms.

Note that:

\[\PSPACE=\bigcup_k\text{STA}(n^k, *, 0)=\bigcup_k\text{STA}(*, n^k, n^k)\]

The first equality is the definition of \PSPACE, and the second follows from the fact that the problem \quantifiedbooleanfeasibility is \PSPACE-complete~\cite{arora2009computational}. 
Similar arguments can be used to show the corresponding statement for \EXPSPACE: 

\[\EXPSPACE=\bigcup_k\text{STA}\left(2^{n^k}, *, 0\right)=\bigcup_k\text{STA}\left(*, 2^{n^k}, 2^{n^k}\right).\]

We are not aware of a clear reference in the literature though.
We still hope that those equations give some context to understand the lower bounds on \fotr.

\paragraph{Practical Comparison to \ER and \PSPACE.}
It is reasonable to ask how difficult it is to decide the winner in \devilgame{}s compared to decision problems with real variables or combinatorial games.
We can only give some vague estimates, but these may still help the reader to form their own opinion.
The short answer: \devilgame{}s are unthinkably difficult to solve. 

\begin{figure}
    \centering
    \includegraphics[page =2]{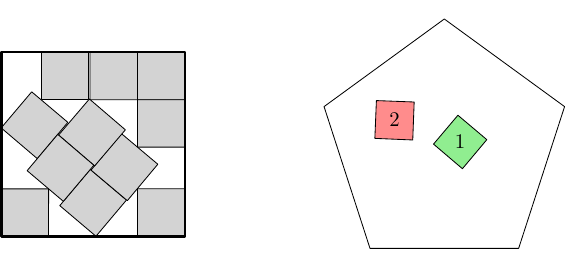}
    \caption{Left: The tightest known way to pack $11$ unit squares into a square container~\cite{Gensane2005}.  Right: How would the devil and the human need to play if the table had the shape of a pentagon?}
    \label{fig:squares}
\end{figure}

The long answer:
Let's start with combinatorial games that are PSPACE-complete.
For those games, it is typically possible to go through all combinations for at least one or two steps ahead by hand.
It is also possible to write an algorithm that can search maybe ten moves ahead, depending on the number of possible options per turn.
If the game ends in 10 turns, a computer can then compute a winning strategy.

Now, let's look at \ER-complete problems. The poster-boy example may be geometric packing.
The practical difficulty can be illustrated well with packing eleven unit squares into a minimum square container.
Surprisingly, despite being such a small number, we do not know the smallest square container and \Cref{fig:squares}
gives the smallest known packing. 
The infinitely many choices in the packing problem seem to make this simple-looking problem even harder than combinatorial games.
In particular, one needs quite sophisticated algorithms to solve polynomial equations to find a solution to the packing problem, whereas, for combinatorial games a simple tree search algorithm can be explained to and implemented by undergraduate students.
We should mention that because $\ER \subseteq \PSPACE$, the infinitely many choices in the packing problem can be turned into finitely many (binary) decisions, but the algorithm to do so is far from trivial or efficient.

So, how does \QR fit into this picture?
Combinatorial search algorithms might be possible if there is a clever heuristic to cluster the infinite number of cases into a small number of groups.
This might be possible for some specific instances.
Algebraic and numeric methods seem to be much worse in the presence of alternating quantifiers.

Consider the original devil's puzzle with a table that is not a disk but a pentagon. It seems close to impossible to analyze an optimal strategy even if the ratio between the table size and the coin size is five. 
We are not aware of any systematic study to find optimal solutions to \devilgame{}s for small instances.
Interestingly, there are many cases where a clever case distinction seems sufficient to analyze the game or where a symmetry argument reveals a winning strategy for one of the players. 

\paragraph{Oblivious Games.}
The way we defined the \devilgame{}, we gave a strict order of the moves and what is allowed in which move.
For example, in the \OrderTypeGame, each turn specifies precisely which point needs to be placed by the player.
One exception is the packing game where, we specify a set of objects for the human and the devil and each turn they can freely choose which of the pieces they want to place.
We call such games oblivious. We are wondering if it is possible to develop techniques to show that 
oblivious versions of the other game studied in this article are \QR-complete as well.
Can we define an oblivious version of \fotrinv?

\paragraph{Candidate Games.}
There are many more \devilgame{}s to study. 
We want to mention a few possible sources of such games.
In general almost every \ER-complete problem can be turned into a game, by letting the players alternate on choosing parts of the real witness.
Clearly this is more natural for some problems than for others. 
Let us highlight some games that we consider natural.

Geometric graph recognition problems can often be turned into 
\textit{geometric graph realizability games}.
Consider the recognition problem for segment intersection graphs: given a graph, is it the intersection graph of line segments in the plane?
It is known that this problem is \ER-complete~\cite{KM94,matousek2014intersection}.
If we also specify a linear order on the vertices of the graph, we can turn this question into a \devilgame{} as follows.
Players alternate in placing segments in the given order.
The first player who cannot place a segment with the correct intersection pattern loses.
If all segments are placed correctly the human wins.
We call this game the \StickGame.
If we consider contact graphs of unit disks, we get the \PennyGame~\cite{PennyGraphs}, if we consider the intersection of disks in the plane, we get the \DiskGame, and so on.

Another natural game is the \EuclideanVoronoiGame.
Players alternate claiming points in a given polygonal domain.
After $t$ moves the geodesic Voronoi diagram is computed. 
Both players compute the area of their Voronoi cells and the player with the larger area wins.
This game was studied in the literature~\cite{AhnCCGO04, CheongEH07, FeketeM05}.
We consider the \EuclideanVoronoiGame as very natural and see a strong practical motivation.
Consider a city and two super market chains want to open up stores. 
They know that one important factor for their customers is physical proximity to their super market.
We can see the \EuclideanVoronoiGame as an abstraction of the dynamics between the two chains.

\paragraph{Simple Order Types.} An order type is \emph{simple} if the chirotope only takes values in $\{-1,1\}$. Equivalently, in any realization of a simple order type, no three points lie on the same line. The problem is equivalent by point-line duality to the \textsc{Simple Stretchability} problem~\cite{shor1991stretchability} for pseudoline arrangements, which is \ER-complete~\cite{SS17} and an important starting point $\exists \mathbb{R}$-completeness reductions. Therefore, it would be interesting to have \QR-completeness for a game based on simple order type. However, the standard \emph{scattering} tool (see, e.g., Matou\v{s}ek~\cite[Section~4.3]{matousek2014intersection}) to make an order type simple does not apply easily to our gamified setting. The main issue is that the scattering requires knowing the entire order type, and thus we cannot implement it on the fly when some points have been placed but subsequent points have not been revealed yet.

\paragraph{Fixed Move Games and the Real Polynomial Hierarchy.}
We formulated our games for the case that the number of moves is part of the input.
It is also reasonable to model games with a fixed number $k$ of turns.
(We simultaneously assume that each turn becomes more complex. For example, a $k$ turn packing game where each player places $n$ pieces per turn etc.)
We believe that those games are natural candidates to be complete for the $k$-th level of the real polynomial hierarchy.
A starting point to show this is \Cref{thm:hierarchyinv} and our geometric reductions.

One curiosity is that \Cref{thm:hierarchyinv} only holds for the case that the innermost quantifier is an existential quantifier.
The technical reason is that many parts of the reduction introduce an existential quantifier at the innermost part of the formula.
For example, when we transform from $\forall y : y>0$ to 
$\forall y  \exists x: yx^2  = 1$.
This type of trick is used at many different places and we do not have alternative reductions. 
We can imagine though that clever usage of negations could help.
On the other hand, we cannot think of an underlying reason why we should see a different behavior for games where the devil makes the last move, compared to games where the human makes the last move.

\subsection{Proof Techniques}

We sketch the key ideas behind our most  technical result, \Cref{thm:FOTRINV}; to improve readability of this overview, we suppress many of the technicalities.

\paragraph{On the Importance of Being Compact.}
Membership is easy and follows immediately from the definition, as every \fotrinv instance is also an \fotr instance.
The difficult part of \Cref{thm:FOTRINV} is showing hardness. 
A key difference between \fotrinv and \fotr is that all the constraints in an instance of \fotrinv are closed (that is, there are no strict inequalities), and every variable is restricted to a compact interval.
In contrast, \fotr is allowed inequalities, and all variables are bounded by the open and unbounded set \R.
To show \QR-hardness of \fotrinv, we need to take a formula in \fotr and transform it into one with closed constraints and 
all the variables ranging over compact intervals.

Either one of these restrictions is relatively easy to achieve. Lemma 3.2 from \cite{SS17}, gives a way to make all the constraints closed, while Proposition 2.12 from~\cite{SS25} gives a way to bound all the variables. However, it is much harder to achieve both of these simultaneously.
To understand this, we notice that the results in \cite{SS17,SS25} work in some sense by constructing homeomorphisms between sets. For example, the constraint $x>0$ can be replaced by 
$\exists y:y\ge 0\wedge xy=1$. 
This strategy relies on the fact that the sets $(0,\infty)$ and $\{(x, y) | y\geq 0 \land  xy=1\}$ are homeomorphic.
Similarly, we can can bound a variable $x$ by replacing all instances of $x$ with $\varphi(x)$, where $x$ is a homeomorphism from $(-1, 1)$ to $\mathbb{R}$, such as $\varphi(x) \equiv \frac{x}{(x-1)(x+1)}$.

In contrast, the relevant sets for an instance of \fotrinv are always closed and bounded, so they are compact. 
Since $(0,\infty)$
is not compact, there is no direct way to encode the constraint $f(x)>0$ with an instance of \fotrinv. This is a key point: in order to show that \fotrinv is \QR-hard, \emph{we need to use transformations that change the topology of the sets}.

The key technical result then
is that compact \fotr is \QR-hard when restricted to compact instances. Once we show this, we can show that \fotrinv is \QR-hard using essentially the same techniques used to show \ER-hardness of \etrinv.

\paragraph{Limits and Smoothing.}
We achieve the compactification of an \fotr formula using two key ideas, which we will illustrate in a simple example.
We first introduce the notion of a limit for first-order real formulas. 
For $Q$ a quantifier (either $\exists$ or $\forall$), the formula $Qx:\Phi(x)$ is equivalent to $\lim_{b\rightarrow \infty}Qx\in [-b, b]:\Phi(x)$. 
Thus given some \fotr formula, we can replace all
unbounded variables using limits and bounded variables.
Using tools from real algebraic geometry, we show that we can choose values of the variables in the limits that ``realize'' the truth value of the formula, if all the limits are in front of the formula.
The main technical challenge is to move all the limits to the front of the formula. 

To illustrate the difficulty, consider the formula:
\[\forall y\exists x:y(xy-1)=0\]
The formula is true, but as $y$ goes to $0$, the value of $x$ required to satisfy the formula can become arbitrarily large. 
We can add limits to the quantifiers in order to obtain:

\[\lim_{a\rightarrow \infty}\forall y\in [-a, a]\lim_{b\rightarrow \infty}\exists x\in [-b, b]:y(xy-1)=0\]
However, there is no constant value of $b$ that realizes the limit $\lim_{b\rightarrow \infty}$, since we need $b\ge |\frac1y|$. 
Specifically, the formula is not equivalent to 
\[\lim_{a\rightarrow \infty} \lim_{b\rightarrow \infty} \forall y\in [-a, a]\exists x\in [-b, b]:(xy-1)=0,\]
as the second limit depends on $y$.
To illustrate the issue, we refer to the red set $S$ from \Cref{fig:CompactExample}~(a).
We notice that the projection of $S$ onto the $y$-axis is the $y$-axis. 
A fact that is equivalent to $\forall y \exists x: ((x,y)\in S)$.
However, for any bound on $x$ and $y$, the projection onto the $y$-axis misses a small part around the origin.
This second fact is a geometric way to see that the sentence $\forall y \in [-a,a] \exists x\in [-b,b]: ((x,y)\in S)$ is incorrect for any fixed value of $a$ and $b$.

\begin{figure}[tbph]
    \centering
    \includegraphics[page = 2]{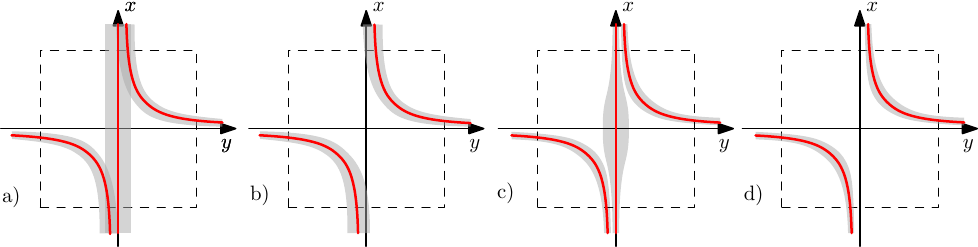}
    \caption{Attempting to smooth the formulas $\forall y\exists x:y(xy-1)=0$ and $\forall y\exists x:xy-1=0$. The goal is to thicken the formula so that the truth value depends only on the behavior inside the dashed box. a) The formula $y(xy-1) = 0$ in red and in gray its uniform smoothing.
    b) The formula $xy= 1$ in red and in gray its uniform smoothing.
    c) and d) The same formulas as in a) and b), but with non-uniform smoothing.
    All four drawings are conceptual and not perfectly accurate in order to increase readability.
    }
    \label{fig:CompactExample}
\end{figure}
The problem occurs because the set $y(xy-1)=0$ has two different behaviors, one when $y=0$ and one when $y\ne 0$. Our idea is to smooth the boundary between the two modes by thickening the instance in the $y$ direction. 
We replace the condition $\Phi(x,y) \equiv  y(xy-1)=0$ with $\Psi \equiv \exists z:z(xz-1)=0\wedge |y-z|\le \varepsilon$, as illustrated by the gray area in \Cref{fig:CompactExample}~(a).
We call this process smoothing the formula.
Before we explain the purpose of this transformation it is instructive to consider the two underlying \semialgebraic sets $S = \{(x,y)\in \R^2 : \Phi(x,y)\}$ and $T = \{(x,y)\in \R^2 : \Psi(x,y)\}$.
If a point $(a,b) \in S$ then the line segment from $(a, b-\varepsilon)$ to $(a, b+\varepsilon)$ is in $T$.
Furthermore for any set $S \subseteq \R^2$ the sentence 
$\forall y \exists x : (x,y) \in S$ is true whenever the projection of $S$ onto the $y$-axis is the entire line. 
With this geometric understanding it is easy to see that smoothing formulas makes it easier to satisfy them in our example.

The goal of smoothing the formula is to achieve two properties.
\begin{itemize}
    \item We can move all the limits to the front without changing the truth value of the formula.
    In our example that is equivalent to the second limit not depending on $y$.
    \item We also want that the smoothed formula and the original formula have the same truth value for sufficiently small $\varepsilon$.
\end{itemize}

The first property will turn out to be true, but the second property is violated.
Specifically, the formula $\forall y\exists x:xy-1=0$ is false, but the smoothed formula is true for any $\varepsilon>0$. This is illustrated in \Cref{fig:CompactExample}~(b). No matter how small we make $\varepsilon$, the solution set will intersect the $x$-axis. 

To avoid this, we thicken the solution set non-uniformly. That is, for values closer to the origin we thicken the set more, but the further we move away from the origin the smaller the thickening will be. 
So instead of constraining $|y-z|\le \varepsilon$ we ``squeeze'' the constraint by $x$, replacing it with $|y-z|\le \varepsilon\frac{1}{1+x^2}$. This is illustrated in \Cref{fig:CompactExample} (c) and (d). 
This technique can be generalized to arbitrary formulas, but we might need $|y-z|\le \varepsilon \frac{1}{(1+x^2)^d}$, where $d$ depends on the degree of the original formula.

\paragraph{Summary}
Given a formula $\Phi$, we will first transform it to have a single closed constraint.
Then we introduce bounded limits to all quantified variables.
Thereafter we ``gently'' smooth the formula, which thickens the solution set just the right amount.
This enables us to move all the limits to the front of the formula.
Using tools from real algebraic geometry, we can then find suitable upper bounds to all of those limits.

%%%%%%%%%%%%%%%%%%%%%%%%%%%%%%%%%%%%%%%%%%%%%%%%%%%%
\newpage
\section{Formula Compactification}
\label{sec:RangeBounds}
%%%%%%%%%%%%%%%%%%%%%%%%%%%%%%%%%%%%%%%%%%%%%%%%%%%%
This section is dedicated to the proof of the following theorem.

\FOTRINVTHM*

%%%%%%%%%%%%%%%%%%%%%%%%%%%%%%%%%%%%%%
\subsection{Definitions and basic tools}
%%%%%%%%%%%%%%%%%%%%%%%%%%%%%%%%%%%%%%

\paragraph{\Semialgebraic Sets.}
We say that a set $S\subseteq \mathbb{R}^n$ is a \semialgebraic set if it is the set of values of real variables $x_1, \dots, x_n$ that satisfies $P(g_1, \dots, g_k)$, where the $g_i$ are polynomials in $x_1, \dots, x_n$ and $P$ is a (quantifier-free) Boolean formula involving literals of the form $g_i>0$, $g_i=0$, and $g_i<0$. We say that a \semialgebraic set has degree $d$ if all the polynomials $g_i$ have degree at most $d$. We allow the coefficients of the $g_i$ to be arbitrary real numbers. This gives a somewhat larger class of \semialgebraic sets compared to some other definitions that restrict to rational or integer coefficients. 
In contrast, the polynomials in an instance of \fotr have integer coefficients, but our analysis will sometimes need to consider \semialgebraic sets that have more complicated coefficients. Let $S\subseteq \mathbb{R}^n\times \mathbb{R}^m$. For $x\in \mathbb{R}^n$, we can define the restriction $S|_x=S\cap \left(\{x\}\times \mathbb{R}^m\right)$. With our definition, $S|_x$ is a \semialgebraic set for any value of $x\in \mathbb{R}^n$. We will sometimes need to argue that some property of $S|_x$ holds for all values of $x$, motivating us to allow \semialgebraic sets to have arbitrary real coefficients. 

Let $S\subseteq \mathbb{R}^{n+k}$ be a \semialgebraic set that depends on variables $x_1, \dots, x_n$ and $y_1, \dots, y_k$. There are three operations that we use to obtain a new set in $\mathbb{R}^n$:

\begin{itemize}
    \item We can \emph{restrict} to the hyperplane where the variables $y_1, \dots, y_k$ take given constant values
    \item We can \emph{project} onto the $x_i$ variables, giving a new set $\{x\in \mathbb{R}^n : \exists y\in \mathbb{R}^k: (x, y)\in S\}$
    \item We can \emph{coproject} onto the $x_i$ variables, giving a new set $\{x\in \mathbb{R}^n : \forall y\in \mathbb{R}^k: (x, y)\in S\}$
\end{itemize}
See \Cref{fig:set-creation} for an illustration of the three different operations to create new sets.
\begin{figure}
    \centering
    \includegraphics{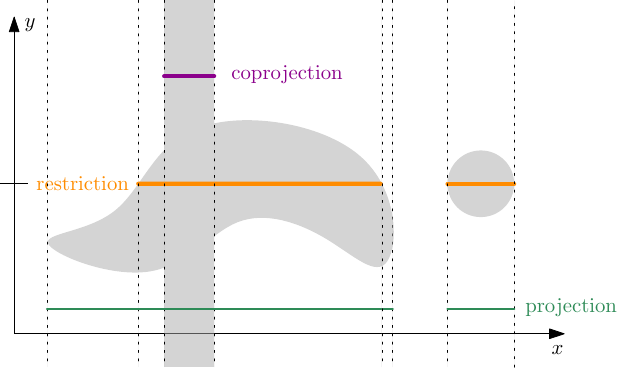}
    \caption{The coprojection of $S$ always lies in any of its restrictions, which always lies in the projection of~$S$. 
    }
    \label{fig:set-creation}
\end{figure}
\paragraph{Projections and Co-Projections.}
It is clear that restriction produces a new \semialgebraic set with the same degree.
Projections and coprojections of \semialgebraic sets are also \semialgebraic, by Tarski~\cite{T51}, but the degree might increase. The following lemma gives an upper bound on the degree of the new sets.

\begin{lemma}\label{lem:projections}(Basu, Pollack, Roy \cite{BasuPollackRoy2006}, Theorem 14.16)
There is some integer $\alpha>0$ such that the following holds: Let $S\subseteq \mathbb{R}^n\times \mathbb{R}^m$ be a \semialgebraic set of degree at most $d$. If $Q\subseteq \mathbb{R}^m$ is a set obtained from $S$ by a sequence of $n$ projections, then $Q$ is a \semialgebraic set of degree at most $d^{\alpha n}$.
\end{lemma}

By taking negations, the result of \Cref{lem:projections} also holds for a sequence of $n$ coprojections. However, if both projections and coprojections are used, then the degree can be larger.

\begin{corollary}\label{lem:(co)projections}
Let $S$ be a \semialgebraic set of degree at most $d$ and let $Q$ be a \semialgebraic set obtained from $S$ by sequences of $\le n$ projections of $\le n$ coprojections, where there are $j$ total blocks. 
Then $Q$ is a \semialgebraic set of degree at most $d^{(\alpha n)^j}$, where $j$ is the number of blocks.
\end{corollary}

By \Cref{lem:(co)projections}, a formula $\Phi$ with a free variable-vector $x\in \mathbb{R}^n$ can be viewed as a \semialgebraic set $S = \{x\in \R^n : \Phi(x)\} \subseteq \mathbb{R}^n$. 
Where convenient, we often write $x\in S$ instead of $\Phi(x)$.

\paragraph{Basic Properties of Limits.}
Another way to transform algebraic sets is to take a limit. Given a formula $\Phi$ depending on $x\in \mathbb{R}^n$ and $y\in \mathbb{R}$, we write $\lim_{y\rightarrow \infty} \Phi(x, y)$ as shorthand for the formula:

\begin{equation}\label{eq:lim1}
\exists b\forall y>b: \Phi(x, y)
\end{equation}

We observe that \Cref{eq:lim1} is equivalent to:

\begin{equation}\label{eq:lim2}
\forall a\exists y>a:\Phi(x, y)
\end{equation}

This is because, for fixed $x$, the truth value of $\Phi(x, y)$ only changes finitely many times as $y\rightarrow \infty$. In particular, limits commute with negations, so $\neg \lim_{y\rightarrow \infty}\Phi(x, y)\iff \lim_{y\rightarrow \infty}\neg \Phi(x, y)$.
We say a value of $y^*$ \textit{realizes} the limit if,
$\lim_{y\rightarrow \infty}\Phi(x, y) = \{x \in \R^n : \Phi(x,y^*) \}$.

We also sometimes take a limit going to $0$ instead of $\infty$. 
The limit:

\[\lim_{\varepsilon\rightarrow 0^+}\Phi(x, \varepsilon)\]

\noindent is equivalent to:

\[\exists b: \forall \varepsilon\in (0, b):\Phi(x, \varepsilon)\]

\noindent or equivalently:

\[\forall a\exists \varepsilon \in (0, b):\Phi(x, \varepsilon)\]

A limit $\varepsilon\rightarrow 0^+$ of $\Phi(\varepsilon)$ can be viewed as a limit $y\rightarrow \infty$ of $\Phi(x, \frac{1}y)$, so most of our results on limits work for limits going either to $0^+$ or to $\infty$.

All the limits that we will need to work with will be monotone. Let $S$ be a \semialgebraic set depending on variables $x\in \mathbb{R}$ and $y\in \mathbb{R}^k$. We say that $S$ is $x+$-monotone (resp. $x-$-monotone) if, for all $x_1<x_2$ (resp. $x_1>x_2$) and $y\in \mathbb{R}^k$, $(x_1, y)\in S$ implies $(x_2, y)\in S$. This is illustrated in \Cref{fig:monotone}. 

\begin{figure}
    \centering
    \includegraphics{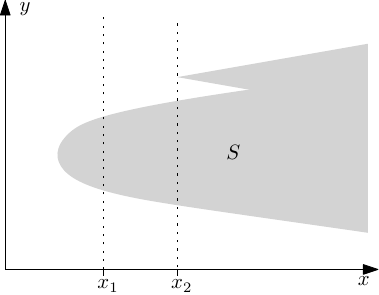}
    \caption{The set $S$ is $x+$-monotone.}
    \label{fig:monotone}
\end{figure}

We say that a limit is positive if the condition becomes easier to satisfy as the variable approaches the limit. For example, if $S$ is an $a+$-monotone \semialgebraic set, then the limit:

\begin{equation}\label{eq:monotoneexample}
    \lim_{a\rightarrow \infty}\left(a\in S\right)
\end{equation}
is positive. 
A limit is negative if the condition gets harder to satisfy as the variable approaches the limit, for example, if $S$ is $a-$-monotone instead of $a+$-monotone, then \eqref{eq:monotoneexample} is negative.

In the monotone case, the definition of a limit can be simplified. In particular, a positive limit can be replaced by an existential quantifier and a negative limit can be replaced by a universal quantifier:

\begin{lemma}\label{lem:monotonelimits}
If the limit:

\begin{equation}\label{eq:monotonelimits}
\lim_{a\rightarrow \infty}(a\in S)
\end{equation}

\noindent is positive, then \eqref{eq:monotonelimits} is equivalent to:

\[\exists a:a\in S\]

Similarly, if the limit is negative, then \eqref{eq:monotonelimits} is equivalent to:

\[\forall a:a\in S\]
\end{lemma}

\begin{proof}
This follows immediately from the definition of monotonicity.
\end{proof}

\paragraph{Limit Realization.}
Given a formula $\lim_{y\rightarrow \infty}\Phi(y)$ with no free variables, there is some constant value of $y$ that realizes the limit. Assuming that the coefficients of the polynomials defining $y$ are integers with a known bit-length, we can actually calculate an upper bound for the value needed:

\begin{lemma}[Limit Realization]
\label{lem:limitevaluation}
There is some integer $\beta>0$
such that the following holds: Let $\Phi$ be a formula in the first-order theory of the reals of with
\begin{itemize}
    \item $n$ quantifiers,
    \item $j$ quantifier alternations,
    \item degree $d$,
    \item integer coefficient with bit sizes at most $\tau$, and  
    \item $k$ free variables $y_1, \dots, y_k$.
\end{itemize}
    Then, for any $r\ge 2k$, the limit:

\[\lim_{y_1\rightarrow \infty}\dots\lim_{y_k\rightarrow \infty}\Phi(y_1, \dots, y_k)\]

\noindent is realized by:

\[y_i=\exp_2\left(\tau d^{i\left(\beta (n+r)\right)^{j+r}+i-1}\right)\]
\end{lemma}

Note that we use the $\exp_2(x)$ to mean $2^{x}$ in order to make the formula more readable. 

\begin{proof}
This is by induction on $k$. When $k=0$, there are no limits, so the result is clear trivially. 

For the purpose, of illustrating the proof, we also explicitly handle the case $k=1$.
Thus, we want to find $y_1$ realizing the formula 
\[\lim_{y\rightarrow \infty}\Phi(y) \equiv  \exists b \forall y>b : \Phi(y). \]
We apply Theorem 14.16 from \cite{BasuPollackRoy2006} to $\Phi$ and get a new quantifier free formula $\Psi$ that depends on the signs of some polynomials $p_1,\ldots,p_m$ with:

\begin{itemize}
    \item degree at most $q\leq d^{\gamma\left(\alpha (n+r)\right)^{j+r}}$, and 
    \item bit length at most $\omega  = \tau d^{\left(\gamma (n+r) \right)^{j+r}}$.    
\end{itemize}

In order to find a value of $y$ that realizes the limit, we just need to choose it so that it is larger than any root of one of the $p_i$s (except for $p_i$s that are identically zero). 
A polynomial of degree $q$ with integers coefficients of magnitude at most $2^\omega$ has all roots in $[-q2^{\omega}, q2^{\omega}]$, since the leading term dominates outside of this range. 
So it is sufficient to choose 
\[y = q2^{\omega} \leq d^{\gamma\left(\alpha (n+r)\right)^{j+r}} 2^{\tau d^{\left(\gamma (n+r) \right)^{j+r}}}. \]

We now use this idea inductively to handle the case for general $k$. By the definition of a limit, the formula:

\begin{equation}\label{eq:limitevaluationinduction}
\lim_{y_2\rightarrow \infty}\dots\lim_{y_k\rightarrow \infty}\Phi(y_1, \dots, y_k)
\end{equation}

\noindent is equivalent to a formula in the first-order theory of the reals with $j+2(k-1)$ quantifier alternations, $n+2(k-1)$ quantifier variables ($n$ variables from $\Phi$ and $2$ for each of the limits), and $1$ free variable ($y_1$). 
By the quantifier elimination from Basu, Pollack, and Roy (\cite{BasuPollackRoy2006}, Theorem 14.16), \eqref{eq:limitevaluationinduction} is equivalent to a quantifier-free formula depending on the signs of polynomials $p_1, \dots, p_m$ of degree at most $d^{\gamma\left(\alpha (n+r)\right)^{j+r}}$, and the coefficients of the polynomials $p_i$ are integers each with bit length at most 
$\omega = \tau d^{\left(\gamma (n+r) \right)^{j+r}}$. 
Here, $\gamma$ is the universal constant from the theorem in \cite{BasuPollackRoy2006}.

In order to find a value of $y_1$ that realizes the limit, we just need to choose it so that it is larger than any root of one of the $p_i$s (except for $p_i$s that are identically zero). A polynomial of degree $q$ with integer coefficients of magnitude at most $2^\omega$ has all roots in $[-q2^{\omega}, q2^{\omega}]$, since the leading term dominates outside of this range. So it is sufficient to choose:

\[y_1=\exp_2\left({\tau d^{\left(\beta (n+r)\right)^{j+r}}}\right)\]

\noindent for the constant $\beta = \alpha+\gamma$. Note that $\beta$ is 
independent of $i, j, k, r, n$ and $\tau$.

After substituting this value of $y_1$, \eqref{eq:limitevaluationinduction} becomes a formula with $k-1$ limiting variables and bit-sizes of length at most $\tau d^{\left(\beta (n+2k)\right)^{j+2k}+1}$. Since $r\ge 2k$, it must also be at least $2(k-1)$. So by inductive assumption, the limit is realized by:

\[y_i=\exp\left(\tau d^{\left(\beta (n+r)\right)^{j+r}+1}d^{(i-1)\left(\beta (n+r)\right)^{j+r}+i-2}\right)=\exp\left(\tau d^{i\left(\beta (n+r)\right)^{j+r}+i-1}\right)\]

\noindent as required.
\end{proof}

\paragraph{Closed Constraints.}
In order to compactify an instance of \fotr, we need compactly bounded quantifiers and closed constraints. 
The following lemma below from the literature shows how to achieve that all constraints are closed.

\begin{lemma}(\cite{SS17}, Lemma 3.2)\label{lem:closedconstraints}
Let $\Phi$ be a quantifier-free formula depending on $n$ variables. Then there is a polynomial $f$ of degree at most $4$ and integer $k$ such that $\Phi(x)$ is equivalent to $\exists y\in \mathbb{R}^k:f(x, y)=0$. Furthermore, if $\Phi$ has integer coefficients, then we can compute $f$ in time polynomial in the length of $\Phi$. 
\end{lemma}

Note that we can also apply this lemma to show \QR-hardness for \quantifiedrealfeasibility, which is defined as follows.
We are given a sentence in the form 
$\forall y_1 \exists x_1 \ldots \forall y_n \exists x_n :f(y_1, x_1,\ldots, y_n,x_n)=0$,
with $f$ a polynomial of degree at most $4$.
We need to decide if the sentence is true or not.

\subsection{Commuting Limits and Formula Compactification}
In order to illustrate our main ideas, we start by working out how to compactify the formula:

\[\Phi=\forall y\in \mathbb{R}^n:\exists x\in \mathbb{R}^n:f(x, y)=0\]

\noindent where $f$ is a polynomial of degree at most $4$. 
Along the way, we will show how to transform a formula to a smoothed formula and how to move the limits all to the front of the formula.
Those are the key steps before we can apply \Cref{lem:limitevaluation}.
Working with a formula with only one quantifier alternation makes the notation easier and still captures the main ideas.

\paragraph{Bounding Ranges.}
The first step is to replace each unbounded quantifier in $\Phi$ with a bounded quantifier and a limit:

\[\lim_{a\rightarrow\infty}\forall y\in [-a, a]^n:\lim_{b\rightarrow \infty}\exists x\in [-b, b]^n:f(x, y)=0\]

We cannot use \Cref{lem:limitevaluation} here to set the value of $b$, because the condition $\exists x\in [-b, b]^n: f(x, y)=0$ depends on the variable $y$, and there is no bit-complexity bound for $y$. In order to use \Cref{lem:limitevaluation}, we need to pull all the limits out to the front.

\paragraph{Condition Smoothing.}
In order to swap the $b\rightarrow \infty$ limit and the $\forall y$ quantifier, we need to smooth the condition. 
We replace the condition $f(x, y)=0$ with:

\[\lim_{\varepsilon\rightarrow 0^+}\exists z:f(x, z)=0\wedge |y-z|\le \varepsilon\]

\paragraph{Swapping a closed limit with an existential quantifier.}
Next, we need to exchange the order of the $\varepsilon\rightarrow 0^+$ limit and the $\exists x$ quantifier.

\begin{lemma}\label{lem:closedswap}
Let $S \subset \R^2$ be a \semialgebraic set depending on $a$ and $x$ that is $a-$-monotone such that, for each fixed $a$, the restriction $S|_{a}$ of $S$ to that value of $a$ is closed.
Then for $I$ a compact interval:

\begin{equation}\label{eq:closedswap1}
\exists x\in I:\lim_{a\rightarrow \infty}:(a, x)\in S
\end{equation}

is equivalent to:

\begin{equation}\label{eq:closedswap2}
\lim_{a\rightarrow \infty}\exists x\in I:(a, x)\in S
\end{equation}
\end{lemma}

\begin{figure}[tbph]
    \centering
    \includegraphics{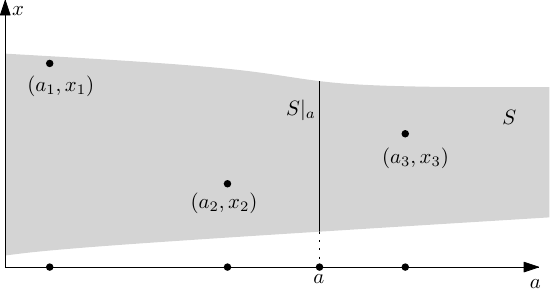}
    \caption{The setup of \Cref{lem:closedswap}.}
    \label{fig:closedswap}
\end{figure}

\begin{proof}
If \eqref{eq:closedswap1} holds, then \eqref{eq:closedswap2} holds straightforwardly. So suppose that \eqref{eq:closedswap2} holds, and choose a sequence of values of $a_i$ approaching $\infty$. 
We refer to \Cref{fig:closedswap} for an illustration.
For each $a_i$ large enough, there is some $x_i$ such that $(a_i, x_i)\in S$. 
By compactness of $I$, we can suppose that the $x_i$ go to some limit $x^*\in I$. 
(We use the fact that every sequence has a converging subsequence on a compact set.)
By $a-$-monotonicity of $S$, 
for each $a$, $(a, x_i)\in S$ for all $i$ such that $a_i\ge a$. 
Since $S|_a$ is closed, it holds that $(a, x^*)\in S$. 
Here, we used that the limit of every converging sequence in a closed set is contained in that set.
This holds for all $a$, so \eqref{eq:closedswap1} holds. So \eqref{eq:closedswap1} and \eqref{eq:closedswap2} are equivalent.
\end{proof}

We view the limit $\varepsilon\rightarrow 0^+$ as a limit $\frac1\varepsilon\rightarrow \infty$. 
The condition $f(x, z)=0\wedge |y-z|\le \varepsilon$ is $\varepsilon+$-monotone, and therefore it is (essentially) $\frac1\varepsilon-$-monotone. 
The closedness property comes from the fact that closed constraints and no negations are used.
So we can apply \Cref{lem:closedswap} to see that the formula:

\[\lim_{b\rightarrow \infty}\exists x\in [-b, b]^n:\lim_{\varepsilon\rightarrow 0^+}\exists z:f(x, z)=0\wedge |y-z|\le \varepsilon\]

\noindent is equivalent to:

\begin{equation}\label{eq:E1}
\lim_{b\rightarrow \infty}\lim_{\varepsilon\rightarrow 0^+}\exists x\in [-b, b]^n:\exists z:f(x, z)=0\wedge |y-z|\le \varepsilon
\end{equation}

\paragraph{Exchanging Limits with Degree Bounds I.}
The next step is to commute the $b$ and $\varepsilon$ limits. This is possible by the following lemma:

\begin{lemma}\label{lem:bswap}
Let $S$ be a \semialgebraic set depending on $a$ and $x\in \mathbb{R}^n$ that is $a-$-monotone. If $Q=\{(a, b)\in \mathbb{R}^2 : \exists x\in [-b, b]^n:(a, x)\in S\}$ is a \semialgebraic set of degree $d$, then:

\begin{equation}\label{eq:bswap1}
\lim_{a\rightarrow \infty}\lim_{b\rightarrow\infty}\exists x\in [-b, b]^n:\left(a\left(1+|x|^2\right)^{d}, x\right)\in S
\end{equation}

\noindent implies:

\begin{equation}\label{eq:bswap2}
\lim_{b\rightarrow\infty}\lim_{a\rightarrow \infty}\exists x\in [-b, b]^n:(a, x)\in S
\end{equation}
\end{lemma}

\begin{proof}
Suppose that \eqref{eq:bswap2} is false, so for every $b$ there is some $a$ so that $\exists x\in [-b, b]^n:(a, x)\in S$ is false. Let $Q$ be the set of values of $a$ and $b$ such that $\exists x\in [-b, b]^n:(a, x)\in S$. By assumption, $Q$ has degree $d$. So membership in $Q$ is determined by the signs of polynomials $p_1, \dots, p_\ell$ each of degree at most $d$ in $a$ and $b$. 

Consider some $p_i$. The value of $\lim_{b\rightarrow \infty}\lim_{a\rightarrow \infty} \text{sign}(p_i)$ is the sign of the highest-order term lexicographically in $a$ then $b$. If $p_i$ is non-zero, then write $p_i(a, b)=\gamma a^\alpha(b^\beta+g(b))+h(a, b)$ where $\gamma$ is non-zero, $g$ has degree less than $\beta$ and $h$ has degree less than $\alpha$ in $a$. There is some value $B$ such that $b^\beta+g(b)\ge 1$ for $b\ge B$. There is some $A$ such that $|h(a, b)|< Aa^{\alpha-1}b^{d}$. So the sign of $p_i$ is the same as the sign of $\gamma$ when $b\ge B$ and $a\ge |\gamma|^{-1}Ab^{d}$. Taking the maximum of $|\gamma|^{-1} A$ and $B$ over all the $p_i$, we conclude that there are values $A$ and $B$ such that $(a, b)\notin Q$ for $b\ge B$ and $a\ge Ab^{d}$. 

Suppose that $\left(a\left(1+|x|^2\right)^{d}, x\right)\in S$ for some $x\in \mathbb{R}^n$ and $a\ge A$. Let $b=\max_i\{|x_i|\}$, so $\left(a\left(1+|x|^2\right)^{d}, b\right)\in Q$. Since $S$ is $a-$-monotone, $\left(a\left(1+b^2\right)^{d}, b\right)\in Q$. We conclude that $b<B$, so $x\in [-b, b]^n$.

If \eqref{eq:bswap1} is true, then for every $a$ there is some $x(a)$ such that $\left(a\left(1+|x(a)|^2\right)^{d}, x(a)\right)\in S$. By $a-$-monotonicity, $(a, x(a))\in S$. By the above, $x(a)\in [-B, B]^n$ for all $a\ge A$. In particular, \eqref{eq:bswap2} is true, contradicting the earlier assumption that \eqref{eq:bswap1} is false. So it cannot be the case that \eqref{eq:bswap2} is false and \eqref{eq:bswap1} is true, proving the claim.
\end{proof}

Let $\alpha$ be the constant from \Cref{lem:projections}. Viewing the limit $\varepsilon\rightarrow 0^+$ as a limit $\frac{1}\varepsilon\rightarrow \infty$, \Cref{lem:projections,lem:bswap} show that:

\begin{equation}\label{eq:E2}
\lim_{\varepsilon\rightarrow 0^+}\lim_{b\rightarrow \infty}\exists x\in [-b, b]^n,z:f(x, z)=0\wedge |y-z|\le \varepsilon\left(1+|x|^2\right)^{-4^{\alpha n}}
\end{equation}

\noindent implies \eqref{eq:E1}. To get the reverse implication, we observe that we can substitute $\varepsilon\rightarrow (1+|x_1|^2)^{-4^{\alpha n}}\varepsilon$ before swapping the $\exists x_1$ and $\lim_\varepsilon$ terms. That is to say:

\[\lim_{b\rightarrow \infty}\exists x\in [-b, b]^n:\lim_{\varepsilon\rightarrow 0^+}\exists z:f(x, z)=0\wedge |y-z|\le \varepsilon\]

\noindent is equivalent to:

\begin{equation}\label{eq:Einternal}
\lim_{b\rightarrow \infty}\exists x\in [-b, b]^n:\lim_{\varepsilon\rightarrow 0^+}\exists z:f(x, z)=0\wedge |y-z|\le \varepsilon(1+|x|^2)^{-4^{\alpha n}}
\end{equation}

This is because $(1+|x|^2)^{-4^{\alpha n}}$ is just a positive constant at this stage, so the substitution $\varepsilon\rightarrow (1+|x_1|^2)^{-4^{\alpha n}}\varepsilon$ doesn't change the limit. It is clear that \eqref{eq:Einternal} implies \eqref{eq:E2}, since allowing the values of $b$ and $x$ to depend on $\varepsilon$ only makes the condition easier to satisfy.

\paragraph{Exchanging Monotone Limits with Quantifiers.}

\begin{lemma}\label{lem:easyswap}
Positive limits commute with existential quantifiers and negative limits commute with universal quantifiers.
\end{lemma}

\begin{proof}
Straightforward by \Cref{lem:monotonelimits}. 
\end{proof}

The proof of \Cref{lem:easyswap} also shows that any two positive limits commute with each other and that any two negative limits commute with each other.

As we have shown, $\Phi$ is equivalent to:

\begin{equation}\label{eq:AE0}
\lim_{a\rightarrow \infty}\forall y\in [-a, a]^n:\lim_{\varepsilon\rightarrow 0^+}\lim_{b\rightarrow \infty}\exists x\in [-b, b]^n,z:f(x, z)=0\wedge |y-z|\le \varepsilon\left(1+|x|^2\right)^{-4^{\alpha n}}
\end{equation}

The limit $\lim_{\varepsilon\rightarrow 0^+}$ is negative, so by \eqref{lem:easyswap}, \eqref{eq:AE0} is equivalent to:

\begin{equation}\label{eq:AE1}
\lim_{a\rightarrow \infty}\lim_{\varepsilon\rightarrow 0^+}\forall y\in [-a, a]^n:\lim_{b\rightarrow \infty}\exists x\in [-b, b]^n,z:f(x, z)=0\wedge |y-z|\le \varepsilon\left(1+|x|^2\right)^{-4^{\alpha n}}
\end{equation}

\paragraph{Swapping an $\varepsilon$-open Limit with a Universal Quantifier.}
We say that a \semialgebraic set $S$ is $\varepsilon$-open if it is $\varepsilon+$-monotone and, for all sufficiently small $\delta>0$, the union:

\[\bigcup_{0\le \varepsilon<\delta}S|_\varepsilon\]

\noindent is an open subset of $\mathbb{R}^k$.

Finally, we are ready to exchange the order of the $\forall y$ quantifier and the $b\rightarrow \infty$ limit.

\begin{lemma}\label{lem:openswap}
Let $S$ be a \semialgebraic set depending on $b,\varepsilon$ and $y\in \mathbb{R}^n$ that is $b+$-monotone, where $S|_b$ is $\varepsilon$-open for each fixed $b$. Then for $K\subseteq \mathbb{R}^n$:

\begin{equation}\label{eq:openswap1}
\lim_{\varepsilon\rightarrow 0^+}\forall y\in K:\lim_{b\rightarrow \infty}(y, b, \varepsilon)\in S
\end{equation}

\noindent is equivalent to:

\begin{equation}\label{eq:openswap2}
\lim_{\varepsilon\rightarrow 0^+}\lim_{b\rightarrow \infty}\forall y\in K:(y, b, \varepsilon)\in S
\end{equation}
\end{lemma}

\begin{proof}
If \eqref{eq:openswap2} holds, then \eqref{eq:openswap1} holds straightforwardly. So suppose that \eqref{eq:openswap1} holds. Let $\delta$ be such that $\forall y\lim_{b\rightarrow \infty}:(y, b, \varepsilon)\in S$ for all $\varepsilon<\delta$. For each $y\in K$, there is some $B_y$ such $(y, b, \frac12\delta)\in S$ for every $b\ge B_y$. By $\varepsilon$-openness, the set:

\[\bigcup_{\varepsilon<\delta}S|_{\varepsilon=\varepsilon, b=B_y}\]

\noindent is open. Since this set contains $y$, there is an open neighborhood $U$ of $y$ such that, for $w\in U$, there is some $\varepsilon<\delta$ such that $(w, b_y, \varepsilon)\in S$. By compactness of $K$, $K$ is covered by a finite number of these neighborhoods. For every $b$ at least the max of $b_y$ over this finite cover, we see by $\varepsilon+$- and $b+$- monotonicity that $(y, b, \delta)\in S$. This works for all sufficiently small $\delta$, proving that \eqref{eq:openswap2} holds. So \eqref{eq:openswap1} and \eqref{eq:openswap2} are equivalent.
\end{proof}

In order to use this lemma, we need to show that, for each fixed $b$, the set of values of $(y, \varepsilon)$ such that $\exists x\in [-b, b]^n,z:f(x, z)=0\wedge |y-z|\le \varepsilon\left(1+|x|^2\right)^{-4^{\alpha n}}$ is $\varepsilon$-open. The set:

\[Q(b)=\bigcup_{0\le \varepsilon<\delta}\left\{y:\exists x\in [-b, b]^n,z:f(x, z)=0\wedge |y-z|\le \varepsilon\left(1+|x|^2\right)^{-4^{\alpha n}}\right\}\]

\noindent is equivalent to:

\[\left\{y:\exists x\in [-b, b]^n,z:f(x, z)=0\wedge |y-z|<\delta\left(1+|x|^2\right)^{-4^{\alpha n}}\right\}\]

For each fixed $x$, the set $Q(b, x)=\{y:\exists z:f(x, z)=0\wedge |y-z|<\delta\left(1+|x|^2\right)^{-4^{\alpha n}}\}$ is open. The set $Q(b)$ is the union of $Q(b, x)$ for $x\in [-b, b]^n$. Since the union of open sets is open, the set $Q(b)$ is open. So we can use \Cref{lem:openswap} to see that \eqref{eq:AE1} (and so $\Phi$) is equivalent to:

\[\lim_{a\rightarrow \infty}\lim_{\varepsilon\rightarrow 0^+}\lim_{b\rightarrow \infty}\forall y\in [-a, a]^n:\exists x\in [-b, b]^n,z:f(x, z)=0\wedge |y-z|\le \varepsilon\left(1+|x|^2\right)^{-4^{\alpha n}}\]

\paragraph{Application.}
This formula is now in the appropriate form to use \Cref{lem:limitevaluation}. 
To be specific let \[\Psi(a,b,\varepsilon) = \forall y\in [-a, a]^n:\exists x\in [-b, b]^n,z:f(x, z)=0\wedge |y-z|\le \varepsilon\left(1+|x|^2\right)^{-4^{\alpha n}}\]
We note that $\Psi$ has 
\begin{itemize}[noitemsep, topsep=0pt, parsep=0pt, partopsep=0pt]
    \item $2n+1$ quantifiers,
    \item $1$ quantifier alternations,
    \item degree $d \approx 4^{\alpha n}$,
    \item bit-size $\tau$, and
     \item $3$ free variables $a,b$, and $\varepsilon$.
\end{itemize}
Thus \Cref{lem:limitevaluation} gives us as concrete (and very large) integer values 
$a^*,b^*, 1/\varepsilon^*$ such that 
\[\lim_{a\rightarrow \infty}\lim_{\varepsilon\rightarrow 0^+}\lim_{b\rightarrow \infty} \Psi(a,b,\varepsilon)\]
is equivalent to 
\[ \Gamma = \forall y\in [-a^*, a^*]^n:\exists x\in [-b^*, b^*]^n,z:f(x, z)=0\wedge |y-z|\le \varepsilon^*\left(1+|x|^2\right)^{-4^{\alpha n}}.\]
And this formula is in turn equivalent to our original formula 
\[\Phi=\forall y\in \mathbb{R}^n:\exists x\in \mathbb{R}^n:f(x, y)=0\] from the beginning of this section.
The values are
\[a^* = \exp_2\left(\tau d^{\left(\beta (2n+7)\right)^{7}}\right),\]
\[1/\varepsilon^* = \exp_2\left(\tau d^{2\left(\beta (2n+7)\right)^{7}+1}\right),\]
and 
\[b^* = \exp_2\left(\tau d^{3\left(\beta (2n+7)\right)^{7}+2}\right),\]
with $\beta$ being the constant from \Cref{lem:limitevaluation}.

Note that $\Gamma$ has compactly bounded domains and only closed constraints as promised.
In the next, section we will show in full generality how to do a similar transformation for an unbounded number of quantifier alternations.

%%%%%%%%%%%%%%%%%%%%%%%%%%%%%%%%%%%%%%%%%%%%%
\subsection{The Compactification theorem}
%%%%%%%%%%%%%%%%%%%%%%%%%%%%%%%%%%%%%%%%%%%%%

We are almost ready to generalize the approach from the previous section to formulas with arbitrarily many quantifier alternations, but there are a few more lemmas that we will need. In order to use \Cref{lem:closedswap,lem:openswap}, we will need a way to verify that certain sets are open or closed.

\paragraph{Bounded Projections and Coprojections.}

\begin{lemma}\label{lem:boundedopenprojections}
Let $S\subseteq \mathbb{R}^m\times \mathbb{R}^n$ and let $\varphi:\mathbb{R}^n\rightarrow \mathbb{R}$ be a continuous function. Let $Q$ be a quantifier, either $\exists$ or $\forall$. Let $R=\{y\in \mathbb{R}^n|Q x\in [-\varphi(y), \varphi(y)]^m:(x, y)\in S\}$. If $S$ is open, then $R$ is open. If $S$ is closed, then $R$ is closed. 
\end{lemma}

\begin{proof}
We take $S$ to be open, since the case for $S$ closed follows by taking complements. First suppose that $Q$ is $\exists$, and let $z$ be such that $\exists x\in [-\varphi(z), \varphi(z)]^m:(x, z)\in S$. Since $S$ is open, we can find such an $x$ where $x\in (-\varphi(z), \varphi(z))^m$. Also, there is an open neighborhood $V$ of $z$ such that $(x, y)\in S$ for $y\in V$. By continuity of $\varphi$, there is an open neighborhood $W$ of $z$ such that $x\in [-\varphi(y), \varphi(y)]$ for $y\in V$. The intersection of $V$ and $W$ gives an open neighborhood of $z$ in $\{y\in \mathbb{R}^n|Q x\in [-\varphi(y), \varphi(y)]^m:(x, y)\in S\}$.

Now suppose that $Q$ is $\forall$ and let $z$ be such that $\forall x\in [-\varphi(z), \varphi(z)]^m:(x, z)\in S$. Since $S$ is open, for each $x\in K$, there are open neighborhoods $U_x$ of $x$ and $V_x$ of $z$ such that $U_x\times V_x\subseteq S$. 

By compactness of $[-\varphi(z), \varphi(z)]^m$, there is a finite set $\mathcal{T}\subseteq [-\varphi(z), \varphi(z)]^m$ such that the union $R=\bigcup_{x\in \mathcal{T}}U_x$ covers $[-\varphi(z), \varphi(z)]^m$. Since the boundary of $[-\varphi(z), \varphi(z)]^m$ is compact, $R$ contains the larger set $(-\varphi(z)-\gamma, \varphi(z)+\gamma)^m$ for some $\gamma>0$. Since $\varphi$ is continuous, there is some open neighborhood $W\subseteq \mathbb{R}^n$ of $z$ such that $[-\varphi(y), \varphi(y)]\subseteq (-\varphi(z)-\gamma, \varphi(z)+\gamma)$ for $y\in W$. So the intersection of $W$ and the $V_x$ for $x\in \mathcal{T}$ gives an open neighborhood of $z$ in $R$, proving the claim.
\end{proof}

Next, we need a reversed version of \Cref{lem:bswap}.

\paragraph{Exchanging Limits with Degree Bounds II.}

\begin{lemma}\label{lem:aswap}
Let $S$ be a \semialgebraic set depending on $b$ and $y\in \mathbb{R}^n$ that is $b+$-monotone. 
If $Q=\{(a, b)\in \mathbb{R}^2 : \forall y\in [-a, a]^n:(b, y)\in S\}$ 
is a \semialgebraic set of degree $d$, then:

\begin{equation}\label{eq:aswap1}
\lim_{a\rightarrow \infty}\lim_{b\rightarrow\infty}\forall x\in [-a, a]^n:\left(b, y\right)\in S
\end{equation}

\noindent implies:

\begin{equation}\label{eq:aswap2}
\lim_{b\rightarrow\infty}\lim_{a\rightarrow \infty}\forall x\in [-a, a]^n:\left(b\left(1+|y|^2\right)^d, y\right)\in S
\end{equation}
\end{lemma}

\begin{proof}
Follows from \Cref{lem:bswap} by taking negations.
\end{proof}

\paragraph{Merging Monotone Limits.}
Finally, we show how to merge limits. This is useful for reducing the number of limits that we need to keep track of.

\begin{lemma}\label{lem:monotonemerge}
Let $S$ be a \semialgebraic set depending on $a_1$ and $a_2$ that is $a_1-$- and $a_2-$-monotone. Then:

\begin{equation}\label{eq:backwardsmerge1}
\lim_{a_1\rightarrow \infty}\lim_{a_2\rightarrow \infty} (a_1, a_2)\in S
\end{equation}

\noindent is equivalent to:

\begin{equation}\label{eq:backwardsmerge2}
\lim_{a\rightarrow \infty}(a, a)\in S
\end{equation}

The same is true if $S$ is $a_1+$- and $a_2+$- monotone instead.
\end{lemma}

\begin{proof}
Suppose $S$ is $a_1-$- and $a_2-$- monotone. Suppose that \eqref{eq:backwardsmerge1} holds. So for all sufficiently large $a_1$, there is some $A(a_1)$ such that $(a_1, a_2)\in S$ for $a_2>A(a_1)$. But $S$ is $a_2-$-monotone, so $(a_1, a_1)\in S$, since $a_1<a_2$ for some $a_2>A(a_1)$. So \eqref{eq:backwardsmerge2} holds.

Now suppose that \eqref{eq:backwardsmerge2} holds. Let $A$ be such that $(a, a)\in S$ for $a>A$. Fix $a_1>A$ and $a_2>a_1$. So $(a_2, a_2)\in S$, and by $a_1-$-monotonicity, $(a_1, a_2)\in S$, proving that \eqref{eq:backwardsmerge2} holds.

The case where $S$ is $a_1+$- and $a_2+$-monotone follows by taking negations. 
\end{proof}

\paragraph{The General Transformation.}
We now define the transformation that we use. We will write $[a,b] t$ for the closed interval $[at,bt]$.
We also use this for boxes, like $[a,b]^n t = [at,bt]^n$.
We hope that this helps readability for the reader, by removing redundant clutter.

\begin{definition}[Compactification]
Let $\Phi$ be  a formula of the form $\forall y_k\exists x_k\dots\forall{y_1}\exists x_1: f(x_1, y_1, \dots, x_k, y_k)=0$, where $f$ is a polynomial of degree $4$, possibly depending on some free variables,
and each $x_i$ or $y_i$ is in $\mathbb{R}^n$, $n\ge 1$.  

For $\ell=1,\dots, k$, define:

\[p_\ell=4^{(6\alpha n)^{\ell(2\ell-1)}}\text{ and }q_\ell=4^{(6\alpha n)^{\ell(2\ell+1)}}\] 

\noindent where $\alpha$ is the constant from \Cref{lem:projections}.

We define the \emph{compactification} of $\Phi$ to be a new formula $\Psi$, which is obtained from $\Phi$ by replacing each existential quantifier $\exists x_\ell\in \mathbb{R}^n$ with:

\[\exists (x_\ell, b_\ell)\in [-b_{\ell+1}, b_{\ell+1}]^{n+1}\left(1+|(y_\ell, a_\ell)|^2\right)^{q_\ell},z_\ell\in \mathbb{R}^n:a_\ell|z_\ell-y_\ell|^2\le \left(1+|(x_{\ell}, b_\ell)|^2\right)^{-p_{\ell}}\wedge |z_\ell-y_\ell|^2\le 1\]

\noindent replacing each universal quantifier $\forall y_\ell\in \mathbb{R}^n$ for $\ell<k$ with:

\[\forall (y_\ell, a_\ell)\in [-a_{\ell+1}, a_{\ell+1}]^{n+1}\left(1+|(x_{\ell+1}, b_{\ell+1})|^2\right)^{p_{\ell+1}}\]

\noindent replacing the universal quantifier $\forall y_k\in \mathbb{R}^n$ with:

\[\forall (y_k, a_k)\in [-a_{k+1}, a_{k+1}]^{n+1}\]

\noindent replacing the condition $f(x_1, y_1, \dots, x_k, y_k)=0$ with:

\[f(x_1, z_1, \dots, x_k, z_k)=0 \]

\noindent and adding limits:

\[\lim_{b_{k+1}\rightarrow \infty}\lim_{a_{k+1}\rightarrow\infty}\]

\noindent to the front.

If instead $\Phi$ is a formula of form $\exists x_{k+1}\forall y_{k}\dots \forall y_1\exists x_1: f(x_1, y_1, \dots, y_{k-1}, x_k)=0$, then we define the compactification of $\Phi$ by making the same replacements as before, replacing the existential quantifier $\exists x_{k+1}\in \mathbb{R}^n$ with:

\[\exists (x_{k+1}, b_{k+1})\in [-b_{k+1}, b_{k+1}]^{n+1}\]

\noindent and adding limits:

\[\lim_{b_{k+2}\rightarrow \infty}\lim_{a_{k+1}\rightarrow\infty}\]

\noindent to the front. 
\end{definition}

\paragraph{The Compactification Theorem.}
We are now ready to state the main theorem of this section.

\begin{theorem}\label{thm:compactformula}
Let $\Phi$ be a formula of the form 
$\forall y_k\exists x_k\dots\forall{y_1}\exists x_1: f(x_1, y_1, \dots, x_k, y_k)=0$, where $f$ is a polynomial of degree $4$, possibly depending on some additional free variables, and each $x_i$ or $y_i$ is in $\mathbb{R}^n$, $n\ge 1$. Then $\Phi$ is true if and only if the compactification $\Psi$ is true. 
\end{theorem}

\begin{proof}
This is by induction on $k$. The case where $k=0$ is clear, since:

\[\lim_{b_{1}\rightarrow \infty}\lim_{a_{1}\rightarrow\infty}f=0\]

\noindent is equivalent to:

\[f=0\]

We now consider the case for arbitrary $k$. Since the compactification of $\Phi$ doesn't depend the free variables of $\Phi$, we will assume that these have been fixed, since it is sufficient to prove the claim for any fixed values of the free variables. 

We apply compactification to the innermost $2(k-1)$ quantifier blocks in $\Phi$, obtaining (by inductively assumption) and equivalent formula:

\begin{equation}\label{eq:inductionbase}
\forall y_k: \exists x_k:\lim_{b_{k}\rightarrow \infty}\lim_{a_k\rightarrow \infty} \Psi_{k-1}(a_k, b_k, x_k, y_k)
\end{equation}

Where $\Psi_{k-1}$ is the compactification of the innermost $2(k-1)$ quantifiers. We now add limits to obtain:

\begin{equation}\label{eq:inductionlimits}
\lim_{a\rightarrow \infty}\forall y_k\in [-a, a]^n:\lim_{b\rightarrow \infty}\exists x_k\in [-b, b]^n:\lim_{b_{k}\rightarrow \infty}\lim_{a_k\rightarrow \infty} \Psi_{k-1}(a_k, b_k, x_k, y_k)
\end{equation}

By \Cref{lem:easyswap}, we can swap the $\exists x_k$ quantifier and the $b_k\rightarrow \infty$ limit. Every quantifier in $\Psi_{k-1}$ is bounded by a continuous function of previously-defined variables, so for each fixed $a_k, b_k$, the set $\{(x_k, y_k)| \Psi_{k-1}(a_k, b_k, x_k, y_k)\}$ is closed by repeated applications of \Cref{lem:boundedopenprojections}. So the condition $\Psi_{k-1}(a_k, b_k, x_k, y_k)$ is equivalent to:

\begin{equation}\label{eq:insidepart}
\lim_{\varepsilon\rightarrow0^+}\exists z_k\in \mathbb{R}^n:|z_k-y_k|^2\le \varepsilon\wedge\Psi_{k-1}(a_k, b_k, x_k, z_k)
\end{equation}

We want to use \Cref{lem:monotonemerge} to merge the negative limits $a_k\rightarrow \infty$ and $\varepsilon\rightarrow 0^+$. We can view the limit $\epsilon\rightarrow 0^+$ as a limit $\frac{1}{\varepsilon}\rightarrow \infty$, but we have to be careful because the condition $|z_k-y_k|^2\le \varepsilon$ is not quite $\frac{1}{\varepsilon}-$-monotone. However, the condition $\frac{1}{\varepsilon}|z_k-y_k|^2\le 1$ is $\frac{1}{\varepsilon}-$-monotone, and is equivalent when $\varepsilon>0$. So we use \Cref{lem:monotonemerge} to obtain:

\begin{equation}\label{eq:induction1}
\lim_{a\rightarrow \infty}\forall y_k\in [-a, a]^n:\lim_{b\rightarrow \infty}\lim_{b_{k}\rightarrow \infty}\exists x_k\in [-b, b]^n:\lim_{a_k\rightarrow \infty} \exists z_k\in \mathbb{R}^n:a_k|z_k-y_k|^2\le 1\wedge \Psi_{k-1}(a_k, b_k, x_k, z_k)
\end{equation}

Note that the value of $z_k$ is bounded by a function of $y_k$ and $a_k$, so \Cref{lem:boundedopenprojections} shows that set of values of $x_k$ satisfying \eqref{eq:insidepart} is closed for each $a_k, b_k, y_k$. Using \Cref{lem:easyswap}, we can exchange the order of the $\exists x_k$ quantifier and the $a_k\rightarrow \infty$ limit. So by \Cref{lem:closedswap}, \eqref{eq:induction1} is equivalent to:

\begin{equation}\label{eq:induction2}
\lim_{a\rightarrow \infty}\forall y_k\in [-a, a]^n:\lim_{b\rightarrow \infty}\lim_{b_{k}\rightarrow \infty}\lim_{a_k\rightarrow \infty}\exists x_k\in [-b, b]^n: \exists z_k\in \mathbb{R}^n:a_k|z_k-y_k|^2\le 1\wedge \Psi_{k-1}(a_k, b_k, x_k, z_k)
\end{equation}

Using \Cref{lem:monotonemerge}, we can now merge the positive limits $b\rightarrow \infty$ and $b_k\rightarrow \infty$ into one limit $b_{k+1}\rightarrow \infty$. The next step is to use \Cref{lem:bswap} to exchange the order of the $b_{k+1}\rightarrow \infty$ and $a_k\rightarrow\infty$ limits, but we aren't quite in the right setting to do this yet, since $\Psi_{k-1}$ depends on $b_k$. Since $\Psi_{k-1}(a_k, b_{k+1}, x_k, z_k)$ is $b_{k+1}+$-monotone, it is equivalent to $\exists b_k\in [-b_{k+1}, b_{k+1}]:\Psi_{k-1}(a_k, b_k, x_k, z_k)$. So \eqref{eq:induction2} becomes:

\begin{multline}\label{eq:induction3}
\lim_{a\rightarrow \infty}\forall y_k\in [-a, a]^n:\lim_{b_{k+1}\rightarrow \infty}\lim_{a_k\rightarrow \infty}\exists (x_k, b_k)\in [-b_{k+1}, b_{k+1}]^{n+1},z_k\in \mathbb{R}^n:\\a_k|z_k-y_k|^2\le 1\wedge \Psi_{k-1}(a_k, b_k, x_k, z_k)
\end{multline}

\noindent which is now in the correct form to apply \Cref{lem:bswap} to \eqref{eq:induction3}. In order to use \Cref{lem:bswap}, we need a bound for the degree of:

\begin{equation}\label{eq:bswapset}
\left\{(a_k, b_{k+1})|\exists (x_k, b_k)\in [-b_{k+1}, b_{k+1}]^{n+1},z_k\in \mathbb{R}^n:a_k|z_k-y_k|^2\le 1\wedge \Psi_{k-1}(a_k, b_k, x_k, z_k)\right\}
\end{equation}

When $k>1$, the polynomials defining \eqref{eq:bswapset} have degree at most $2q_{k-1}+1$, the highest degree ones coming from the bounds on the $\exists (x_{k-1}, b_{k-1})$ quantifier, e.g. $b_{k-1}\le b_k\left(1+|(y_{k-1}, a_{k-1})|^2\right)^{q_{k-1}}$. When $k=1$, these polynomials have degree at most $4$, the highest-degree on being $f(x_1, y_1)=0$.

\eqref{eq:bswapset} has $k-1$ blocks of universal quantifiers with $n+1$ variables ($y_\ell\in \mathbb{R}^n$ and $a_\ell\in \mathbb{R}$) each, and $k$ blocks of existential quantifiers with $2n+1$ variables ($x_\ell\in \mathbb{R}^n$, $b_\ell\in \mathbb{R}$, and $z_\ell\in \mathbb{R}^{n}$) each. So by \Cref{lem:projections}, the degree of \eqref{eq:bswapset} is at most $q_{k-1}^{(6\alpha n)^{2k-1}}$ (when $k>1$) or $4^{6\alpha}$ (when $k=1$). In either case, this is at most $p_k$. So by \Cref{lem:bswap}, \eqref{eq:induction3} is implied by:

\begin{multline}\label{eq:induction4}
\lim_{a\rightarrow \infty}\forall y_k\in [-a, a]^n:\lim_{a_k\rightarrow \infty}\lim_{b_{k+1}\rightarrow \infty}\exists (x_k, b_k)\in [-b_{k+1}, b_{k+1}]^{n+1},z_k\in \mathbb{R}^n:\\ a_k\left(1+|(x_k, b_k)|^2\right)^{p_k}|z_k-y_k|^2\le 1\wedge \Psi_{k-1}\left(a_k\left(1+|(x_k, b_k)|^2\right)^{p_k}, b_k, x_k, z_k\right)
\end{multline}

Note that \eqref{eq:induction1} is equivalent to:

\begin{multline}\label{eq:induction5}
\lim_{a\rightarrow \infty}\forall y_k\in [-a, a]^n:\lim_{b_{k+1}\rightarrow \infty}\exists (x_k, b_k)\in [-b_{k+1}, b_{k+1}]^{n+1}:\lim_{a_k\rightarrow \infty} \exists z_k\in \mathbb{R}^n:\\ a_k\left(1+|(x_k, b_k)|^2\right)^{p_k}|z_k-y_k|^2\le 1\wedge \Psi_{k-1}\left(a_k\left(1+|(x_k, b_k)|^2\right)^{p_k}, b_k, x_k, z_k\right)
\end{multline}

\noindent since $\left(1+|(x_k, b_k)|^2\right)^{p_k}$ is just a positive constant at this stage. By \Cref{lem:monotonelimits}, the positive limit $b_{k+1}\rightarrow \infty$ can be viewed as an existential quantifier and the negative limit $a_{k}\rightarrow \infty$ can be viewed as a universal quantifier. This makes it clear that \eqref{eq:induction5} implies \eqref{eq:induction4}, since allowing $b_{k+1}$ and $x_k$ to depend on $a_k$ can only make the formula easier to satisfy. So \eqref{eq:induction4} is equivalent to \eqref{eq:inductionbase}.

Starting from \eqref{eq:induction4}, we use \Cref{lem:easyswap} to exchange the $a_k\rightarrow \infty$ limit and the $\forall y_k$ quantifier. We then use \Cref{lem:monotonemerge} to regroup the negative limits at the front, obtaining:

\begin{multline}\label{eq:induction6}
\lim_{a_{k}\rightarrow \infty}\lim_{\varepsilon\rightarrow 0^+}\forall y_k\in [-a_k, a_k]^n:\lim_{b_{k+1}\rightarrow \infty}\exists (x_k, b_k)\in [-b_{k+1}, b_{k+1}]^{n+1},z_k\in \mathbb{R}^n:\\|z_k-y_k|^2\le \varepsilon\left(1+|(x_k, b_k)|^2\right)^{-p_k}\wedge \Psi_{k-1}\left(a_{k}\left(1+|(x_k, b_k)|^2\right)^{p_k}, b_k, x_k, z_k\right)
\end{multline}

We now want show $\varepsilon$-openness so that we can use \Cref{lem:openswap} to exchange the order of the $b_{k+1}\rightarrow \infty$ limit and the $\forall y_k$ quantifier. For any fixed $a_{k}, x_k, b_k$ and $\delta$, the set:

\begin{equation}\label{eq:epsilonopenset}
\bigcup_{\varepsilon<\delta}\left\{y_k|\exists z_k\in \mathbb{R}^n:|z_k-y_k|^2\le \varepsilon\left(1+|(x_k, b_k)|^2\right)^{-p_k}\wedge \Psi_{k-1}\left(a_{k}\left(1+|(x_k, b_k)|^2\right)^{p_k}, b_k, x_k, z_k\right)\right\}
\end{equation}

\noindent is open. This is clear since, if $y_k$ is in the set, then there is some $\varepsilon<\delta$ and $z_k$ such that $|y_k-z_k|^2\le \varepsilon$ and $\Psi_{k-1}(\dots, z_k)$. So for $y_k'$ sufficiently close to $y$, the same value of $z_k$ satisfies $|y_k'-z_k|^2\le \frac12(\varepsilon +\delta)<\delta$. Now the set:

\[\bigcup_{\varepsilon<\delta}\left\{y_k|\exists (x_k, b_k)\in [-b_{k+1}, b_{k+1}]^{n+1},z_k\in \mathbb{R}^n:|z_k-y_k|^2\le \varepsilon\left(1+|(x_k, b_k)|^2\right)^{-p_k}\wedge \Psi_{k-1}\left(\dots\right)\right\}\]

\noindent is the union of values of \eqref{eq:epsilonopenset} for $(x_k, b_k)\in [-b_{k+1}, b_{k+1}]^n$, so is open for any fixed $a_{k}$ and $b_{k+1}$. This is the necessary condition to use \Cref{lem:openswap}, proving that \eqref{eq:induction6} (and so \eqref{eq:inductionbase}) is equivalent to:

\begin{multline}\label{eq:induction7}
\lim_{a_{k}\rightarrow \infty}\lim_{\varepsilon\rightarrow 0^+}\lim_{b_{k+1}\rightarrow \infty}\forall y_k\in [-a_k, a_k]^{n}:\exists (x_k, b_k)\in [-b_{k+1}, b_{k+1}]^{n+1},z_k\in \mathbb{R}^n:\\|z_k-y_k|^2\le \varepsilon\left(1+|(x_k, b_k)|^2\right)^{-p_k}\wedge \Psi_{k-1}\left(a_{k}\left(1+|(x_k, b_k)|^2\right)^{p_k}, b_k, x_k, z_k\right)
\end{multline}

When $\epsilon\ge 1$, the constraint $|z_k-y_k|^2\le \varepsilon\left(1+|(x_k, b_k)|^2\right)^{-p_k}$ is stronger than $|z_k-y_k|^2\le 1$. Since we are taking the limit as $\epsilon\rightarrow \infty$, we can add a constraint $|z_k-y_k|^2\le 1$ without changing the truth value of the formula. Then, using \Cref{lem:monotonemerge}, we re-merge the negative $\varepsilon\rightarrow 0^+$ and $a_{k}\rightarrow \infty$ limits into a single $a_{k}\rightarrow \infty$ limit. The final step is now to swap the $b_{k+1}\rightarrow \infty$ and $a_{k+1}\rightarrow \infty$ limits using \Cref{lem:aswap}. Similarly to how we obtained \eqref{eq:induction3}, we need to ``leave behind'' a quantifier over $a_k$, so we replace \eqref{eq:induction7} with the equivalent formula:

\begin{multline}\label{eq:induction8}
\lim_{a_{k+1}\rightarrow \infty}\lim_{b_{k+1}\rightarrow \infty}\forall (y_k, a_k)\in [-a_{k+1}, a_{k+1}]^{n+1}:\exists (x_k, b_k)\in [-b_{k+1}, b_{k+1}]^{n+1},z_k\in \mathbb{R}^n:\\a_k|z_k-y_k|^2\le \left(1+|(x_k, b_k)|^2\right)^{-p_k}\wedge |z_k-y_k|^2\le 1 \wedge \Psi_{k-1}\left(a_{k}\left(1+|(x_k, b_k)|^2\right)^{p_k}, b_k, x_k, z_k\right)
\end{multline}

In order to use \Cref{lem:aswap}, we need a bound on the degree of:

\begin{multline}\label{eq:aswapset}
\Big\{(a_{k+1}, b_{k+1})|\forall (y_k, a_k)\in [-a_{k+1}, a_{k+1}]^{n+1}:\exists (x_k, b_k)\in [-b_{k+1}, b_{k+1}]^{n+1},z_k\in \mathbb{R}^n:\\a_k|z_k-y_k|^2\le \left(1+|(x_k, b_k)|^2\right)^{-p_k}\wedge |z_k-y_k|^2\le 1 \wedge \Psi_{k-1}\left(a_{k}\left(1+|(x_k, b_k)|^2\right)^{p_k}, b_k, x_k, z_k\right)\Big\}
\end{multline}

The polynomials defining \eqref{eq:aswapset} have degree at most $2p_k+3$ (the highest degree one being $a_k\left(1+|(x_k, b_k)|^2\right)^{p_k}|z_k-y_k|^2\le 1$), there are $k$ universal quantifier blocks with $n+1$ variables each and $k$ existential quantifier blocks with $2n+1$ variables each, so by \Cref{lem:projections}, \eqref{eq:aswapset} has degree at most $p_k^{(6\alpha n)^{2k}}\le q_k$. Using \Cref{lem:aswap} we finally see that \eqref{eq:induction8} (and so \eqref{eq:inductionbase}) is equivalent to:

\begin{multline}\label{eq:induction9}
\lim_{b_{k+1}\rightarrow \infty}\lim_{a_{k+1}\rightarrow \infty}\forall (y_k, a_k)\in [-a_{k+1}, a_{k+1}]^n:\exists (x_k, b_k)\in [-b_{k+1}, b_{k+1}]^n\left(1+|(y_k, a_k)|^2\right)^{q_k},z_k\in \mathbb{R}^n:\\a_k|z_k-y_k|^2\le \left(1+|(x_k, b_k)|^2\right)^{-p_k}\wedge |z_k-y_k|^2\le 1\wedge \Psi_{k-1}\left(a_{k}\left(1+|(x_k, b_k)|^2\right)^{p_k}, b_k, x_k, z_k\right)
\end{multline}

Expanding out $\Psi_{k-1}$, \eqref{eq:induction9} is the formula we wanted, proving the claim by induction.
\end{proof}

\paragraph{Adding one Existential Quantifier.}
We have now established \Cref{thm:compactformula}. 
However, we might also need the theorem for formulas that start with an existential instead of a universal quantifier. 
We establish this in the next corollary.
\begin{corollary}\label{cor:compactformulaexists}
Let $\Phi$ be a formula of the form $\exists x_{k}\forall y_{k-1}\dots\forall{y_1}\exists x_1: f(x_1, y_1, \dots, y_k, x_{k+1})=0$, where $f$ is a polynomial of degree $4$, possibly depending on some free variables, and each $x_i$ or $y_i$ is in $\mathbb{R}^n$, $n\ge 1$. Then $\Phi$ is true if and only if the compactification $\Psi$ is true. 
\end{corollary}

\begin{proof}
By \Cref{thm:compactformula}, the formula:

\[\forall y_{k-1}\dots\forall{y_1}\exists x_1: f(x_1, y_1, \dots, y_k, x_{k+1})=0\]

\noindent is equivalent to its compactification $\Psi_{k-1}$.

So $\Phi$ is equivalent to:

\begin{equation}\label{eq:existsbase}
\exists x_k:\lim_{b_{k}\rightarrow \infty}\lim_{a_k\rightarrow \infty} \Psi_{k-1}(a_k, b_k, x_k, y_k)
\end{equation}

The limit $b_k\rightarrow \infty$ is positive and the limit $a_k\rightarrow \infty$ is negative. We add limits to \eqref{eq:existsbase} obtain:

\begin{equation}\label{eq:existslimits}
\lim_{b\rightarrow \infty}\exists x_k\in [-b, b]^n:\lim_{b_{k}\rightarrow \infty}\lim_{a_k\rightarrow \infty} \Psi_{k-1}(a_k, b_k, x_k, y_k)
\end{equation}

Using \Cref{lem:easyswap}, we can exchange the order of the $\exists x_k$ quantifier and the positive $b_k\rightarrow \infty$ limit, obtaining:

\begin{equation}\label{eq:exists1}
\lim_{b\rightarrow \infty}\lim_{b_{k}\rightarrow \infty}\exists x_k\in [-b, b]^n:\lim_{a_k\rightarrow \infty} \Psi_{k-1}(a_k, b_k, x_k, y_k)
\end{equation}

We use \Cref{lem:monotonelimits} to replace the positive limit $b_k\rightarrow \infty$ with $\exists b_k$, or equivalently $\lim_{b_{k+1}\rightarrow \infty}\exists b_k\in [-b_{k+1}, b_{k+1}]$. Next, we can use \Cref{lem:boundedopenprojections} to show that $\{x_k : \Psi_k(a_k, b_k, x_k)\}$ is closed for fixed $a_k$ and $b_k$, so we can use \Cref{lem:closedswap} to exchange the order of the $\exists x_k$ quantifier and the negative $a_k\rightarrow \infty$ limit. We obtain a formula:

\begin{equation}\label{eq:exists2}
\lim_{b\rightarrow \infty}\lim_{b_{k+1}\rightarrow \infty}\lim_{a_k\rightarrow \infty} \exists x_k\in [-b, b]^{n}, b_k\in [-b_{k+1}, b_{k+1}]:\Psi_{k-1}(a_k, b_k, x_k)
\end{equation}

Finally, we use \Cref{lem:monotonemerge} to merge the $\lim_{b\rightarrow \infty}$ and $\lim_{b_{k+1}\rightarrow \infty}$ limits, obtaining a formula:

\begin{equation}\label{eq:exists3}
\lim_{b_{k+1}\rightarrow \infty}\lim_{a_k\rightarrow \infty} \exists (x_k, b_k)\in [-b_{k+1}, b_{k+1}]^{n+1}:\Psi_{k-1}(a_k, b_k, x_k)
\end{equation}

\noindent equivalent to $\Phi$. \eqref{eq:exists3} is the compactifiction of $\Phi$, as required. 
\end{proof}

\paragraph{Fully-Closed Compactification.}
Note that $\forall y\in S:\Phi(y)$ is equivalent to $\forall y:(y\in S)\implies \Phi(y)$, which is equivalent to $\forall y:(y\notin S)\vee \Phi(y)$. If we want this to be a closed condition, then we actually want $S$ to be open, not closed. 

\begin{definition}[Fully-closed compactification]
Let $\Phi$ be a formula of the form 
$\forall y_k\exists x_k\dots\forall{y_1}\exists x_1: f(x_1, y_1, \dots, x_k, y_k)=0$
or 
$\exists x_k\forall y_{k-1}\dots\forall{y_1}\exists x_1: f(x_1, y_1, \dots, y_{k-1}, x_k)=0$, 
where $f$ is a polynomial of degree $4$ and each $x_i$ or $y_i$ is in $\mathbb{R}^n$, $n\ge 1$. 
The we define the fully-closed compactification of $\Phi$ to be the formula $\Psi$ obtained from $\Psi$ in the same way as the compactification of $\Psi$, except that the outermost universal quantifier $\forall y_k\in \mathbb{R}^n$ is replaced with:

\[\forall (y_k, a_k)\in (-a_{k+1}, a_{k+1})^{n+1}\]

\noindent and each remaining universal quantifier $\forall y_\ell\in \mathbb{R}^n$ is replaced with:

\[\forall (y_\ell, a_\ell)\in (-a_{\ell+1}, a_{\ell+1})^{n+1}\left(1+|(x_{\ell+1}, b_{\ell+1})|^2\right)^{p_{\ell+1}}\]
\end{definition}

\begin{lemma}\label{lem:universalbounds}
Let $S$ be a closed semialgebraic set. Let $I$ be a set and let $\bar{I}$ be its closure. Then the formulas $\forall y\in I:y\in S$ and $\forall y\in \bar{I}:y\in S$ are equivalent. 
\end{lemma}

\begin{proof}
$I\subseteq \bar{I}$, so if $\forall y\in \bar{I}:y\in S$ is true, then $\forall y\in I:y\in S$ is also true.

Suppose that $\forall y\in I:y\in S$ is true. Let $y\in \bar{I}$, so there is a sequence of points $y_i$ in $I$ converging to $y$. By assumption, each $y_i$ is in $S$. Since $S$ is closed, $y\in S$. In particular, $\forall y\in \bar{I}:y\in S$. 
\end{proof}

\begin{corollary}[Corollary of \Cref{thm:compactformula}]\label{cor:compactformula2}
Given a formula $\Phi$, its fully-closed compactification $\Psi$ is true if and only if $\Phi$ is true.
\end{corollary}

\begin{proof}
We start by using \Cref{thm:compactformula} to obtain an equivalent compactified formula $\Psi$. We then use \Cref{lem:universalbounds} to replace each restricted quantifier $\forall y\in [-a, a]$ with $\forall y\in (-a, a)$, starting from the outermost ones and working inwards. We can use \Cref{lem:boundedopenprojections} to show that the required sets are closed.
\end{proof}

%%%%%%%%%%%%%%%%%%%%%%%%%%%%%%%%%%%%%
\subsection{Applications}
%%%%%%%%%%%%%%%%%%%%%%%%%%%%%%%%%%%%%

Now, we use \Cref{thm:compactformula} to prove \Cref{thm:FOTRINV} and \Cref{thm:hierarchyinv}. 

\paragraph{Replacing Multiplication Constraints by Inversion.} 
In order to replace multiplication constraints by inversion, we use the idea from \cite{AAM22}. In their proof that \etrinv is \ER-complete, the authors of \cite{AAM22} implicitly prove the following:

\begin{lemma}(Abrahamsen, Adamaszek, and Miltzow \cite{AAM22})\label{lem:invconversion}
Let $\Phi(x_1, \dots, x_n)$ be a quantifier-free formula in the first-order theory of the reals with constraints of form $x=\frac18$, $x+y=z$, and $xy=z$. Then we can construct in a polynomial time a quantifier-free formula $\Psi(x_1', \dots, x_n', y_1, \dots, y_k)$ such that:

\begin{itemize}
    \item If $x_1, \dots, x_n\in [-\frac18, \frac18]$, $\Phi(x_1, \dots, x_n)$ is true and $x_i'=x_i+\frac78$, then there exist $y_1, \dots, y_k\in [\frac12, 2]$ making $\Psi(x_1', \dots, x_n', y_1, \dots, y_k)$ true
    \item If $\Psi(x_1', \dots, x_n', y_1, \dots, y_k)$ is true for some values of $x_1', \dots, x_n', y_1, \dots, y_k\in [\frac12, 2]$ and $x_i'=x_i+\frac78$, then $\Phi(x_1, \dots, x_n)$ is true.
\end{itemize}

The formula $\Psi$ is a conjunction of constraints of form $x=1$, $x+y=z$, and $xy=1$. 
\end{lemma}

\paragraph{Hardness of \Arealhierarchy{2k} and \Erealhierarchy{2k+1} for the Real Polynomial Hierarchy}
\HIERARCHYINVTHM* 

\begin{proof}
Let $\Phi$ be a first-order formula of form:

\[\forall y_k\exists x_k\dots\forall y_1\exists x_1:f(x_1, y_1, \dots, x_k, y_k)\]

\noindent where $f$ is a polynomial of degree at most $4$, and each $x_i$ or $y_i$ is in $\mathbb{R}^n$. By \Cref{lem:closedconstraints}, the decision problem for such formula is complete for \Arealhierarchy{2k}. Using \Cref{cor:compactformula2}, we construct a formula:

\begin{equation}\label{eq:convertedformula}
\lim_{b_{k+1}\rightarrow\infty}\lim_{a_{k+1}\rightarrow \infty}\forall (y_k, a_k)\in \mathbb{R}^{n+1}\dots\exists(x_1, b_1, z_1)\in \mathbb{R}^{n+1+n}: \Psi(y_k, a_k, x_k, b_k, z_k, \dots)
\end{equation}

\noindent equivalent to $\Psi$. Here we have pushed all the restrictions on the quantifiers into the formula $\Psi$. So $\Psi$ is the quantifier-free formula:

\begin{equation}\label{eq:psiform}
P_k\implies \left(Q_k\wedge \left(P_{k-1}\implies \left(Q_{k-1}\wedge \dots P_1\implies \left(Q_1\wedge R\right)\right)\right)\right)
\end{equation}

\noindent where:

\begin{itemize}
    \item $P_k\equiv (y_k, a_k)\in (-a_{k+1}, a_{k+1})^{n+1}$
    \item for $1\le \ell\le k-1$, $P_\ell\equiv \bigwedge_{\ell=1}^{k-1}(y_\ell, a_\ell)\in (-a_{\ell+1}, a_{\ell+1})^{n+1}\left(1+|(x_{\ell+1}, b_{\ell+1})|^2\right)^{p_{\ell+1}}$
    \item for $1\le \ell\le k$, $Q_\ell\equiv (x_\ell, b_\ell)\in [-b_{\ell+1}, b_{\ell+1}]^{n+1}(1+|(y_\ell, a_\ell)|^2)^{q_\ell}\wedge\left(1+|(x_\ell, b_\ell)|^2\right)^{p_\ell}a_\ell|z_\ell-y_\ell|^2\le 1\wedge|z_\ell-y_\ell|^2\le 1$
    \item $R\equiv f(x_1, z_1, \dots, x_k, z_k)=0$
\end{itemize}

\eqref{eq:convertedformula} has $2k$ blocks of quantifiers with $n+1$ variables each. The polynomials defining $\Psi$ have degree at most $2p_k+1$, where $q_k=4^{\left(6\alpha n\right)^{k(2k+1)}}$, where $\alpha$ is the constant from \Cref{lem:projections}. So by \Cref{lem:limitevaluation}, we can realize the limits with:

\[b_{k+1}=\exp_2\left(\tau 4^{\left(6\alpha n\right)^{k(2k+1)}(\beta(n+4))^{k+4}}\right), a_{k+1}=\exp_2\left(\tau 4^{\left(6\alpha n\right)^{k(2k+1)}\left(2(\beta(n+4))^{k+4}+1\right)}\right)\]

\noindent where $\tau$ is an upper bound for the bit-length of coefficients in $f$. Note that $\tau\le |\Phi|$. 

Using these values for $b_{k+1}$ and $a_{k+1}$, we now want to find a single constant upper bound $\Lambda$ such that the truth value of \eqref{eq:convertedformula} does not change if we restrict each quantified variable to $[-\Lambda, \Lambda]$. 

It is useful to think about \eqref{eq:convertedformula} as a game with a devil and a human. The human wants to make $\Psi$ true and the devil wants to make it false. On the first turn, the devil chooses values for $y_k$ and $a_k$, then the human chooses values for $x_k$, $b_k$ and $z_k$. The players take turns until all the variables are set. 

We will inductively find variables $A_i$ and $B_i$ such that, if the devil (resp. human) is setting the variables $y_\ell$ and $a_\ell$ (resp. $x_\ell$, $b_\ell$, and $z_\ell$) and they have a winning move, then they have a winning move where all the variables set have absolute values at most $A_i$ (resp. $B_i$).

On the first turn, the devil should never choose $(y_j, a_k)\notin (-a_{k+1}, a_{k+1})^{n+1}$, since this would immediately make $\Psi$ true. So we can choose $A_k\ge a_{k+1}$. 

On the human's first turn, there are two cases. If the devil has chosen $(y_k, a_k)\notin (-a_{k+1}, a_{k+1})^{n+1}$, then the formula $\Psi$ will always be true, regardless of what the players do on their remaining turns. So the human can choose $x_k, b_k$ and $z_k$ to all be $0$. If the devil choose $(y_k, a_k)\in (-a_{k+1}, a_{k+1})^{n+1}$ in the previous turn, then $P_k$ is satisfied, so the human should choose $(x_k, b_k)\in [-b_{k+1}, b_{k+1}]^{n+1}(1+|(y_k, a_k)|^2)^{q_k}$, otherwise $Q_k$ will be false and the human will lose. Since $(y_k, a_k)\in (-A_k, A_k)$, we can set $B_k\ge b_{k+1}\left(1+(n+1)A_k^2\right)^{q_k}$. 

As soon as a player chooses a variable outside the range for a turn, the outcome of the game becomes fixed, so on remaining moves the players can always choose $0$ for all the variables. On the other hand, if all previously set variables are in the range, then when the human is setting variables $x_\ell, b_\ell$ and $y_\ell$ for $\ell>1$, they should never choose values of variables larger than $B_{\ell+1}\left(1+(n+1)A_\ell^2\right)^{q_\ell}$. Similarly, when the devil is choosing $y_\ell$ and $a_\ell$ for $\ell<k$, they should never choose values larger than $A_{\ell+1}\left(1+(n+1)B_{\ell+1}^2\right)$. So we can choose:

\[A_{k}\ge a_{k+1}, B_k\ge b_{k+1}\left(1+(n+1)A_k^2\right)^{q_k}\]

\[A_{\ell}\ge A_{\ell+1}\left(1+(n+1)B_{\ell+1}^2\right)^{p_{\ell+1}}, B_{\ell}\ge B_{\ell+1}\left(1+(n+1)A_{\ell}^2\right)^{q_\ell}\]

The constant $\beta$ from \Cref{lem:limitevaluation} is at least as large as the constant $\alpha$ from \Cref{lem:projections}, so we can upper bound:

\[a_{k+1},b_{k+1}\le \exp_2\left(4^{\tau \left(6\beta n\right)^{6k^2}}\right)\]

We can also upper bound $p_\ell, q_\ell\le 4^{(6\beta n)^{6k^2}}$, so it is sufficient to set:

\[A_{k-i}=\exp_2\left(4^{(\tau +6i)\left(6\beta n\right)^{6k^2}}\right), B_{k-i}=\exp_2\left(4^{(\tau+3+6i)\left(6\beta n\right)^{6k^2}}\right)\]

\noindent and so we can set:

\[\Lambda = \exp_2\left(4^{(\tau +6k+3)\left(6\beta n\right)^{6k^2}}\right)\]

Let $R=2(\tau +6k+3)\left(6\beta n\right)^{6k^2}$, so $\Lambda=\exp_2\left(2^R\right)$. Note that $R$ is at most $|\Phi|^{\text{poly}(k)}$. 

The next step is to create a formula $\Upsilon$ equivalent to $\Psi$ where $\Upsilon$ has only a polynomial number of constraints of form $xy=z,x+y=z$ and $x=1$. Note that the constraint $x=0$ is equivalent to $x+x=x$. The formula $\Psi$ is quantifier-free, but $\Upsilon$ will be allowed to use existential quantifiers. For example, we could replace a constraint $x^2=1$ with:

\[\exists V: x\cdot x=V\wedge V=1\]

$\Psi$ involves some polynomials of exponential degree, but we can use repeated exponentiation to obtain the high-degree terms. We observe that $x=y^{\mu^{j}}$ is equivalent to:

\[\exists z_0=y, \dots, z_j=x:z_1=z_0^\mu, \dots, z_j=z_{j-1}^{\mu}\]

Since $\left(y^{\mu^i}\right)^\mu=y^{\mu^{i+1}}$, so by induction we have $z_i=y^{\mu^i}$. So we can construct constraints of form $x=y^{\mu^j}$ so long as $\mu$ and $j$ are polynomially bounded. By setting $y$ to a constant, we can also obtain exponentially large constants in this way. 

With addition, multiplication, and constant constraints, using repeated exponentiation when necessary, we can construct new variables equal to the constants $a_{k+1}$, $b_{k+1}$ and each of the non-constant terms appearing in \eqref{eq:psiform}. Write $V_1, \dots, V_m$ for the all the new variables introduced. So far we have constructed a formula:

\begin{equation}\label{eq:vformula}
\exists V_1, \dots, V_m:\Theta(V_1, \dots, V_m, y_k, a_k x_k, b_k, z_k, \dots)
\end{equation}

\noindent where $\Theta$ is a conjunction of constraints of form $xy=1$, $x+y=z$, and $x=1$. The formula \eqref{eq:vformula} is always true; for any assignment of the $y_\ell, a_\ell, x_\ell, b_\ell$ and $z_\ell$ variables there is a unique assignment of the $V_i$ variables making $\Theta$ true. 

We want to add some new new constraints to create a formula equivalent to $\Psi$. The strategy is similar to the proof of \Cref{lem:closedconstraints} (Lemma 3.2 in \cite{SS17}), but we repeat it here because we need to make sure that we can bound the new variables that are introduced. 

Replacing each term of form $P\implies Q$ in \eqref{eq:psiform} with $\neg A\vee B$, we obtain a monotone Boolean formula with literals of form $V_i=0$ or $V_i\ge 0$. We can replace each term $V_i\ge 0$ with an expression of form $\exists W:V_i-W^2=0$, so that we have a monotone Boolean formula involving literals of form only $V_i=0$.

Given real variables $V_i$ and $V_j$, $V_i=0\wedge V_j=0$ is equivalent to $V_i^2+V_j^2=0$ and $V_i=0\vee V_j=0$ is equivalent to $V_iV_j=0$. So we can combine literals in the formula inductively, adding new $V$ variables equal to $V_i^2+V_j^2$ or $V_iV_j$. We eventually obtain a single condition $V_i=0$ equivalent to \eqref{eq:psiform}. 

We have now constructed a formula $\Upsilon$ with free variables $y_\ell, a_\ell, x_\ell, b_\ell$ and $z_\ell$, quantifiers:

\[\exists V_1, \dots, V_{m'}, W_1, \dots, W_r\]

\noindent and constraints of form $x+y=z$, $xy=z$ and $x=1$, such that $\Upsilon$ is equivalent to $\Psi$. The formula $\Upsilon$ has length at most $|\Phi|^{\text{poly}(k)}$. In particular, the formula:

\begin{equation}\label{eq:upsilonformula}
\forall (y_k, a_k)\in \mathbb{R}^{n+1}\dots\exists(x_1, b_1, z_1)\in \mathbb{R}^{n+1+n}: \Upsilon(y_k, a_k, x_k, b_k, z_k, \dots)
\end{equation}
is equivalent to \eqref{eq:convertedformula}. 

The next step is to create a new formula $\Upsilon'$ such that:

\begin{equation}\label{eq:upsilonprimeformula}
\forall (y_k', a_k')\in [-\tfrac18, \tfrac18]^{n+1}\dots\exists(x_1', b_1', z_1')\in [\tfrac18, \tfrac18]^{n+1+n}: \Upsilon'(y_k', a_k', x_k', b_k', z_k', \dots)
\end{equation}
is equivalent to \eqref{eq:convertedformula}. The formula $\Upsilon'$ should use constraints of form $x+y=z, xy=z$ and $x=\frac18$, and can have additional variables with quantifiers of form $\exists x\in [\frac12, 2]$.

By an earlier calculation, we know that the truth value of $\Psi$ doesn't change if all the $y_\ell, a_\ell, x_\ell, b_\ell$ and $z_\ell$ variables are restricted to $[-1, 1]\exp_2\left(2^R\right)$. The $V$ variables are products or sums of previously-defined variables, and each $W$ variable can be chosen to be either $0$ or the square root of a previously-defined $V$ variable. So if each of the $y_\ell, a_\ell, x_\ell, b_\ell$ and $z_\ell$ variables are in $[-1, 1]\exp_2\left(2^R\right)$, then if there is a satisfying assignment of the $V$ and $W$ variables, then there is such an assignment with the $V$ and $W$ variables in $[-1, 1]\exp_2\left(2^{R+|\Upsilon|}\right)$.

Starting with $\frac18$ and repeatedly squaring, we construct constants $\lambda_1=\exp_2\left(-3\cdot 2^{R}\right)$ and $\lambda_2=\exp(-3\cdot 2^{R+|\Upsilon|})$. 

For each $y_\ell, a_\ell, x_\ell, b_\ell$ and $z_\ell$ in \eqref{eq:convertedformula}, we replace it with a new variable $y_\ell', a_\ell', x_\ell', b_\ell'$ and $z_\ell'$ that should be scaled down by a factor of $\lambda_1$, e.g. $y_k'=\lambda y_k$. For each of the additional $V$ and $W$ variables in $\Upsilon$, we replace it with a corresponding $V'$ or $W'$ variable that should be scaled down by a factor of $\lambda_2$, e.g. $V_1'=\lambda_2V_1$. 

For a variable $v$ in \eqref{eq:convertedformula}, let $\lambda_v$ be the scale factor for that variable. That is, $\lambda_v=\lambda_2$ if $v$ is a $V$ or $W$ variable and $\lambda_v=\lambda_1$ otherwise. Each constraint $u+v=w$ in $\Upsilon$ can be replaced by $\lambda_v\lambda_wu'+\lambda_u\lambda_wv'=\lambda_u\lambda_vw'$ each constraint $uv=w$ can be replaced by $\lambda_wu'v'=\lambda_u\lambda_vw'$, and each constraint $v=1$ can be replaced by $v'=\lambda_v$. 

Since $\lambda_1$ and $\lambda_2$ are smaller than $1$, multiplying by one of them won't cause a variable to leave the range $[-\frac18, \frac18]$. So we have obtained a formula $\Upsilon'$ such that \eqref{eq:upsilonprimeformula} is equivalent to \eqref{eq:convertedformula}. 

Using \Cref{lem:invconversion}, we can construct a formula $\Upsilon''$ such that:

\begin{multline}\label{eq:invformula}
\forall (y_k'', a_k'')\in [\tfrac34, 1]^{n+1}:\exists (x_k'', b_k'', z_k'')\in [\tfrac34, 1]^{n+1+n}\dots \\\forall (y_1'', a_1'')\in [\tfrac34, 1]^{n+1}:\exists (x_1'', b_1'', z_1'')\in [\tfrac34, 1]^{n+1+n}:\Upsilon''(\dots)
\end{multline}

\noindent is equivalent to $\Phi$, where $\Upsilon$ has constraints of form $xy=1$, $x+y=z$, and $x=1$. The formula $\Upsilon$ has length $|\Phi|^{\text{poly}(k)}$. Note that $\Upsilon''$ is not quantifier free, but it has only existential quantifiers of form $\exists v\in [\frac34, 1]$ (if $v$ is one of the variables from $\Upsilon'$) or $\exists v\in [\frac12, 2]$ (if $v$ is one of the variables added by \Cref{lem:invconversion}).

In order to obtain a \Arealhierarchy{2k} formula, we replace each existential quantifier $\exists v\in [\frac34, 4]$ with one $\exists v\in [\frac12, 2]$. By the construction of $\Upsilon'$, this doesn't change the truth value of \eqref{eq:invformula} since all those variables are restricted to be in a smaller range some constraints in $\Upsilon''$.

This shows that there is a polynomial-time reduction from \Arealhierarchy{2k} to \Arealhierarchyinv{2k}. The result for \Erealhierarchy{2k+1} follows by an essentially identical argument. 
\end{proof}

\paragraph{\QR-Hardness of \fotrinv}

The reduction used in \Cref{thm:hierarchyinv} is polynomial time for each fixed level of the hierarchy, but exponential in the number of quantifier alternations. The exponential overhead occurs because we need to construct polynomials with doubly-exponential degree and variables with triply-exponential size. In order to get a polynomial-time reduction for \QR, we use extra quantifier alternations in order to create expressions of doubly-exponential degree. The idea comes from the proof of Theorem 11.18 in \cite{BasuPollackRoy2006}.

\FOTRINVTHM* 

\begin{proof}
The proof follows in the same way as the proof of \Cref{thm:hierarchyinv}. The only difference is in how we produce the large constants and high-degree polynomials.

In the proof of \Cref{thm:hierarchyinv}, we used repeated exponentiation to produce constraints of exponential degree. Here, we use a trick from \cite{BasuPollackRoy2006} (see Theorem 11.18) to produce constraints of doubly-exponential degree. We call the idea \emph{quantified repeated exponentiation}.

Given a variable $x$ and integers $d, \alpha$ and $j$, we will show how to construct a formula $\Psi_j(x, y)$ equivalent to $x=y^{d^{\alpha^j}}$. 
This is by induction on $j$. 
When $j=0$, we can directly construct the formula $x-y^d=0$, which is equivalent to $x=y^d=y^{d^{\alpha^0}}$. 
Now suppose that we have constructed a formula $\Psi_j(x, y)$ equivalent to $x=y^{d^{\alpha^j}}$. We can construct $\Psi_{j+1}$ by:

\[\Psi_{j+1}(z_\alpha, z_{0})=\exists z_1,\dots, z_{\alpha-1}:\forall i\in \{1, \dots, \alpha\}:\Psi_j(z_{i-1}, z_i)\]

We can't directly create the constraint $\Psi_j(z_{i-1}, z_i)$, since the indices depend on the variable $i$. Instead, we use the equivalent formula: 

\[\exists a, b: \Psi_j(a, b) \wedge \bigwedge_{k=1}^{\alpha}\left(i=k\implies (a=z_{k-1}\wedge b=z_{k})\right)\]

To see that $\Psi_{j+1}(z_\alpha, z_0)$ is equivalent to $z_\alpha=z_0^{d^{\alpha^j}}$, we use induction on $i$ to show that $z_i$ must be equal to $z_0^{d^{i\alpha^j}}$. The case $i=0$ is clear since $z_0=z_0^{d^{0}}$. If $z_i=z_0^{d^{i\alpha^j}}$, then the constraint $\Psi_j(z_{i}, z_{i+1})$ requires that $z_{i+1}=\left(z_0^{d^{i\alpha^j}}\right)^{d^{\alpha^j}}=z_0^{d^{i\alpha^j}\cdot d^{i\alpha^j}}=z_0^{d^{(i+1)\alpha^j}}$, as required. By induction, we conclude that $z_\alpha=z_0^{d^{\alpha^j}}$.

Clearly, the length of $\Psi_j$ is bounded by $\mathcal{O}(j\alpha+d)$, so is polynomial in $j$, $\alpha$, and $d$.

We can then proceed exactly as in the proof of \Cref{thm:hierarchyinv}, using quantified repeated exponentiation instead of simple repeated exponentiation where necessary. We obtain a formula: 

\begin{multline*}
\forall (y_k'', a_k'')\in [\tfrac34, 1]^{n+1}:\exists (x_k'', b_k'', z_k'')\in [\tfrac12, 2]^{n+1+n}\dots \\\forall (y_1'', a_1'')\in [\tfrac34, 1]^{n+1}:\exists (x_1'', b_1'', z_1'')\in [\tfrac12, 2]^{n+1+n}:\Upsilon''(\dots)
\end{multline*}
 equivalent to $\Phi$, with length polynomial in $|\Phi|$ independently of $k$. The difference here from the proof of \Cref{thm:hierarchyinv} is that $\Upsilon''$ has both existential and universal quantifiers, so the total number of quantifier alternations increases.
\end{proof}

\paragraph{The problem \rangedfotrinv{}.}

For the packing game, we want to restrict the variables to a small range, where the range needs to get smaller when there are more variables. So we will use a different version of \fotrinv called \rangedfotrinv. 

\begin{definition}[\rangedfotrinv]
    For each integer $c>0$, there is a problem $\rangedfotrinv$, where we are given a quantified formula $\exists x_1\forall x_2 \dots Q_n x_n : \Phi(x_1,\dots,x_n)$, and intervals $[a_i, b_i]$ for $i=1, \dots, n$, where $\Phi$ is a conjunction of constraints of form $x=1,\quad x+y=z,\quad x\cdot y=1,$
    for $x,y,z \in \{x_1, \ldots, x_n\}$ and each $|b_i-a_i|\le n^{-c}$. 
    Each quantifier bounds exactly one variable and the quantifiers keep alternating between $\exists$ and $\forall$.
    The goal is to decide whether $\exists x_1\in [a_1, b_1]\forall x_2\in [a_2, b_2]\dots Q_n x_n\in [a_n, b_n]:\Phi(x_1, \dots, x_n)$ is true.
\end{definition}

The construction from \cite{AAM22} can be used to obtain a ranged version of \Cref{lem:invconversion}:

\begin{lemma}\label{lem:rangedinvconversion}
There is some universal constant $C$ such that the following is true: Let $\Phi(x_1, \dots, x_n)$ be a quantifier-free formula in the first-order theory of the reals with constraints of form $x=\frac18$, $x+y=z$, and $xy=z$. For each variable $x_i$, let $[a_i, b_i]\subseteq [\frac18, \frac18]$ be an associated interval, and let $\varepsilon$ such that $|b_i-a_i|\le \varepsilon$ for all $i$. Then we can construct in a polynomial time a quantifier-free formula $\Psi(x_1', \dots, x_n', y_1, \dots, y_k)$ and intervals $[c_i, d_i]$ for $i=1, \dots, k$ such that:

\begin{itemize}
    \item $|d_i-c_i|\le C\varepsilon$ for all $i=1, \dots, k$
    \item If each $x_i\subseteq [a_i, b_i]$, $\Phi(x_1, \dots, x_n)$ is true and $x_i'=x_i+\frac78$, then there exist $y_i\in [c_i, d_i]$ making $\Psi(x_1', \dots, x_n', y_1, \dots, y_k)$ true
    \item If $\Psi(x_1', \dots, x_n', y_1, \dots, y_k)$ is true for some values of $x_1', \dots, x_n', y_1, \dots, y_k\in [\frac12, 2]$ and $x_i'=x_i+\frac78$, then $\Phi(x_1, \dots, x_n)$ is true.
\end{itemize}

The formula $\Psi$ is a conjunction of constraints of form $x=1$, $x+y=z$, and $xy=1$, and the length of $\Psi$ is $\mathcal{O}(|\Phi|)$. 
\end{lemma}

\begin{theorem}\label{thm:rangedfotrinv}
For each $c>0$, \rangedfotrinv is \fotr-hard
\end{theorem}

\begin{proof}
This follows by essentially the same argument as in the proofs of \Cref{thm:hierarchyinv,thm:FOTRINV}, with the differences being that:

\begin{itemize}
    \item When we create the scaled formula $\Upsilon'$, we scale down far enough that the variables are in $[-\varepsilon, \varepsilon]$ instead of $[-\frac18, \frac18]$. This can be done by simply using smaller values of $\lambda_1$ and $\lambda_2$
    \item We use \Cref{lem:rangedinvconversion} rather than \Cref{lem:invconversion} to create $\Upsilon''$ from $\Upsilon'$
\end{itemize}

We want $\varepsilon$ to be smaller than $|\Upsilon''|^{-c}$, but in order to scale down $\lambda_1$ and $\lambda_2$, we need to add additional constraints to $\Upsilon'$, making $\Upsilon''$ larger. However, we can scale down $\lambda_1'=2^{-j}\lambda_1, \lambda_2'=2^{-j}\lambda_2$ using only $\mathcal{O}(j)$ extra constraints. Since the exponential $2^{-j}$ becomes small much faster than the polynomial $|\Upsilon'|^{-c}$, and since $|\Upsilon''|=\mathcal{O}(|\Upsilon'|)$ by \Cref{lem:rangedinvconversion}, we can choose $j$ large enough to satisfy $\varepsilon\le |\Upsilon''|^{-c}$.
\end{proof}

%%%%%%%%%%%%%%%%%%%%%%%%%%%%%%%%%%%%%%%%%%
\newpage
\section{Machine Model}
\label{sec:MachineModel}
%%%%%%%%%%%%%%%%%%%%%%%%%%%%%%%%%%%%%%%%%%
This section is dedicated to prove the following theorem.

\MachineModelTHM*

The proof uses the existing and similar statement from the class \ER which was provided in~\cite{EvdHM20, BSS89}.
The class \ER can be either described with the help of a real Turing machine or a logical formula of a certain form and we use the proof of this fact to prove our theorem.
The crucial long known idea behind the theorem is that machines can compute the value of a logical sentence and that logical sentences are powerful enough to simulate computations if they are allowed to use existential quantifiers. 
As many of those theorems are formulated using real Turing machines instead of algorithms, we will use the language of real Turing machines here.
It was shown (Lemma 28~\cite{OracleSeparation}) that algorithms using the real RAM and real Turing machines can simulate one another in polynomial time.

We use a recent lemma from Kirn, Meijer, Miltzow and Bodlander that extracts the technical part that we need in a convenient form~\cite{KirnMeijerMiltzowBodlaender2025}.
(Many papers have this technical part hidden as part of the their proof or omit it as ``straightforward''~\cite{BC06}. 
We are happy to hear if there is a better source.)
For the sake of completeness for the readers not familiar with the technique, we restate their lemma (Lemma~13~\cite{KirnMeijerMiltzowBodlaender2025}) with some small adaptions to our situation.
Below, we explain the changes in detail.

\begin{lemma}[Encoding Computations in First-Order Formulas~\cite{KirnMeijerMiltzowBodlaender2025}]
\label{lem:EncodingComputationFOFormula}
    Let $M$ be a polynomial-time real Turing machine, $q$ be a polynomial, and $w \in \{0,1\}^*$ be a string of length $n$. 
    
    Then one can compute in polynomial time a first-order formula $\Phi = \exists z\varphi(u,z)$ , 
    with $\varphi$ quantifier free, such that for all $u \in \R^{q(n)}$, 
    \[M(w,u)=1 \text{ if and only if }  \exists z : \varphi(u,z) .\]
\end{lemma}

We now state the small changes in the way that we state the lemma.

\begin{itemize}
    \item The lemma is originally stated for a general structure $\mathcal{R}$, we replace this by the specific structure $\R = (\R,0,1,-,+,\cdot,\leq)$ of the reals that we are working with.
    So here, we just apply it to a special case, the reals.
    \item Next, the lemma in the article omits to mention that $\varphi$ can be constructed in polynomial time on a normal Turing machine.
    However, later in the article the authors discuss this in Theorem 14.
    \item Furthermore, the original lemma was working with $w\in \R^*$ instead of $w\in \{0,1\}^*$, but again the authors mention that the lemma is also valid with binary input, as discussed in Corollary 15.
    \item At last, we slightly changed the names of the variables to be consistent with the scheme of naming variables in this article.
\end{itemize}

We are now ready for the proof of the main theorem of this section.

\begin{proof}
    We begin with the forward direction ``$\Rightarrow$''.
    Let $L$ be a language in \QR. 
    Then there is a polynomial time algorithm that translates any instance $w$ to a sentence $\Phi$ in the first-order theory of the reals. 
    We can assume without loss of generality that $\Phi$ is in prenex normal form.
    That is all the quantifiers are in the front of $\Phi$, i.e.,
    \[\Phi = \exists x_1 \in \R^n \forall y_1 \in \R^n \ldots \exists x_k \in \R^n \forall y_k \in \R^n \varphi(x_1,y_1,\ldots ),\]
    for $k,n$ polynomial in the length of $w$.
    Furthermore, given the real vectors $x_1,y_1,\ldots $ there is a real Turing machine that evaluates if $\varphi(x_1,y_1,\ldots )$ is true.
    Thus our real Turing machine $M$ that we need first translates $w$ into $\varphi$ and then evaluates it on the real vectors that it is given.

    For the reverse direction ``$\Leftarrow$'', we use the lemma above.
    Let $L$ be a language and $M$ the machine mentioned in the theorem.
    It holds that 
    \begin{equation} \label{eqn:MachineProperty}
        w \in L \Leftrightarrow \exists x_1 \in \R^n \forall y_1 \in \R^n \ldots \exists x_k \in \R^n \forall y_k \in \R^n : M(x_1,y_1,\ldots, x_k,y_k,w) = 1
    \end{equation}
    We have to design a polynomial time algorithm $A$ on a normal Turing machine
    that reduces $L$ to \fotr. 
    In other words, for each instance $w$ to $L$, we have to construct in polynomial time an \fotr sentence $\Phi$ such that $w\in L$ if and only if $\Phi$ is true.
    Let $u = (x_1,y_1,\ldots)$ from the assumption about $M$.
    But let's think about them as free variables. 
    Then \Cref{lem:EncodingComputationFOFormula} implies that there is a first-order formula $\Phi(u) = \exists z \varphi(z,u)$ that has the same truth value as the evaluation of $M$.
    Thus we can replace $M(u,w)$ by $\exists z \varphi(z,u)$ and this gives.

    \[
    w \in L \Leftrightarrow \Phi \]
     With 
    \[ \Phi = \exists x_1 \in \R^n \forall y_1 \R^n \ldots \exists x_k \in \R^n \forall y_k \R^n \exists z \in \R^n : \varphi(z,x_1,y_1,\ldots, x_k,y_k),\]
    being an \fotr sentence constructed in polynomial time.
\end{proof}

%%%%%%%%%%%%%%%%%%%%%%%%%%%%%%%%%%%%%%%%%%
\newpage
\section{The \PackingGame}
\label{sec:Packing}
%%%%%%%%%%%%%%%%%%%%%%%%%%%%%%%%%%%%%%%%%%
This section is dedicated to show the following theorem.

\PackingTHM* 

\paragraph{Proof overview.}
The membership is straightforward and uses the machine model of \QR.
To show hardness, we reduce from \fotrinv. 
Our reduction is reminiscent of the  \ER-hardness reductions in~\cite{PackingER, Westerdijk2024}, 
with some additional tricks to deal with the alternating nature of the game.
The core idea is to encode every variable with a single \textit{variable piece}.
The value of the variable then corresponds one-to-one with the vertical placement of that specific variable piece, as in \Cref{fig:notches}.

\begin{figure}[h]
    \centering
    \includegraphics{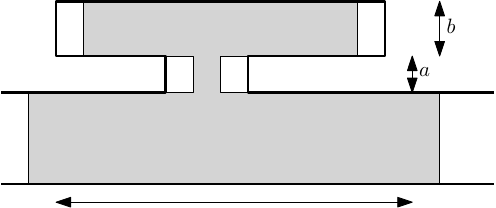}    
    \caption{This variable piece can be moved left and right corresponding ot the value of the variable that it represents.}
    \label{fig:notches}
\end{figure}

Note here that the piece has a jigsaw puzzle like extension on the top that fits exactly at one location in the boundary of the container.
The idea appeared to the best of our knowledge first in the paper by Abrahamsen, Miltzow and Seiferth in a much more sophisticated way~\cite{PackingER}.
In this simplified form it was also seen in the master's thesis of Westerdijk~\cite{Westerdijk2024}.

In order to enforce all the constraints, we connect the variables to an engine of pieces that enforce all the constraints that we need.
This engine was described before~\cite{Westerdijk2024}, and we will simply reuse this engine and also sketch it again for the convenience of the reader.
The trick to enforce that the first variable piece is placed first goes as follows.
The variable pieces are very big compared to all other pieces.
They are in fact so big that all other pieces can fit into just one of the spots of the variable pieces, even if the opponent tries to block of that space.
Thus, the players really need to prioritize the variable pieces.
Furthermore,  first variable piece is the largest, the second piece is a bit smaller and so forth. 
Intuitively this means that it makes sense to for both players to always place their biggest piece first, which will then correspond to the correct order.

\begin{figure}[btph]
    \centering
    \includegraphics{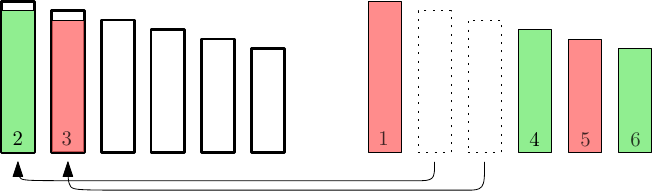}    
    \caption{The decreasing sizes of the variable pieces forces both players to place them in the correct order.}
    \label{fig:increasing-szies}
\end{figure}

At last, let us discuss how to build the engine that enforces all the constraints.
Actually, we just let the human construct it.
Meanwhile all the remaining pieces of the devil only fit in a distraction area of the container. Specifically the remaining devils pieces will be long and skinny.
In this way, the pieces of the devil fit only in the distraction area and the pieces of the human only fit inside the engine part of the construction. 
Assuming that all the constraints can be satisfied then the human has a natural motivation to fit all the pieces correctly inside the engine as then the human will win as eventually all the pieces are inside the  container. 
If not all the constraints of the \fotrinv instance can be satisfied then the engine cannot be completed and there will be at least one piece that cannot be fit by the human.

\begin{figure}
    \centering
    \includegraphics{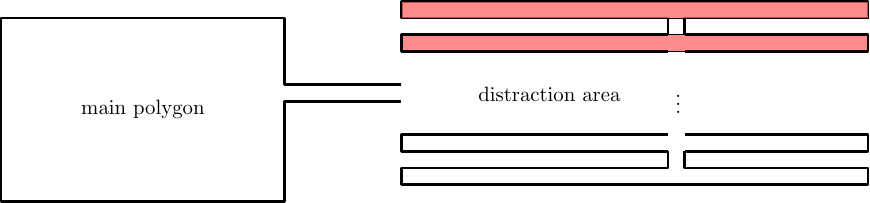}    
    \caption{The distraction pieces belong to the devil.
    They are super thin and super long compared to the remaining pieces of the human, but still dwarfed by the variable pieces.}
    \label{fig:packing-distraction}
\end{figure}

\subsection{\QR-membership}

In this paragraph, we show the following lemma.

\begin{lemma}
    The \PackingGame is in \QR.
\end{lemma}

We will use \Cref{thm:MachineModel} and describe an algorithm to recognize the \PackingGame.
Let $w$ be an instance of the packing game.
The $x_i$ and $y_i$ indicate both the position and which piece  each player is selecting. 
The algorithm $A$ has the task of checking that after each step we have a packing without overlapping pieces and that the piece that was selected is still available.
Although it is tedious to describe the algorithm $A$ in detail, the general experience of the theoretical computer science community assures that such an easy algorithm can be implemented using real RAM.

%%%%%%%%%%%%%%%%%%%%%%%%%%%%%%%%%%%%%%%%%%%%%%%
\subsection{\QR-hardness}
%%%%%%%%%%%%%%%%%%%%%%%%%%%%%%%%%%%%%%%%%%%%%%%

\begin{figure}[btph]
        \centering
        \includegraphics{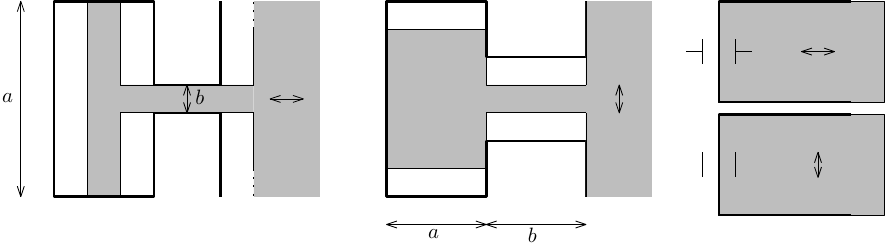}
        \caption{Left: We have two ways to link polygons. Right: We use the schematic symbols on the polygons to indicate if there is a linking.}
        \label{fig:packing-linking}
    \end{figure}

We reduce from \rangedfotrinv. Thus let us consider such an instance :

\[\exists x_1 \in \ \forall y_1 \in  \ \dots \exists x_n \ \forall y_n \in  \ \Phi(x_1,y_1,\ldots,x_n,y_n)\]
Here, $\Phi$ is a conjunction of constraints of the form $x+y = z$ and $x\cdot y = 1$.
We are furthermore given for each variable $x_i$ or $y_i$ a range $[a_i,b_i] \subseteq[1/2,2]$ of size $\varepsilon$. We want to know if the sentence is true with every variable restricted to its respective range.
As mentioned before each variable will be encoded using one variable piece.
Before we will describe the variable piece in detail, we will describe how to link polygons.
This is one important tool to ensure that all the pieces are at the intended position.

\paragraph{Linking Polygons.}
    
We can create little notches either on the boundary of the container or on the boundary of a piece.
The notches are shown in \Cref{fig:packing-linking}.
We have two types of notches.
Those where we can push towards or away from the notch with the piece that fits inside, see to the left of \Cref{fig:packing-linking}.
The other type of notch allows to slide along the boundary.
We indicate this symbolically by either two $\vdash$'s or two parallel lines as indicated on the right of \Cref{fig:packing-linking}.
Furthermore, each notch comes with two rational numbers $a,b$ that describe the shape of the notch.
It holds that for any two notches the counter part can only fit into exactly one notch.

\paragraph{Encoding Variables.}
\begin{figure}[tbp]
    \centering
    \includegraphics{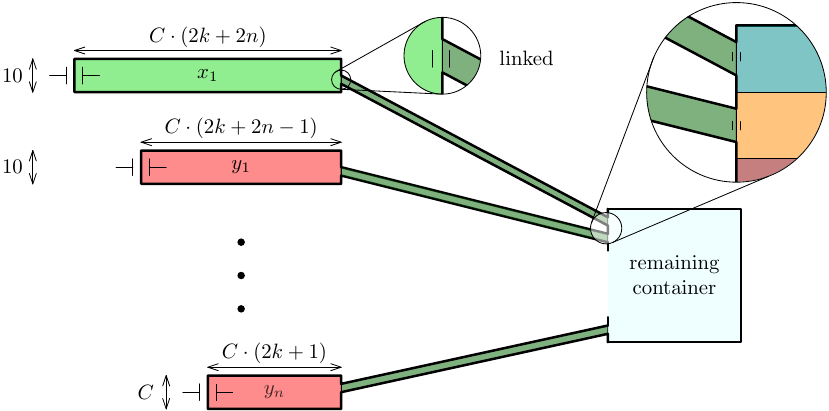}
    \caption{We have $2n$ variable pieces corresponding to the $2n$ variables $x_1,y_1,x_2,\ldots$~. Each piece has a quarter of the size of the previous piece.}
    \label{fig:variables}
\end{figure}
We have $2n$ variable pieces $p_1,p_2,\ldots$ corresponding to the $2n$ variables $x_1,y_1,x_2,\ldots$~. 
Each piece has length $C\cdot 2k 4^{2n-i}$ and height $10$.
We choose height $10$, because every forthcoming piece has height at most $10$.
Furthermore, $k$ indicates the total number of remaining pieces.
The number  $C$ is defined as an upper bound on the length of the longest piece in the remaining construction.
(We will argue later that $C$ is polynomial in terms of the input.)

We have pockets for them at the left of the container that fit exactly those pieces.
Furthermore, each piece is linked with a notch as explained in the previous paragraph.
The amount that each piece can move to the left or to the right is exactly the range of the corresponding variable.
We show the following lemma.

\begin{figure}[tbph]
    \centering
    \includegraphics{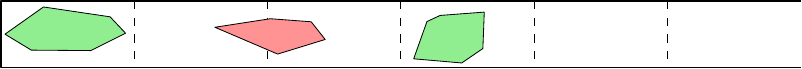}
    \caption{We can subdivide the spot for piece $p_{2n}$ into $2k$ boxes of size $10\times C$.}
    \label{fig:packing-filling}
\end{figure}

\begin{lemma}
    Each player must place their largest remaining variable piece on every turn; otherwise, the other player can force a win.
\end{lemma}    
\begin{proof}
    We show that if one of the players deviates from this order they will certainly lose the game.
    We start for simplicity with the case that the last piece $p_{2n}$ was not placed at the correct position, by the devil.
    When it is the humans turn they can place one piece at the spot for piece $p_{2n}$.
    Consequently, the devil can never place piece $p_{2n}$ for the rest of the game.
    To show that the devil will lose, it remains to argue that the human can place all its remaining pieces in the spot for $p_{2n}$.
    To see this we note that if we can subdivide the spot for piece $p_{2n}$ into $2k$ boxes of size $10 \times C$.  
    Clearly, every such box is large enough to fit any of the $k$ remaining pieces.
    Furthermore, any piece can intersect at most two such boxes. 
    Thus regardless of the strategy of the devil and the human there is always a way to place for the human to place any of its remaining pieces.
    
    Now we consider the case that one of the earlier pieces was not placed.
    To make notation easier, we assume that this is the first piece $p_1$ that was not placed at its intended spot by the human.
    (The case for the piece $p_i$ is the same up to renaming the pieces and potentially switching the roles of the human and the devil.)
    The strategy for the devil for the remainder is to place its biggest remaining piece in the biggest remaining spot until there are no variable pieces left and then one variable spot will be left and the devil can employ the strategy that we described above for the human.
    It remains to show that this is a winning strategy for the devil.
    
    First, we can already note here that the human can never place $p_1$.
    There are several ways to prove that the devil has a winning strategy. 
    Then it is sufficient to show that the devil can place all its pieces.
    For that purpose note that for $p_i$ there are $i$ possible spots where the piece can be placed.
    Thus if the devil places piece  $p_i$ in the $i$-th turn than at least one of the spots must stay available simply due to the number of remaining spots is one larger than the number of the pieces placed so far.
    Once all possible variable pieces are placed and one spot remains empty, we use the strategy described above to fill all the remaining pieces into the empty variable spot.
    The human will lose as it can never place $p_1$.
\end{proof}

%%%%%%%%%%%%%%%%%%%%%%%%%%%%%%%%%%%%%%%%%%%%%%%
\paragraph{Distraction Area}
%%%%%%%%%%%%%%%%%%%%%%%%%%%%%%%%%%%%%%%%%%%%%%%
All the remaining pieces have that we will describe will be human pieces and the devil receives an equal number of distraction pieces.
Those distraction pieces will be chosen longer than all the human pieces combined and thinner than any of the human pieces.
At the same time we build part of the polygon that fits exactly those pieces and no other piece of the human.
See \Cref{fig:packing-distraction} for an illustration.
We get the following lemma.
\begin{lemma}
    The human has a winning strategy if and only if they can place all the remaining pieces in the main part of the container. 
\end{lemma}
\begin{proof}
    We first consider the first case:
    The human has a piece that they cannot place.
    As the human and the devil have an equal number of pieces and the devil can place all of its distraction pieces.
    Eventually the human will lose.

    We know consider the second case that
    the human can place all of the remaining pieces in the main area of the polygon.
    As the devil can also place all its pieces all pieces will be placed and the human wins by definition of the \PackingGame.
\end{proof}

From now on, we assume that all the remaining pieces we describe are meant for the human.

%%%%%%%%%%%%%%%%%%%%%%%%%%%%%%%%%%%%%%%%%%%%%%%
\paragraph{Connecting to the Engine.}
%%%%%%%%%%%%%%%%%%%%%%%%%%%%%%%%%%%%%%%%%%%%%%%
As described in the previous paragraph, we have variable pieces and later we will describe gadgets that will simulate constraints given by the \fotrinv $\Phi$ that we started from.

The first step is to wire the information from the variable pieces to all the remaining pieces. 
The difficulty here is that the variable pieces are huge compared to the remaining pieces.
We do this using thin pieces, see the green pieces in \Cref{fig:variables}.
(Less thin then the Devil's distraction pieces of course.)
They are linked to the variable pieces and they are much shorter than in the visualization.
All of those green pieces are human pieces.

%%%%%%%%%%%%%%%%%%%%%%%%%%%%%%%%%%%%%%%%%%%%%%%
\paragraph{The Engine.}
%%%%%%%%%%%%%%%%%%%%%%%%%%%%%%%%%%%%%%%%%%%%%%%

Note that all the remaining pieces and the remaining part of the boundary corresponds exactly to the \ER-hardness description of the packing problem.
A long and intricate proof for convex pieces was done by Abrahamsen, Miltzow and Seiferth~\cite{PackingER} in combination of an article of Miltzow and Schmiermann~\cite{miltzow2024classifying}.
A much easier proof was described by Westerdijk~\cite{Westerdijk2024}, but only for polygonal pieces that are not necessarily convex.
We follow the approach of Westerdijk, but only sketch this engine as this part is literally the same.
We will sketch the following lemma, which follows from the correctness of the \ER-hardness hardness  of geometric packing~\cite{Westerdijk2024}.

\begin{lemma}
    Let $u_1,v_1,\ldots u_n,v_n$ be the values corresponding to the placement of the variable pieces
    placed by the devil and the human.
    Then the human can place all his remaining pieces in the engine part of the polygon.
\end{lemma}

We start by providing an overview on how the \textit{engine} works.
For an illustration, we refer the reader to \Cref{fig:packing-engine+switches}.
\begin{figure}[tbph]
    \centering
    \includegraphics{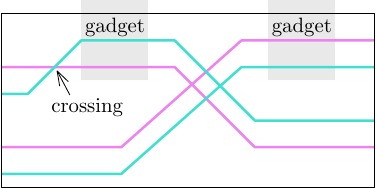}
    \hfill
    \includegraphics{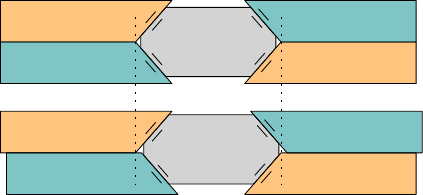}
    \caption{Left: A high level overview on the inner working of the engine.
    Right: }
    \label{fig:packing-engine+switches}
\end{figure}
The information from the variable pieces will be wired through lanes to the gadgets.
We will describe both how the lanes work as well as the way that two lanes can switch without losing any information.
The lanes are responsible to bring the right variable information to the gadgets on the top of the engine. 
The gadgets are then responsible to enforce all the constraints of $\Phi$.

%%%%%%%%%%%%%%%%%%%%%%%%%%%%%%%%%%%%%%%%%%%%%%%
\paragraph{Lanes and Switches.}
%%%%%%%%%%%%%%%%%%%%%%%%%%%%%%%%%%%%%%%%%%%%%%%

A sketch of the lane pieces and the switch pieces is depicted in \Cref{fig:packing-engine+switches}.
Although the variable pieces have already some fixed value, it is helpful to see what happens if we move a variable left or right.
The lane piece moves along it naturally as it is tightly connected via our linking gadget.
In \Cref{fig:packing-engine+switches} we see how the movement of the cyan left piece only influences that position of the cyan right piece and none of the orange pieces.
The same holds for symmetry reasons for the orange pieces and we see how information about the variable pieces switches lane.

%%%%%%%%%%%%%%%%%%%%%%%%%%%%%%%%%%%%%%%%%%%%%%%
\paragraph{Gadgets.}
%%%%%%%%%%%%%%%%%%%%%%%%%%%%%%%%%%%%%%%%%%%%%%%
The addition and inversion gadgets are illustrated in \Cref{fig:packing-addition} and \Cref{fig:packing-inversion}.

\begin{figure}[btph]
    \centering
    \includegraphics{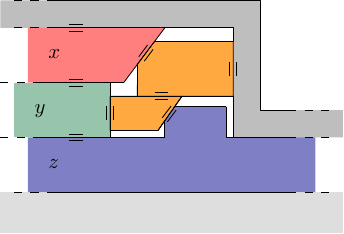}
    \hfill
    \includegraphics{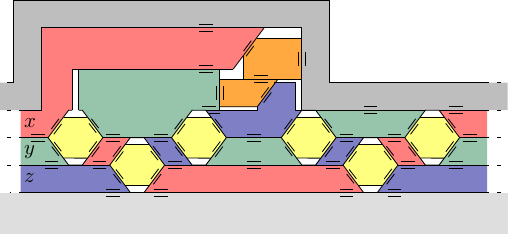}
    \caption{The idea of the addition gadget. Figure taken from~\cite{Westerdijk2024}.}
    \label{fig:packing-addition}
\end{figure}

We describe first the idea of the addition gadget.
Imagine that the red piece representing $x$ is pushed to the right.
That would correspond to an increase of the $x$ variable.
Then this pushes also the orange pieces down and to the right and eventually pushes the blue piece representing $z$ to the right by the same amount.
Similarly with the green piece representing $y$.

\begin{figure}[tbph]
    \centering
    \includegraphics{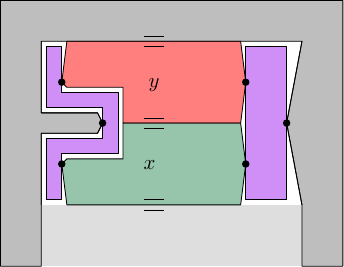}
    \hfill
    \includegraphics{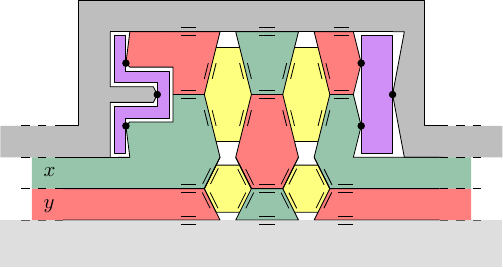}
    \caption{The idea of the inversion gadget. Figure taken from~\cite{Westerdijk2024}.}
    \label{fig:packing-inversion}
\end{figure}
Now, the inversion gadget is more intricate.
When the piece green piece representing $x$ is moved to the right then the two purple pieces rotate slightly,
which in turn pushes the red piece representing $y$ to the left.
A small geometric calculation reveals that the relation between the value $x$ and $y$ is precisely described by $x\cdot y = 1$.

%%%%%%%%%%%%%%%%%%%%%%%%%%%%%%%%%%%%%%%%%%%%%%%%%%%%%%%
\newpage
\section{Plane Graph Drawing Game}
\label{sec:GraphGame}
%%%%%%%%%%%%%%%%%%%%%%%%%%%%%%%%%%%%%%%%%%%%%%%%%%%%%%%
In this section, we consider the closely related problems \GraphInPolygon and \PartialDrawingExtensibility and define a \devilgame{} version of those.
We start with a formal definition of \GraphInPolygon as introduced by~\cite{LMM}.

\begin{definition}[\GraphInPolygon]
    In the \GraphInPolygon problem, we are given a planar graph $G$ and a polygonal region $R$ with some vertices of $G$ assigned to fixed positions on the boundary of $R$. We must decide whether $G$ admits a planar straight-line drawing in $R$ respecting the fixed vertices.
\end{definition}

Note that the edge of the graph 
$G$ are allowed to be drawn on the boundary of $R$ and also overlap with vertices of $G$.
The \PartialDrawingExtensibility problem was introduced by Patrignani's~\cite{P06} in 2006 and is defined as follows.

\begin{definition}[\PartialDrawingExtensibility]
    In the \PartialDrawingExtensibility problem, we are given a planar graph $G$ and a planar straight-line drawing $\mathcal{D}(G')$ of a subgraph $G'$ of $G$. We must decide whether $\mathcal{D}(G')$ admits an extension of it that is a planar straight-line drawing of $G$.
\end{definition}

We would like to highlight the similarities between \GraphInPolygon and \PartialDrawingExtensibility.
Any polygon can be modeled as a straight-line drawing of a graph.
However, in \GraphInPolygon edges may go through vertices of the polygon, whereas for \PartialDrawingExtensibility, the `polygon' is part of the graph.
In that case, this would clearly violate planarity.
We translate those algorithmic problems into their \devilgame{} variants.
We start by giving a formal definition of \GraphInPolygonGame{}, the devil's version of \GraphInPolygon.

\begin{definition}[\GraphInPolygonGame]
    We are given a planar graph $G$ and a polygonal region $R$ with some vertices of $G$ assigned to fixed positions on the boundary of $R$ and an ordered list $x_1, x_2, \dots, x_n$ of the $n$ non-fixed vertices. The two players (the \devil and the \human) alternate in drawing vertices in the graph (including straight-line edges to already-drawn neighbors) according to the given order, where the graph must stay planar. The first player who cannot draw a vertex without violating planarity loses, while the other player wins. If the entire graph is drawn, the human wins.
\end{definition}

We now give a formal definition of \PlanarExtensionGame{}, the game version of \PartialDrawingExtensibility.

\begin{definition}[\PlanarExtensionGame]
We are given a planar graph $G$, a planar straight-line drawing of a subgraph $G'$ in the plane, and
    a linear ordering on the vertices, not in $G'$.
    The human and the devil take turns in placing vertices of $G$ according to the given linear order.
    All edges must be drawn in a straight line fashion without intersections.
    That is, no two edges are allowed to share a point in their interior, any two vertices must be at distinct locations and no vertex is allowed to be in the interior of an edge.
    The game ends when the first player cannot place a vertex without creating an intersection.
    If the graph is drawn completely in the plane the human wins.
\end{definition}

Again, note that \GraphInPolygonGame and \PlanarExtensionGame differ.
In \GraphInPolygonGame, the graph may have edges and vertices on the polygon boundary.
On the other hand, in \PlanarExtensionGame, the boundaries of the sections drawn are part of the graph, so drawing more edges and vertices there would violate the planarity of the graph.

We first show that some planar versions of \fotrinv are \QR-complete in \Cref{sub:planarfotrinv}.
Then we show completeness of \GraphInPolygonGame in \Cref{sub:ClosedGraphGame}. 
At last, we show \QR-completeness of \PlanarExtensionGame in \Cref{sub:OpenGraphGame}
We emphasize that both \GraphInPolygonGame and \PlanarExtensionGame are \QR-complete.
The \PartialDrawingExtensibility is not known to be \ER-complete.

%%%%%%%%%%%%%%%%%%%%%%%%%%%%%%%%%%%%%%%%%%%%%%%%%%%%%%
\subsection{\planarfotrinv}
\label{sub:planarfotrinv}
%%%%%%%%%%%%%%%%%%%%%%%%%%%%%%%%%%%%%%%%%%%%%%%%%%%%%%
In this section, we will show that \planarfotrinv is \QR-complete.
Both of our proofs, the proof showing \GraphInPolygonGame is \QR-complete and the proof showing \PlanarExtensionGame is \QR-complete, reduce from \planarfotrinv.
However, the two proofs require slightly different versions of \planarfotrinv: one with only equalities ($=$) and one with only inequalities $(\leq)$.
Normally, a formula with access to $\leq$ can clearly model any formula using only $=$, this may not hold in the planar case.
Splitting any constraint of the form $x + y = z$ into two constraints $x + y \leq z$ and $z \leq x + y$ may violate planarity, as this introduces more edges in the variable-constraint graph.
In this section, we prove that both variants are \QR-hard.

\begin{definition}[\planarfotrinv]
    \label{def:planarfotrinv} In the problem \planarfotrinv, we are given a quantified formula $\exists x_1\forall x_2 \dots Q_n x_n : \Phi(x_1,\dots,x_n)$, where $\Phi$ consists of a conjunction between a set of equations and inequalities of the form
    $x + y \preceq z\text{, } x \cdot y \preceq 1\text{, for } x, y, z \in \{x_1,\ldots, x_n\}.$
    We define two variations of \planarfotrinv.
    In \planarfotrinvequal, we let $\preceq$ equal $=$, so the formulas only have equalities.
    However, in \planarfotrinvinequal, $\preceq$ equals either $\leq$ or $\geq$, so the formula only has inequalities.
    When we refer to \planarfotrinv without specifying the variant, the claim holds for both variants.
    
    Furthermore, we require planarity of the \emph{variable-constraint incidence graph}, which is the bipartite graph that has a vertex for every variable and every constraint and an edge when a variable appears in a constraint. 
    The goal is to decide whether the system of equations has a solution where each existentially-quantified variable is restricted to lie in $[1/2,4]$ and each universally-quantified variable to lie in $[3/4, 1]$. 
\end{definition}

\begin{theorem}
    \planarfotrinv is \QR-complete.
\end{theorem}

To prove \planarfotrinv is \QR-complete, we must show that it is in \QR and that there exists a formula $\Psi$ with a planar variable-constraint incidence graph that is true iff $\Phi$ is true.
As such, the proof that \planaretrinv is \ER-hard (for both variations \cite{LMM, DKM}) carries over to our setting, as the only difference is that the number of quantifiers lifts it from \planaretrinv to \planarfotrinv and from \ER-hardness to \QR-hardness.
However, for clarity, we will provide an explicit proof \planarfotrinv is \QR-hard instead of simply referring to the existing proofs.

\begin{lemma}
    \planarfotrinv is in \QR.
\end{lemma}
\begin{proof}
    By definition, both \planarfotrinvequal and \planarfotrinvinequal have a FOTR-formula as input and can thus be decided in \QR.
\end{proof}

We now prove both variants of \planarfotrinv are \QR-hard.

\begin{lemma}
    \planarfotrinvequal is \QR-hard.
\end{lemma}
\begin{proof}
We present the construction from \cite{DKM}.
Consider an instance 
$\exists x_1\forall x_2 \dots Q_n x_n: \Phi$ of \fotrinv. 
Let $G$ be some embedding of the variable-constraint incidence graph $G(\Phi)$ of $\Phi$ in $\R^2$. 
We show that there must exist an embedding $H$ of $G(\Phi)$.
Suppose that $G$ is not crossing-free and consider a pair of crossing edges.
Let $X$ and $Y$ denote the variables corresponding to (one endpoint of) these edges as in~\cref{fig:crossing}. 

\begin{figure}[htb]
	\centering   \includegraphics{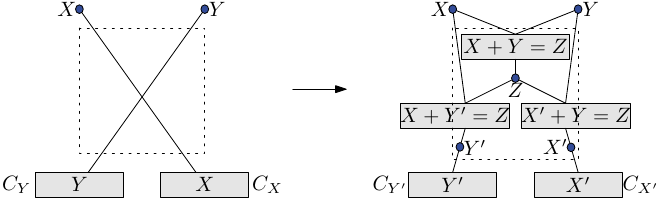}
	\caption{Eliminating crossings. Picture from \cite{DKM}.}   \label{fig:crossing}
\end{figure}

We introduce three new existential variables $X',Y',Z$ and three constraints: 
$X + Y =  Z$, $ X + Y' = Z$, and $X' + Y = Z$.
Observe that these constraints ensure that $X = X'$ and $Y=Y'$.
The new existentially-quantified variables will be added to a (new) $\exists$-quantifier at the end of the quantifier part of the \planarfotrinv-formula.
  
The gadget can be embedded in an arbitrarily small area, so it need not create new crossings.
As the embedding of $G$ can be modified so that the new 
 incidence graph $G'$ has strictly fewer crossings:
$G'$ loses the considered crossing and no new crossing is introduced.
We repeat this procedure until the incidence 
graph of the obtained formula is planar. 
Finally, note that 
$ 1 \leq Z = X + Y \leq 2 + 2=4 $ whenever $1/2 \leq X,Y \leq 2$, and the number of new variables and constraints is polynomial in $|\Phi|$, since the number of variables in each constraint in \fotrinv is at most three.
We emphasize that all three new variables are existentially quantified and will thus be in the range $[1/2, 4]$.
As the range of existentially-quantified variables is larger than the range of universally-quantified variables, this will never cause the new variables to be out of range, even if the constraint $X+Y=Z$ involves universally-quantified variables.
However, for clarity, we add constraints that for each existentially quantified variable $x_i$ in $\Phi$, we add a restriction $\frac{1}{2} \leq x \leq 2$.
That way, it will never be possible to pick a value for $x_i$ that changes the result of $\Psi$, while the auxiliary variables may still be in the range $[1/2, 4]$.
This yields a planar embedding of $G(\Phi)$ and thus proves \planarfotrinvequal is \QR-hard.
\end{proof}

\begin{lemma}
    \planarfotrinvinequal is \QR-hard
\end{lemma}
\begin{proof}
    We present the construction from \cite{LMM}.
    Consider an instance 
    $\left (\exists \; x_1 \forall x_2\ldots Q x_n \right )\colon \Phi(x_1, \dots, x_n)$ of \fotrinv.
    First, we create a new quantifier-free formula $\Psi$, which is equivalent to $\Phi$ but replaces all cases of $x + y = z$ in $\Phi$ with two new constraints: $x + y \leq z$ and $x + y \geq z$.
    Let $\mathcal{D}$ be some straight-line embedding of the constraint-variable graph $G(\Psi)$ in $\mathbb{R}^2$.
    This embedding may have many crossings, but we assume that no three edges cross in the same point, the only points that lie on an edge are its endpoints, no edge self-intersects.
    Now, we only need to eliminate these crossings.
    For this, we use the same gadget from \cite{LMM}, as depicted in \cref{fig:crossing2}.
    For any crossing between two edges incident to variables $x$ and $y$, this gadget builds two pairs of existentially-quantified variables $x, x'$ and $y, y'$ such that $x = x'$ and $y=y'$.
    To add these variables to $\Psi$, we add a $\exists$-clause to the end of $\Psi$ which will bind the new variables.
    (Note that every edge is incident to exactly a single variable, as all edges are between one variable and one constraint in the variable-constraint graph.)
    The gadget never has to cause new intersections, as it can be made arbitrarily small, while it always eliminates one crossing.
    As there are at most $O(n^2)$ intersections in $\mathcal{D}$ and each gadget adds a constant number of variables and constraints, we get an instance of \planarfotrinv with $O(n^3)$ constraints and variables.
    We have now found a polynomial-time reduction from \fotrinv to \planarfotrinv, proving the lemma.
\begin{figure}[htb]
	\centering   \includegraphics{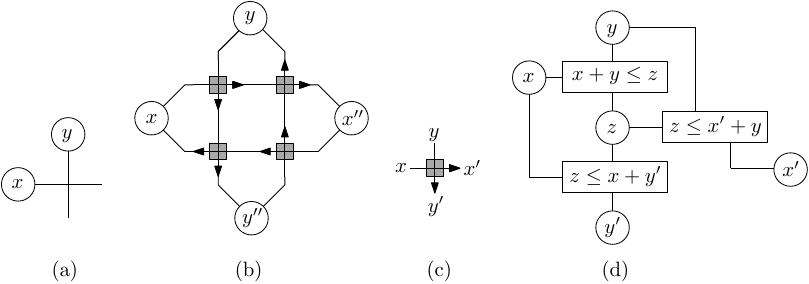}
	\caption{(a) A crossing. (b) A crossing gadget containing four half-crossing gadgets. (c) A half-crossing gadget pictorially. (d) The constraint-variable graph of the half-crossing gadget. Picture from \cite{LMM}.}   \label{fig:crossing2}
\end{figure}
\end{proof}

%%%%%%%%%%%%%%%%%%%%%%%%%%%%%%%%%%%%%%%%%%%%%%%%%%%%%%
\subsection{\GraphInPolygonGame}
\label{sub:ClosedGraphGame}
%%%%%%%%%%%%%%%%%%%%%%%%%%%%%%%%%%%%%%%%%%%%%%%%%%%%%%

Now, we are ready to prove the first theorem of this section by proving \GraphInPolygonGame is \QR-complete.
Compared to \cite{LMM}, the challenge lies in introducing alternating quantifiers, whereas the previous result only had an existentially-quantified formula.
Otherwise, most of the proof is identical to theirs.
For those parts, we include their write-up here with minor modifications.
Notably, the moves the \devil make represent the universally-quantified variables, whereas the quantifier alternations are modeled by the two players taking turns.
If the \human wants to win the game, they need to have a winning strategy for any move the \devil may make.
Otherwise, the \devil can pick the move again which \human cannot win.
This corresponds to the role of the $\forall$-quantifier in a formula;
for any value of the bound variables, the formula must stay true.

\begin{lemma}
    \GraphInPolygonGame is in \QR.
\end{lemma}
\begin{proof}
    Let $I$ be a \GraphInPolygonGame instance, which consist of a planar graph $G$ and polygonal domain $R$ with some vertices of $G$ identified with vertices of $R$ and a linear order on all the vertices of $G$ not on $R$.
    We use the machine model to show \QR-membership.
    We use the variables $x_i \in \R^2$ and $y_i \in \R^2$ to describe the vertices specified by the human and the devil respectively.
    The algorithm then checks for every newly added point that the graph drawn so far is planar.
    If one of the $x_i$ or $y_i$ does not adhere to this then the human, respectively the devil, is declared to lose.
    Otherwise, if all points are placed correctly the human wins.
    As the machine model can verify in polynomial time whether the drawing of $G$ is planar~\cite{LMM}, this concludes the proof that \GraphInPolygonGame is in \QR.
\end{proof}
\begin{lemma}
    \GraphInPolygonGame is \QR-hard.
\end{lemma}
\begin{proof}
    To prove that the problem is \QR-hard we give a reduction from \planarfotrinvinequal.
    Let $I = (\exists x_1 \forall x_2 \dots Q x_n) : \Phi (x_1, \dots, x_n)$ be an instance of \planarfotrinvinequal.
    We will build an instance $J$ of \GraphInPolygonGame such that $J$ admits an affirmative answer if and only if $I$  is a true sentence. 
    The idea of the reduction is to construct a gadget to
    represent variables, and to enforce the addition and inversion inequalities, i.e., $x + y \leq z ,  x + y \geq z, x\cdot y \leq 1,  x\cdot y \geq 1$. 
    We also need gadgets to copy and replicate
    variables ---``wires'' and ``splitters'' as conventionally used in reductions.
    Thereafter, we will
    describe how to combine those gadgets to obtain an instance $J$ of \GraphInPolygonGame.

    \paragraph{Encoding Variables.}
    We will encode the value of a variable in $[1/2,4]$ as the position of a vertex that is constrained to lie on a line segment of length $3.5$,   
    which we call a \emph{variable segment}.
    One end of a variable segment encodes the value $\frac{1}{2}$, the other end encodes the value $4$, and linear interpolation fills in the values between.
    Figure~\ref{fig:segment-end} shows one side of the construction that forces a vertex to lie on a variable segment. 
    The other side is similar.
    Note that, regardless of whether the \human or the \devil is making a move, this gadget ensures they must place the new point on the variable segment, as the \devil must always place vertices in legal positions, too.
    The game begins by the \human and the \devil repeatedly choosing the value for some variable.
    Note that the variables that the \devil must place are restricted to the interval $[\frac{3}{4}, 1]$.
    In principle, the exact bounds of the variable segment are arbitrary, so we alter the variable segment construction to restrict the point from being placed beyond $\frac{3}{4}$ and $1$ instead.
    
    By slight abuse of notation, we will identify
    a variable and the vertex representing it by the same name, if there is no ambiguity.
    For the description of the remaining gadgets, our figures will show variable segments (in green) without showing the polygonal gadgets that create them.

    \paragraph{Padding gadget.}
    After the values of all variables have been set, we use a series of gadgets to verify whether all constraints have been met.
    However, most of the constructions within these gadgets have to be performed by the \human, as otherwise the \devil could pick adversarial inputs that break the constructions.
    So, for simplicity, we want to assume the \human chooses the location of all the points in those gadgets.
    However, as the \human and \devil have to alternately place a point in the plane, we need the construction of a `\padding' gadget, a gadget whose purpose is to be disjoint from the rest of the problem, that allows the game to `skip' the turn of a player.
    
    The \padding gadget needs to withstand adversarial inputs, without them ever affecting other decisions.
    As such, the \padding gadget consists of a small pocket within the polygon, as shown in Figure~\ref{fig:trash}.
    When a player is asked to place a vertex in a \padding gadget, the only possible positions that yield a planar straight-line drawing can never be outside the gadget.
    As such, this gadget allows future gadgets to `skip' the turn of the \devil.
    
    \begin{figure}[htb]
	\centering   \includegraphics{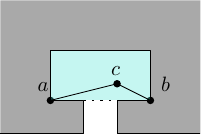}
	\caption{A \padding gadget. Note that $c$ must be within the blue region and can never lie outside of the gadget. We ensure no other point has to lie within this region, meaning $c$ can never threaten the linearity of $J$. }   \label{fig:trash}
    \end{figure}

    \begin{figure}[htb]
	\centering   \includegraphics{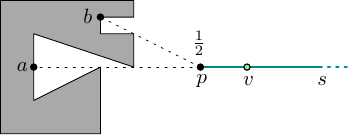}
	\caption{One half of the construction of a variable segment. The point $v$ must have an edge to $a$ and $b$. The edge to $a$ ensures it is on the line extension of the variable segment, whereas the edge to $b$ ensures that $\frac{1}{2}$ is the extremal point where $(v, b)$ stays within the polygon. On the other side of the variable segment, there will be a similar gadget to ensure $4$ is an extremal point, too. Picture inspired by \cite{LMM}.}   \label{fig:segment-end}
\end{figure}

    \paragraph{Copy gadget.}
    Given a variable segment for a variable $x$, we will need to transmit its value along a so-called ``wire''  to other locations in the plane. 
    We do this using a copy gadget in which we construct a variable segment for a new variable $y$ and enforce $x=y$. 
    We show how to construct a gadget that ensures $x \leq x'$ for a new variable $x'$, and then combine four such gadgets, enforcing
    This implies $x=y$.

    The gadget enforcing $x \leq x'$  is shown at the left of 
    Figure~\ref{fig:FullCopy}. 
    It consists of two parallel variable segments.
    In general, these two segments need not be horizontally aligned.
    In the graph we connect the corresponding vertices by an edge.
    The left and the right variables are encoded in opposite ways, i.e., $x$ increases as the vertex moves up and $x'$ increases as the corresponding vertex goes down.
    We place a hole of the polygonal region (shaded in the figure) with a vertex at the intersection point of the lines joining the top of one variable segment to the bottom of the other. The hole must be large enough that the edge from $x$ to $x'$ can only be drawn to one side of the hole.
    An argument about similar triangles, or the ``intercept theorem'', also known as Thales' theorem, implies $x \leq x'$. 

    We combine four of these gadgets to construct our copy gadget,  as illustrated on the right of Figure~\ref{fig:FullCopy}.
    Here, we ensure that $x=y$, by encoding the constraints $x \leq z_1 \leq y$ and $x \geq z_2 \geq y$.

    \begin{figure}[htb]
	\centering   
        \includegraphics[page=2]{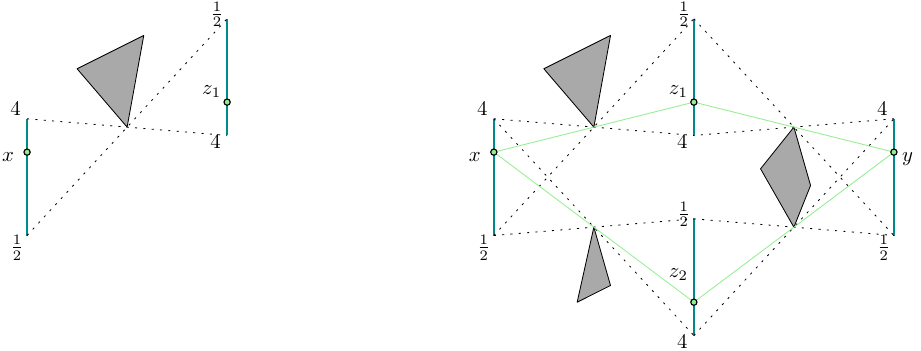}
	\caption{Left: a gadget ensuring $x \leq x'$. Right: a gadget ensuring $x = y$ by using the four subgadgets $x \leq z_1$, $z_1 \leq y$, $x \geq z_2$, and $z_2 \geq y$. Picture inspired by \cite{LMM}.}   \label{fig:FullCopy} 
    \end{figure}

    We note that in our construction, the \human has to set the values of $z_1$ and $z_2$ before the value of either $y$ is chosen (assuming $x$ is the value to be copied).
    Otherwise, the \devil could easily make $x = y$ impossible, by either setting $y$ to a different value, or setting $z_1$ or $z_2$ to a value where $x=y$ is impossible.
    For simplicity, we will simply let the \human set the values of all these variables.
    To give the \human the capacity to set these values after $x$ has been assigned, we assign the \devil points within \padding gadgets until the copy gadget is completed.

\paragraph{Splitter gadget.}
Since a single variable may appear in several constraints,
we may need to split a wire into two wires, each holding the correct value of the same variable.
Figure~\ref{fig:Duplicate} shows a gadget to split
the variable $x$ to variables $y_1$ and $y_2$.  
The gadget consists of two copy gadgets sharing the variable segment for $x$.
We can construct the two copy gadgets to avoid any intersections between them.
As with the copy gadget, we require the addition of some \padding gadgets, to ensure the \human can complete the splitter gadget.

    \begin{figure}[htb]
	\centering   \includegraphics{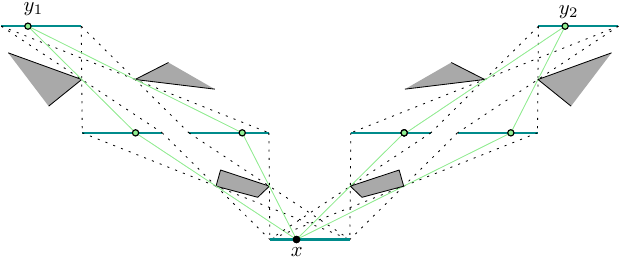}
	\caption{A splitter gadget consisting of two copy gadgets connecting to the same variable segment of $x$. Picture inspired by \cite{LMM}.}   \label{fig:Duplicate}
    \end{figure}

\paragraph{Turn gadget.}
We need to encode a variable both as a vertical and
as a horizontal variable segment. 
To transform one into the other we use a turn gadget, which you can think of as a copy gadget between variables on variable segments with perpendicular orientations.

The key idea is to construct
two diagonal variable segments for variables $z_1$ and $z_2$, and then transfer the value of the vertical variable segment to the horizontal variable segment using $z_1,z_2$. This is in fact very similar to the copy gadget, except that the intermediate variable segments are placed on a line of slope 1.
We do not know if it is possible to enforce the constraint $x\leq z$ directly. 
However, it is sufficient to enforce $x \leq f(z)$ for some  function $f$.
See the left side of \Cref{fig:turn}.
Interestingly, we don't even know the function $f$. However, we do know that $f$ is monotone and we can construct 
another gadget enforcing $y \geq f(z)$,
for the same function $f$,
by making another copy of the first gadget reflected through the line of the variable segment for $z$.

Combining four such gadgets, as on the right of \Cref{fig:turn}, 
 yields the following inequalities:
 $x \leq f_1(z_1), f_1(z_1) \leq y, y \leq f_2(z_2), f_2(z_2) \leq x$.
This implies $x = y$.

Again, as with the copy gadget, we add \padding gadgets so the \human can fix the value of $y$ after the value for $x$ is chosen.

    \begin{figure}[htb]
	\centering   
        \includegraphics[page=1]{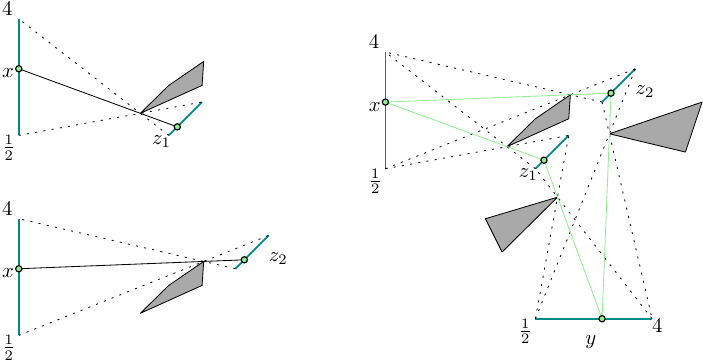}
	\caption{Top left: a gadget encoding $x \leq (f(z)$. Bottom left: a gadget encoding $x \geq f(z)$. Right: a turn gadget that ensures $x=y$ despite the two variables lying on variable segments of perpendicular orientations. Notice the similarities with the copy gadget in the construction.
    }   \label{fig:turn} 
    \end{figure}

    \paragraph{Addition {gadget}.}
    The gadget to enforce
    $x + y \geq z$ is depicted in
    \Cref{fig:Addition}. 
    Important for correctness is that the gaps between the dotted auxiliary lines have equal lengths.
    This proof is inspired by the following thought experiment.
    We assume $z$ to always to be the maximum 
    possible value.
    Furthermore, we assume that if we 
    fix the position of $y$, but move $x$ some distance $d$ to the
    left.
    What we would expect is that $z$ moves by the same distance
    to the left. 
    Actually, showing the last
    statement also proves the lemma, due to symmetry of $x$ and $y$. 
    We denote by $\ell$ the line that contains the variable segments
    of $x$ and $y$. We denote by $t$ half the distance that
    $z$ moves. Note that $t$ has a geometric interpretation
    as indicated in \Cref{fig:AdditionProof}.
    We need to show $d = 2t$. 
    The lengths $A,A',B,B'$ are defined as shown in \Cref{fig:AdditionProof}.
    Note that $B' = 2B$, because $\|a-b\| = \|b-c\| $.
    Similarly, follows $A'= 2A$.
    The lemma follows from
    \[d = B'-A'=2(B-A) = 2t.\]

    \begin{figure}[htb]
	\centering   \includegraphics{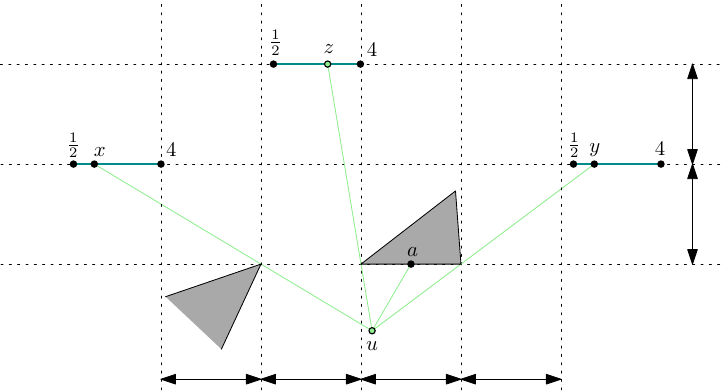}
	\caption{An addition gadget. Picture inspired by \cite{LMM}.}   \label{fig:Addition}
    \end{figure}
    
    \begin{figure}[htb]
    \centering
    \includegraphics[width=.8\textwidth]{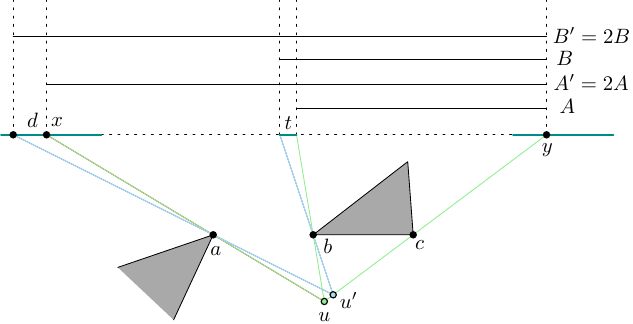}
    \caption{An illustration of the correctness 
    of the addition gadget.}
    \label{fig:AdditionProof}
    \end{figure}

\noindent
The gadget that enforces $x+y \leq z$
is just a mirror copy of the previous gadget.

\paragraph{Inversion {gadget}.}
The inversion gadgets to enforce
$x \cdot y \leq 1$ and $x \cdot y \geq 1$
are 
depicted in Figure~\ref{fig:inversion}.
We use a horizontal variable segment for $x$ and a vertical variable segment for $y$ and align them as shown in the figure, $1.5$ units apart both horizontally and vertically.  We make a triangular hole with its apex at point {$q$} as shown in the figure.
The graph has an edge between $x$ and $y$.

For correctness, observe that if $x$ and $y$ are positioned so that the line segment joining them goes through point {$q$}, then, because triangles $\Delta_1$ and $\Delta_2$ (as shown in the figure) are similar, we have $\frac{x}{1} = \frac{1}{y}$, i.e.~$x \cdot y=1$.
If the line segment goes above point {$q$} 
(as in the left hand side of Figure~\ref{fig:inversion})
then $x\cdot y \geq 1$, and if the line segment goes below then $x \cdot y \leq 1$.

    \begin{figure}[htb]
	\centering   \includegraphics{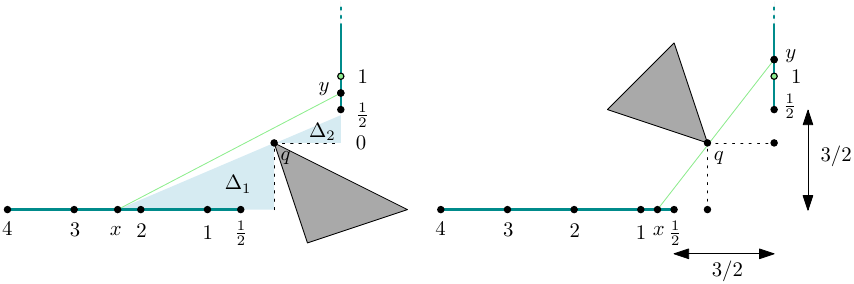}
	\caption{Left: an inversion gadget enforcing $x \cdot y \leq 1$. Right: an inversion gadget enforcing $x \cdot y \geq 1$. Picture inspired by \cite{LMM}.}   \label{fig:inversion}
    \end{figure}

\paragraph{Putting it all together.}
It remains to show how to obtain an instance of \GraphInPolygonGame  in polynomial time from an instance of \planarfotrinv.

All gadgets are based on axis-aligned variable segments.
As such, we want to generate an orthogonal drawing of the variable-constraint graph $G$ of $\Phi$.
In a rectilinear drawing, each vertex can have degree at most four.
To model this, we use splitter gadgets to create enough copies of each variable.
As the splitter gadget will occupy one of the four orthogonal directions of the variable, we generate a rectilinear drawing where each vertex has degree at most three.
Generating such a drawing can be done in polynomial time \cite{NishizekiRahman}.

The edges of $D$ act as wires and we replace each horizontal and vertical segment by a copy gadget, and   
replace every $90$ degree corner, by a turn gadget. 
Every splitter vertex and constraint vertex
will be replaced by the corresponding gadget, possibly using turn gadgets.
We add a big square to the outside, to ensure that everything is inside one polygon.
Finally, if two edges intersect, we add a crossing gadget.

It is easy to see that this can be done in polynomial time, 
as every gadget has a constant size description and can be
described with rational numbers, although, we did not do it explicitly.
In order to see that we can also use integers, note that we can scale everything with the least common multiple of all the denominators of all numbers appearing. This can also be done in polynomial time.
\end{proof}

%%%%%%%%%%%%%%%%%%%%%%%%%%%%%%%%%%%%%%%%%%%%%%%%%
\subsection{\PlanarExtensionGame}
\label{sub:OpenGraphGame}
%%%%%%%%%%%%%%%%%%%%%%%%%%%%%%%%%%%%%%%%%%%%%%%%%
In this section, we prove the following main theorem: 

\begin{theorem}\label{thm:gcgqr}
    \PlanarExtensionGame is \QR-complete.
\end{theorem}
Remember, in the \PlanarExtensionGame, a graph $G$ is provided with an embedding of a subgraph $G'$ of $G$.
Then, the \human and the \devil take turns placing specific vertices of $G \backslash G'$ until one cannot place a vertex without violating planarity.
Whereas \GraphInPolygonGame could directly adapt many gadgets from the \ER-completeness proof of \GraphInPolygon, \PlanarExtensionGame requires novel gadget constructions, as the planarity of the graph prevents the usual collinearity constructions.

We first prove the easy direction.

\begin{lemma}
    \PlanarExtensionGame is in \QR.
\end{lemma}
\begin{proof}
    Let $I$ be a \PlanarExtensionGame instance, which consists of a planar graph $G$, an embedding of some subgraph $G' \subseteq G$, and an ordering on the vertices $V(G) \backslash V(G')$.
    We use the machine model to show \QR-membership.
    The machine model checks after every turn whether the graph is still planar.
    To do this, we use the variables $x_i \in \R^2$ and $y_i \in \R^2$ to describe the vertices specified by the human and the devil.
    After every point placed, the algorithm checks whether the graph is still planar.
    If either player places a vertex that violates the planarity of the graph, then that player loses, and the other player wins the game.
    Otherwise, if all points are placed correctly, the human wins.
    The machine model can verify in polynomial time whether a given drawing of $G$ is planar\cite{LMM}.
    As such, this concludes the proof that \GraphInPolygonGame is in \QR.
\end{proof}

To prove \cref{thm:gcgqr}, we must prove the \QR-hardness of \PlanarExtensionGame.
We reduce from \planarfotrinvequal, by designing gadgets that encode variables, addition, multiplication, et cetera and use them to encode the formulas in \planarfotrinvequal.
The challenge of using gadgets in \PlanarExtensionGame is that, unlike \GraphInPolygonGame, vertices and edges may not be drawn on the `boundary' of the polygon.
In \PlanarExtensionGame, the polygon is replaced by a graph, so any vertices or edges on the `boundary' would violate planarity.
Instead, we make use of the fact that the game has two players;
if one player does not make the move corresponding to some arithmetic operation, we ensure the other player has a winning strategy.
Then, the original player is forced to make the moves corresponding to the arithmetic operations.
While both players can take the role of setting variables according to arithmetic operations, for simplicity we assume the \human always takes this role.
Then, the \devil will always take the role of punishing the \human if they made a different move.

\begin{lemma}
    \PlanarExtensionGame is \QR-hard.
\end{lemma}
\begin{proof}
    We reduce from \planarfotrinvequal.
    Let $I = (\exists x_1 \forall x_2 \dots Q x_n) : \Phi (x_1, \dots, x_n)$ be an instance of \planarfotrinvinequal.
    
    \paragraph{Encoding Variables.}
    We create a gadget to represent the value assigned to a variable in \planarfotrinvequal, which must always be in the range $[1/2,4]$.
    (The variables decided by the devil are instead limited to the range $[3/4, 1]$, but is otherwise constructed the same.
    All auxiliary variable segments will receive a range of $[1/2, 4]$.)
    To do this, we want to have a player place a vertex and ensure this vertex must always be on a line segment.
    Suppose the vertex $v_x$ represents the variable $x$.
    Then, if $v_x$ is placed on one endpoint, it represents $x=1/2$, the other endpoint represents $x=4$, with the value of the positions in between being determined by linear interpolation.
    
    The gadget used by \GraphInPolygonGame would violate planarity, as newly drawn edges would overlap with pre-drawn ones.
    Instead, we present a new gadget.
    Without loss of generality, the \human will be tasked to place a vertex $v_x$ on a variable segment.
    We show that if the \human does not place the vertex on the variable segment, the \devil has a winning strategy.
    Then, the \human must always place the vertex on the variable segment, where it receives a value in the range $[1/2, 4]$.
    The construction of the gadget is shown in \cref{fig:fullgadget}.

    The gadget consists of four subgadgets, each of which ensures that the \human cannot place $v_x$ in some (open) half-plane.
    Two such gadgets ensure that the \human must place $v_x$ on the extension of an edge of the graph.
    The other two gadgets ensure that the \human places $v_x$ in the range $[1/2, 4]$.
    
    In \cref{fig:subgadget}, one subgadget is shown, along with an example of how it ensures the human cannot place $v_x$ in the half-plane away from the gadget.
    Assume the \human places the vertex $v_x$ in the half-plane away from the gadget defined by the dotted line.
    Then, there exists a region within the gadget that $v_x$ cannot see;
    we denoted the extent of $v_x$'s vision of the gadget by the dashed line.
    Now, we ask the \devil to place the point $v_d$ within the convex region of the gadget.
    Finally, we ask the \human to place a vertex inside $v_h$ that is adjacent to both $v_x$ and $v_6$ inside the gadget.
    (Note that in this figure, we have not added additional edges to ensure $v_d$ and $v_h$ cannot be outside the gadget. In the full version, $v_d$ can only be adjacent to $v_5$ and $v_6$ when it is inside this subgadget.)
    Clearly, if $v_h$ is adjacent to $v_x$, then the edge to $v_6$ would intersect both $(v_d, v_5)$ and $(v_d, v_6)$.
    This violates the planarity of the graph, and ensures the \human loses.
    In contrast, if $v_x$ is on the dotted line or further to the right, then the \human can place $v_h$ to the left of $v_d$ within the gadget and be able to draw edges to both $v_x$ and $v_6$ without violating planarity.
    (Indeed, if $v_x$ goes far enough right, it may not be able to view the gadget due to other edges blocking its sight. However, we disregard this issue, as we ensure that within our construction, $v_x$ can never be placed where this is an issue, without already violating different gadgets.)
    
    \begin{figure}
        \centering
        \includegraphics[page=1]{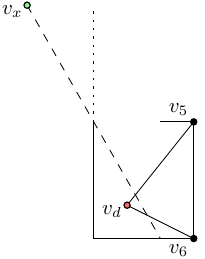}
        \caption{A subgadget of the variable segment gadget, ensuring that $v_x$ does not lie on the other side of the dotted line.
        Displayed is the \human placing $v_x$ on the other side of the dotted line and the \devil placing $v_d$ such that it is not visible from $v_x$.
        This makes it impossible for the \human to place another vertex within the gadget that is adjacent to both $v_6$ and $v_x$.}
        \label{fig:subgadget}
    \end{figure}

    Now, we show how to combine four of these gadgets to encode a variable.
    Here, the only coordinates the \human can place the vertex are exactly those on the variable segment.
    The players will place vertices in the following order:
    \begin{enumerate}
        \itemsep-0.3em 
        \item \human: the vertex $v_x$, which should go on the variable segment.
        \item \devil: a vertex $d_1$ with edges $(v_1, d_1)$ and $(v_2, d_1)$.
        \item \human: a vertex $h_1$ with edges $(v_2, h_1)$ and $(v_x, h_1)$.
        \item \devil: a vertex $d_2$ with edges $(v_3, d_2)$ and $(v_4, d_2)$.
        \item \human: a vertex $h_2$ with edges $(v_4, h_2)$ and $(v_x, h_2)$.
        \item \devil: a vertex $d_3$ with edges $(v_5, d_3)$ and $(v_6, d_2)$.
        \item \human: a vertex $h_3$ with edges $(v_6, h_3)$ and $(v_x, h_2)$.
        \item \devil: a vertex $d_4$ with edges $(v_5, d_4)$ and $(v_7, d_4)$.
        \item \human: a vertex $h_4$ with edges $(v_7, h_4)$ and $(v_x, h_4)$.
    \end{enumerate}

    Now, if the \human places $v_x$ to the bottom or top of the variable segment, step (3) or (9) will be impossible.
    Similarly, if \human places $v_x$ on the left or right of the variable segment, step (5) or (7) will be impossible.
    In all cases, this ensures the \human loses.
    Furthermore, if $v_x$ is placed on the variable segment, then all four subgadgets will be satisfied.

    We point out that the variable segment gadget will only add edges to $v_x$ in one direction;
    we do not need the gadget to be on both sides of the variable segment, unlike the analogous gadget from \GraphInPolygonGame.

    Whenever we model a variable $x$ through a variable segment, we denote the vertex that encodes $x$ by $v_x$.
    We use the notation $v_x : x = c$ to denote the variable $v_x$ if it were to encode $x = c$.

    \begin{figure}
        \centering
        \includegraphics[page=4]{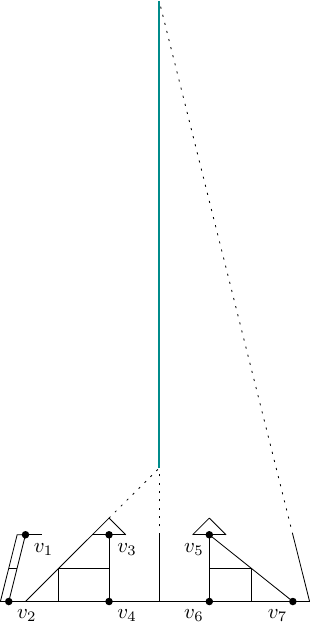}
        \caption{The full construction of a variable segment gadget, with the variable segment itself shown in turquoise. It consists of four subgadgets, which are the four pockets at the bottom of the image. The leftmost one ensures the variable cannot be placed below the variable segment. The two middle subgadgets ensure the variable must be placed on the ray extending from the edge in between them. Finally, the rightmost one ensures the variable cannot be placed above the variable segment. }
        \label{fig:fullgadget}
    \end{figure}

    \paragraph{Variable Ray}
    \begin{figure}[htb]
	\centering   
        \includegraphics[page=5]{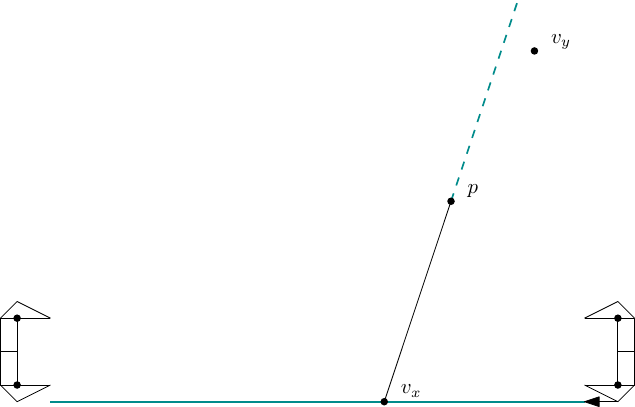}
	   \caption{A variable ray gadget. $v_y$ is placed on the right of the ray extension of $(v_x, p)$, which means it cannot see on the left of $(v_x, p)$. The gadget on the left ensures that, using the methods from \cref{fig:subgadget}, this loses the game for the player who placed $v_y$.}   \label{fig:VariableRay} 
    \end{figure}
    
    The variable segments from the previous section require the gadgets to be fixed at the start of the game.
    However, for the inversion, turn, and addition gadgets, we would like to ensure that a vertex is placed on the extension of an edge drawn by the players.
    Specifically, it should depend on the values chosen for the variables by the players.
    
    Note that we already have a gadget to ensure that a vertex must be placed on the linear extension of a line segment.
    So, we simply wish to extend this to also work for edges that are not fixed by the original graph.
    Our construction is shown in \cref{fig:VariableRay}
    
    To do this, we let the \human place a vertex $v_x$ on a variable segment, but also connect it to a pivot point $p$.
    Then, we can use a construction similar to the variable segment to create a variable ray.
    Above the variable segments, near both of its endpoints, there is a variable segment subgadget.
    This ensures that the vertex $v_y$ has to be placed on the extension of the line segment $(v_x, p)$.
    Otherwise, there will be some area on one side of $(v_x, p)$ that $v_y$ cannot see.
    Using the construction from the variable segment gadget, this would guarantee a loss and ensures that $v_y$ must be placed on the extension of $(v_x, p)$.
    
    However, this one will only ensure that a vertex must be placed on the ray; it will not ensure that the vertex is placed within some part of the ray.
    It would be difficult to ensure that the vertex is placed on a specific part of the ray, as depending on the angle of the ray, this range could change.
    We will use the variable ray in the inversion, turn, and addition gadgets, where it will be a useful tool to project a point onto a different variable segment.
    
    \paragraph{Inversion Gadget}
    Our first application of the variable ray is in the inversion gadget, as depicted in \cref{fig:inversiongadget}.
    The goal of the inversion gadget is that, given two variable segments encoding variables $x$ and $y$, the inversion gadgets ensures that $y = \frac{1}{x}$.
    As such, $x \in [0.5, 2]$. Otherwise, $x > 2$ and $y = \frac{1}{x}  < 0.5$, which is not possible.

    The inversion gadget consists of one horizontal, one vertical variable segment, and a pivot point $p$ lying on the intersection of the line segments $(v_x : x=\frac{1}{2}, v_y : y=2)$ and $(v_x : x=2, v_y : y=\frac{1}{2})$.
    Without loss of generality, we assume a vertex $v_x$ is already present on the horizontal variable segment along with an edge $(v_x, p)$
    We task the \human with inverting $x$ by placing $v_y : y = \frac{1}{x}$.
    Again, we show that if $y$ is set to any other value, the \devil has a winning strategy.
    Notice that $(x, p)$ forms a variable ray gadget and that $y=\frac{1}{x}$ lies on the extension of it.
    As such, $v_y : y = \frac{1}{x}$ is the only position where the vertex $v_y$ will lie both on the variable segment as well as the variable ray.
    If it is not on either the segment or the ray, the \devil has a winning strategy.
    
    \begin{figure}
        \centering
        \includegraphics[page=1]{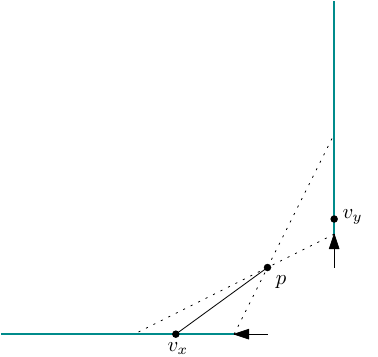}
        \caption{An inversion gadget, where $v_y$ must be placed such that $y = \frac{1}{x}$.}
        \label{fig:inversiongadget}
    \end{figure}

    \paragraph{Turn Gadget}
    The next application of the variable ray gadget is the turn gadget, as depicted in \cref{fig:turngadget}.
    As some gadgets use both horizontal and vertical rays, it is essential that we can copy a value from a horizontal ray onto a vertical ray and vice versa.
    The gadget consists of a horizontal variable segment used to encode the variable $x$, a diagonal variable segment to encode $y$, and a vertical variable segment $z$, along with two pivot points $p_1$ and $p_2$.
    The pivot point $p_1$ is placed at the intersection of the lines $(v_x:x=\frac{1}{2}, v_y:y=4)$ and $(v_x:x=4, v_y:y=\frac{1}{2})$, while $p_2$ is constructed similarly, being placed at the intersection of the lines $(v_y:y=\frac{1}{2}, v_z:z=4)$ and $(v_y:y=4, v_z:z=\frac{1}{2})$.
    
    Without loss of generality, assume the vertex $v_x$ is already present on the horizontal segment along with an edge $(v_x, p_1)$.
    We task the \human to place $y$ and $z$ such that $z = x$.
    Using the variable ray segment, this forces the \human to project $x$ onto the variable segment of $v_y$ such that $v_y : y = f(x)$, where we know $y \in [0.5, 4]$ by construction and $f$ is a bijection.
    Then, using the same construction in reverse, we project $y$ onto $Z$ through $p_2$ in the same manner, by drawing the edge $(y, p_2)$.
    This yields a construction where the \human must place $z$ at $v_z:z = f^{-1}(y) = f^{-1}(f(x)) = x$, as otherwise the \devil has a winning strategy.
    (Note that we do not need to know the exact function $f$, merely that we both have it and its inverse.)

    Note that when necessary, two consecutive turn gadgets may also function as a gadget to copy between two horizontal (or vertical) gadgets whose number scales goes in opposite directions. 
    \begin{figure}
        \centering
        \includegraphics[page=2]{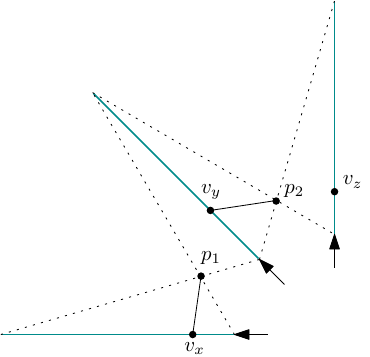}
        \caption{A turn gadget. As $v_y : y = f(x)$ and $v_z : f^{-1}(y)$, $v_z : z = x$.}
        \label{fig:turngadget}
    \end{figure}

    \paragraph{Addition Gadget}
    The addition gadget is the final gadget that applies the variable ray.
    Its goal is to ensure that $x+y=z$, where each of these three variables is placed on a horizontal variable segment.
    For this gadget, we use three variable segments;
    one on the left encoding $x$, the right encoding $z$, and the middle encoding $z$.
    It also includes three pivot points $p_1$, $p_2$, $p_3$, as shown in \Cref{fig:additiongadget}.

    Without loss of generality, we assume $v_x$ and $v_y$ are already placed, along with edges $(x, p_1)$, $(y, p_2)$.
    We task the \human to place $v_z$ at $z=x+y$ and ensure the \devil has a winning strategy if the \human places the vertex elsewhere.
    In addition, the \human will also have to draw $(z, p_3)$.
    Now, the gadget contains three variable ray segments.
    Next, the \human has to place a vertex that is on all three variable rays, originating from $(x, p_1), (y, p_2)$, and $(z, p_3)$ respectively.
    We claim that unless $z = x+y$, these three variable rays will not intersect in the same point, meaning the \devil has a winning strategy.

    Note that the geometry of this gadget is the same as in \GraphInPolygonGame.
    As such, we use the same proof strategy.
    For reference, we again use \Cref{fig:AdditionProof}.
    However, note that instead of using vertices of polygonal holes, $a$, $b$, and $c$ will be the pivot points denoted by $p_1$, $p_2$, and $p_3$ in \Cref{fig:additiongadget}.
    We denote by $\ell$ the line that contains the variable segments
    of $x$ and $y$. We denote by $t$ half the distance that
    $z$ moves. 
    Note that the geometric interpretation of $t$, as indicated in Figure~\ref{fig:AdditionProof}, ensures it moves half the distance of $z$, as it is half the distance from $b$.
    We need to show $d = 2t$.
    The lengths $A,A',B,B'$ are defined as shown in \Cref{fig:AdditionProof}.
    Note that $B' = 2B$, because $\|a-b\| = \|b-c\| $.
    Similarly, follows $A'= 2A$.
    The lemma follows from
    \[d = B'-A'=2(B-A) = 2t.\]
    
        \begin{figure}
        \centering
        \includegraphics[page=3]{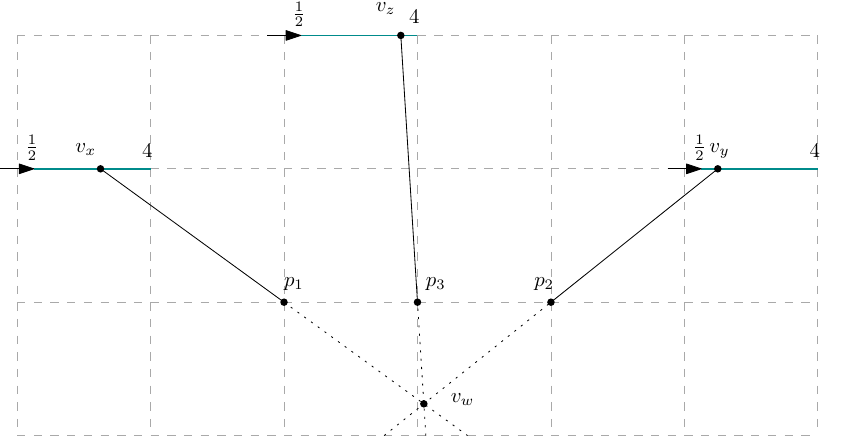}
        \caption{An addition gadgets. The three rays will only collide in a single point if $z = x + y$.}
        \label{fig:additiongadget}
    \end{figure}

    \paragraph{Copy and Split Gadgets}
    Next, we want to create a copy gadget; a gadget that copies the value of one variable onto another, i.e. $v_y : y = x$.
    We would like to emphasize that it is possible to use the variable ray for the copy gadget in a similar manner to the turn gadget.
    In fact, we can repeat turn gadgets to give rise to a copy gadget.
    However, we want to use multiple copy gadgets on the same variable segment within the split gadget.
    This is not possible with a simple variable ray construction, so we present a different, albeit slightly more complicated version of the copy gadget that resembles the copy gadget from \GraphInPolygonGame.
    
    \begin{figure}
        \centering
        \includegraphics[page=2]{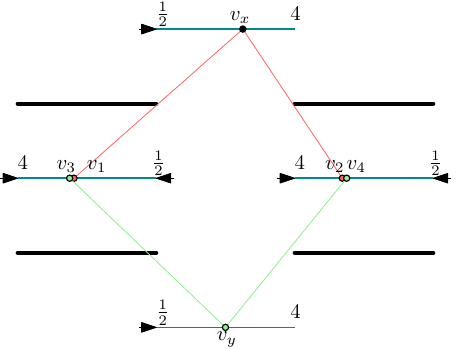}
        \caption{A copy gadget. Unless $v_y : y = x$, the \devil can place $v_1$ and $v_2$, such that either $(v_y, v_3)$ or $(v_y, v_4)$ must violate planarity.}
        \label{fig:placeholder2}
    \end{figure}
    
    This copy gadget resembles the copy gadget from \GraphInPolygonGame in overall geometry.
    However, it makes use of the alternating turns to avoid the collinearities in \GraphInPolygonGame, which would violate planarity in \PlanarExtensionGame.
    The construction of the gadget is displayed in \Cref{fig:placeholder2}.
    We aim to ensure that $v_y: y = x$ through this gadget.
    
    For ease of notation, for vertices $v_1, v_2, v_3, v_4$, we abuse the notation to let them denote both the vertex and the value encoded by the vertex.
    
    \begin{enumerate}
        \item At the start, the black segments are predrawn edges and the blue line segments represent variable segments. 
        Notably, the middle variable segments have a variable segment gadget on both sides.
        The left sides connect to $v_2$ and $v_3$ respectively, while the right sides connect to $v_1$ and $v_4$.
        \item We assume vertex $v_x$ is already drawn by some player and the \human must copy this value such that $v_y : x = y$.
        \item The \human draws $v_y$.
        \item The \devil draws $v_1$ and $v_2$. Due to the construction, these are limited by the constraints $v_1 < x$ and $v_2 > x$.
        \item The \human draws $v_3$ and $v_4$, where $v_1 < v_3 < y$ and $y < v_4 < v_2$.
        
    \end{enumerate}
    We claim that the \devil can force the \human to lose if and only if the \human does not give $v_2$ the same value as $v_1$.
    
    First, suppose the \human sets $v_y : y = x$.
    Then, we see that the constraints ensure $x = y > v_3 > v_1$.
    The \devil only picks a value $v_3$ such that $v_3 < v_1$.
    However, regardless of the value of $v_3$, the \human can pick a number $x > v_3 > v_1$.
    The case for $v_2$ and $v_4$ is symmetric.
    
    Suppose the \human sets $v_y: y \neq x$.
    Without loss of generality, assume $v_y:y < x$ (as in \cref{fig:placeholder2}, the other case is symmetric).
    Then, the \human always loses upon placing $v_3$.
    The \devil can place the vertex $v_1$ such that $y < v_1 < x$
    However, this forces the \human to place the vertex $v_3$ with $v_3 < y < v_1 < v_3$, which is a contradiction.

    \paragraph{Putting it all together}
    Now we have recreated all the gadgets, we show that we can indeed use these to model \planaretrinv using \PlanarExtensionGame within polynomial time.
    We achieve this in a comparable fashion as \GraphInPolygonGame.

    All gadgets are based on axis-aligned variable segments.
    As such, we want to generate an orthogonal drawing of the variable-constraint graph $G$ of $\Phi$.
    In a rectilinear drawing, each vertex can have degree at most four.
    As each variable can be used an arbitrary amount of times in \planarfotrinv, after placing a variable, we duplicate its vertex until each vertex has degree at most three.
    Finally, we generate a rectilinear drawing where each vertex has degree at most three.
    This can be done in polynomial time \cite{NishizekiRahman}.

    The edges of $D$ act as wires between the gadgets and we replace each horizontal and vertical segment by a copy gadget, and   
    replace every $90$ degree corner by a turn gadget. 
    Every splitter vertex and constraint vertex
    will be replaced by the corresponding gadget, possibly using turn gadgets.
    Finally, if two edges intersect, we add a crossing gadget.
    Finally, we also add edges between the gadgets, such that they will not interfere with each other.

    It is easy to see that this can be done in polynomial time, 
    as every gadget has a constant size description and can be
    described with rational numbers, although, we did not do it explicitly.
    In order to see that we can also use integers, note that we can scale everything with the least common multiple of all the denominators of all numbers appearing. This can also be done in polynomial time.

    Finally, to ensure the \human wins when the two players complete the full construction, we add a single vertex at the end of the game for the \devil to place.
    This vertex will always violate planarity.
    For simplicity, suppose a planar $K_4$ graph is already drawn and the \devil is forced to draw a vertex adjacent to all of the vertices in the $K_4$ graph, thus having to draw a $K_5$.
    Clearly, this is impossible and thus forces the \devil to lose if the \human manages to complete the full construction.
\end{proof}

%%%%%%%%%%%%%%%%%%%%%%%%%%%%%%%%%%%%%%%%%%%%%
\newpage
\section{Order Type Game}
\label{sec:OrderTypeGame}
%%%%%%%%%%%%%%%%%%%%%%%%%%%%%%%%%%%%%%%%%%%%%

This section is dedicated to showing the following theorem.

\OrderTypeGameTHM*

The proof of Theorem~\ref{thm:OrderTypeGame} is similar to the classical proofs of the Mn\"ev universality theorem for pseudoline stretchability~\cite{mnev1988universality}, and in particular the version presented by Shor~\cite{shor1991stretchability} where he proved (without naming it) $\exists \mathbb{R}$-hardness of order type realizability. While subsequent work of Richter-Gebert~\cite{richter1995mnev} has streamlined this approach into an arguably cleaner and simpler, more geometric reduction, our focus here is to explain why these techniques readily adapt to the general case of \fotr formulas compared to just existential formulas. Shor's proof, while cumbersome, is conceptually simple and thus lends itself well to this objective. We recommend the presentation of Matou\v{s}ek~\cite{matousek2014intersection} for a nice introduction to these seminal proof techniques.

The proof is obtained by reduction from an intermediate problem that we first introduce.

\begin{definition}
      In the problem \BOfotr, we are given a quantified formula $\exists x_1\forall x_2 \dots Q_n x_n : \Phi(x_1,\dots,x_n)$, where $\Phi$ consists of a conjunction between a set of equations of the form $x=1,\quad x+y=z,\quad x\cdot y=z$ or $x<y$
    for $x,y,z \in \{x_1, \ldots, x_n\}$.
    Each quantifier bounds exactly one variable and the quantifiers keep alternating between $\exists$ and $\forall$.
    The goal is to decide whether the system of equations has a solution when each variable is restricted to the range $(-1,1)$.
\end{definition}

It is immediate from the definition that \BOfotr belongs to \QR. The most important difference between \BOfotr and \fotrinv is that in the former problem, the ranges restricting the variables are open. As explained in the introduction, this makes it much easier to prove completeness for \BOfotr: intuitively this is the case because $(-1,1)$ and $\R$ are homeomorphic. It directly follows from the techniques in~\cite{SS25} that \BOfotr is \QR-complete: Corollary 1.4 in that paper establishes a stronger result that shows that completeness holds for each value $k$ of the number of alternation of quantifiers and does not even require equalities.

 The first step in the proof of Theorem~\ref{thm:OrderTypeGame} is to put a \BOfotr formula in a specific normal form that is well tailored for our problem. The second step is to use classical geometric constructions known as von Staudt's constructions to encode this normal form into an instance of \OrderTypeGame.

We first prove the easy direction of Theorem~\ref{thm:OrderTypeGame}.

\begin{lemma}\label{lem:OTGQR}
\OrderTypeGame is in $Q\mathbb{R}$.
\end{lemma}

\begin{proof}
    We use the machine model to show \QR-membership.
    We use the variables $x_i \in \R^2$ and $y_i \in \R^2$ to describe the points specified by the human and the devil respectively.
    The algorithm then checks for every newly added point that it forms the order type specified by the input.
    If one of the $x_i$ or $y_i$ does not adhere to this then the human, respectively the devil, is declared to lose.
    Otherwise, if all points are placed correctly the human wins.
\end{proof}

\subsection{Totally ordered FOTR}

In order to prove the reverse direction, we introduce the following variant of \fotr, which we phrase as a computational problem.

\begin{definition}[\TOfotr]
In the problem \TOfotr, we are given:
\begin{itemize}
    \item a quantified formula \[\exists a_0, a_1, b_1 \ldots a_n, b_n \exists x_1 \in (a_1,b_1) \forall x_2 \in (a_2,b_2) \ldots Q_n x_n \in (a_n,b_n): \Phi(a_0,a_1, b_1,\ldots a_n, b_n, x_1, \ldots, x_n),\] where $Q_n=\exists$ or $\forall$ and $\Phi$ is a conjunction of equations of the form $x+y=z$ or $x \cdot y = z$,
\item and a linear order $\prec$ on the set of variables $Z=\{a_i,b_i,x_i\}_{i \in [n]}$ so that for all $i$, $a_i \prec x_i \prec b_i$, and if $x_i$ is quantified universally, there are no other variables $z$ such that $a_i \prec z \prec b_i$.
\end{itemize}

 The goal is to decide whether the formula has a solution in which $a_0=1$, all the variables have value greater than one, and the order of the variables given by their real values matches the order $\prec$.
\end{definition}

The main differences between a \BOfotr instance and a \TOfotr instance is that in the latter we are (1) enforcing all the variables except $a_0$ to have value greater than $1$ and $(2)$ we are enforcing a linear order on the variables and a \emph{scaffolding} of each of the $x_i$ variables in intervals $(a_i,b_i)$. Note that although this scaffolding is confirmed by the linear order, it is not redundant since it is enforced when the variables are chosen under universal quantifications, while the total order is only checked at the end. This is the same difference as between $\forall y, y < 2$, which is false, and $\forall y< 2, y< 2$, which is true.

We can turn every \BOfotr formula into a \TOfotr instance, which has the following computational pendant:

\begin{lemma}
The problem \TOfotr is $Q\mathbb{R}$-hard.
\end{lemma}

Our proof is similar to that of Shor~\cite{shor1991stretchability}.

\begin{proof}[Proof]
  We reduce from \BOfotr. The first step of the reduction is to enforce that all variables be greater than $1$. In order to do this, we pick a universal constant $c$ and replace each variable $x$ with a variable $x'$ with the intended meaning that $x=x'-c$. If $c$ is large enough, this makes $x'$ larger than $1$. Since in \BOfotr, all the variables are in $(-1,1)$, we can pick $c=3$ which will restrict all the variables $x'$ to lie in $(2,4)$. Then, in all the equations of $\Phi$ we replace $x$ by $x'-c$ and then decompose the equation further into elementary blocks of the type $x+y=z$ and $x \cdot y = z$.  It is easy but somewhat cumbersome to see that the decomposition into elementary equations can also be done while ensuring that all the intermediate variables are larger than $1$ (see the equations in~\cite[pp.550-551]{shor1991stretchability}). Furthermore, since each variable in \BOfotr is restricted to lie in $(-1,1)$, one can verify in these equations that all the intermediate variables can also be assumed to be upper bounded by some universal constant.

  The second reduction step is more delicate and aims at enforcing a linear order on all the variables. The key technical point behind many proofs of the Mn\"ev universality theorem is to enforce this without actually solving the system of equations. The starting idea of Shor is to use an auxiliary variable $d$ much larger than all the variables $x_i$ (which is easily done due to the universal upper bound in the previous paragraph), and to encode each variable $x_i$ as a new variable $y_i=d^i+x_i$. A different power of $d$ is used for each variable, ensuring the fixed linear order we require since by construction \[d < d+x_1 < d^2< d^2+x_ 2 < ... < d^n <d^n+x_n < d^{n+1}.\]

  There remains to transform the elementary equations using these new variables $y_i$, and to decompose them into elementary equations where the total ordering of the variables is still known. The key here is to use, for each elementary equation (addition, multiplication or strict inequality), another set of unused powers of $d$ to encode that equation. For example, to encode the inequality $x_i<x_j$, we take some $\alpha$ larger than all previously used powers and write $z_i=(d^\alpha-d^i)+y_i$, $z_j=(d^\alpha-d^j)+y_j$ and $z_i<z_j$ (actually inequalities are not explicitly written in the formula $\Phi$ but are specified by the linear order instead). The total ordering of all the intermediate variables here is known, since
  \[ d^i < y_i< d^j<y_j < d^\alpha-d^j < d^\alpha - d^i<  d^\alpha < z_i <z_j<d^{\alpha+1}.\]

With a use of different powers of $d$ and a careful decomposition into elementary equations, one can decompose likewise each equation in the variables $y_i$ into elementary equations where the total ordering of all the variables is known. Here again, we refer to Shor~\cite[pp.550-553]{shor1991stretchability} for the complete list of elementary equations and the resulting total order. The restriction that variables have value greater than $1$ is key at this stage, since it allows us for example to immediately order $x_i \cdot x_j$ or $x_i +x_j$ as larger than $x_i$.
  
There remains to explain how all the variables should be quantified. It suffices to define all the needed powers of $d$ using existential quantifiers, and then introduce each $y_i$ variable in the same order and using the same quantifier as for the corresponding $x_i$ variable. Each of these $y_i$ variables comes with a scaffolding $y_i \in (d^i+2,d^{i}+4)$, and by construction there is no other variable within this scaffolding. Then, our reduction involves a large set of intermediate helper variables coming from the first and the second reduction step, and those should also be quantified. All of these helper variables are introduced at the end and can be controlled by existential quantifiers since their values are uniquely defined from variables that were previously set. The scaffolding $(a_i,b_i)$ of these intermediate variables is chosen to be always satisfiable, with $a_i=c$ and $b_i$ large enough, for example $b_i=d^{\alpha+1}$ where $\alpha$ is the largest power of $d$ previously used in the reduction.

In the end the satisfiability of the \TOfotr formula we obtain is equivalent to the satisfiability of the initial \BOfotr formula, which concludes the proof.
\end{proof}

\subsection{The projective line and von Staudt constructions}

We now explain how to encode a \TOfotr instance into an \OrderTypeGame instance. We first introduce the classical geometric constructions underpinning that reduction: cross-ratios on the projective line and von Staudt constructions. We refer to Richter-Gebert~\cite{richter2011perspectives} for an introduction to projective geometry, but we will use very little of it: in order to understand the proof it suffices to know that it is an extension of the Euclidean geometry of $\mathbb{R}^2$ with a line at infinity, that projective transformations preserve colinearity, that any two distinct lines intersect exactly once, that using projective transformations one can send any line to infinity and that two lines crossing at infinity are parallel.

The setup for the reduction starts with a projective line $\ell$, which is a line in $\mathbb{R}^2$ containing three distinguished points $\mathbf{0}$, $\mathbf{1}$ and $\II$. A point $\mathbf{x}$ in the line $\ell$ is interpreted as a real number $x$ via its \emph{cross-ratio}: if we denote by $d(\mathbf{a},\mathbf{b})$ the signed distance between two points $\mathbf{a}$ and $\mathbf{b}$ on $\ell$, it is defined as

\[(\mathbf{x},\mathbf{1}; \mathbf{0},\II) := \frac{d(\mathbf{x},\mathbf{0}) \cdot d(\mathbf{1},\II)}{d(\mathbf{x},\II) \cdot d(\mathbf{1},\mathbf{0})},\]
and we write just $x$ as a shorthand for this cross-ratio. A main property is that this quantity is invariant under projective transformations. Also note that when $\II$ is sent to infinity using such transformations and $d(\mathbf{0},\mathbf{1})=1$, the cross-ratio $x$ matches the distances $d(\mathbf{x},\mathbf{0})$. Therefore we can think of $\ell$ and the cross-ratios as a projective-invariant version of a standard real axis.

The \emph{von Staudt constructions} are geometric gadgets allowing to encode addition and multiplication of real numbers using incidences of lines in $\mathbb{R}^2$. Both constructions work with an arbitrary second line $\ell_\infty$ in $\mathbb{R}^2$ crossing $\ell$ exactly at the point $\infty$. For $\mathbf{a} \neq \mathbf{b}$, we denote by $(\mathbf{a}\mathbf{b})$ the unique projective line going through $\mathbf{a}$ and $\mathbf{b}$.

\vspace{1em}

\begin{figure}[h]
    \centering
     \def\svgwidth{\textwidth}
     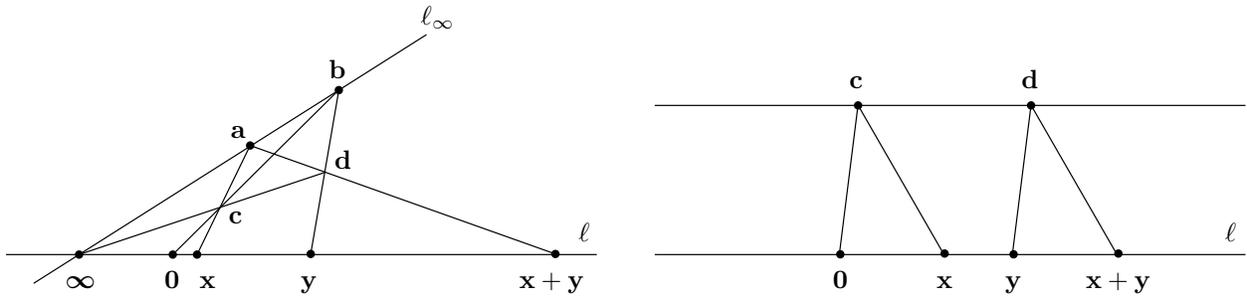
    \caption{Encoding addition geometrically. 
     }
    \label{fig:addition}
\end{figure}

\textbf{Addition.} We refer to Figure~\ref{fig:addition} left for an illustration of this construction. In order to define the addition of two points $\mathbf{x}$ and $\mathbf{y}$ on $\ell$, we use two helper points $\mathbf{a}$ and $\mathbf{b}$ anywhere on $\ell_\infty$. We denote the intersection of $(\mathbf{x}\mathbf{a})$ and $(\mathbf{0}\mathbf{b})$ by $\mathbf{c}$ and the intersection of $(\II\mathbf{c})$ and $(\mathbf{y}\mathbf{b})$ as $\mathbf{d}$. Then the point $\mathbf{x}+\mathbf{y}$ is the point on $\ell$ at the intersection with $(\mathbf{a}\mathbf{d})$.

The geometric intuition behind this construction, pictured on Figure~\ref{fig:addition}, right, is that when one sends $\ell_\infty$ (and in particular $\II$) to infinity using a projective transformation, the line $\mathbf{(cd)}$ becomes parallel to $\ell$, the lines $(\mathbf{0c})$ and $(\mathbf{yd})$ become parallel and $(\mathbf{cx})$ and $(\mathbf{d (x+y)})$ also become parallel. Then basic Euclidean geometry shows that $d(\mathbf{0},\mathbf{x+y})=d(\mathbf{0},\mathbf{x})+d(\mathbf{0},\mathbf{y})$, and thus $(x+y)=x+y$, as desired.

\vspace{1em}

\begin{figure}[h]
    \centering
    \def\svgwidth{\textwidth}
    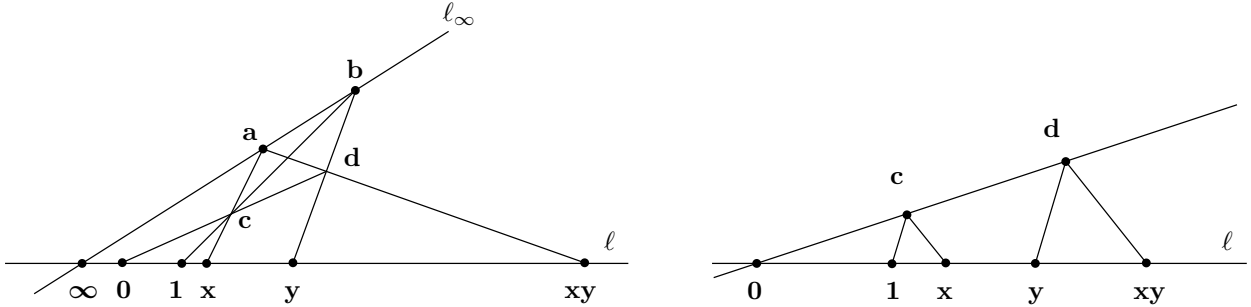
    \caption{Encoding multiplication geometrically. }
    \label{fig:multiplication}
\end{figure}

\textbf{Multiplication.} The gadget for multiplication is only slightly more involved and is pictured on Figure~\ref{fig:multiplication} left. The points $\mathbf{a}$ and $\mathbf{b}$ are still located anywhere on $\ell_\infty$, and $\mathbf{c}$ is on the intersection of $(\mathbf{x}\mathbf{a})$ and $(\mathbf{1}\mathbf{b})$, while $\mathbf{d}$ is at the intersection of $(\mathbf{0}\mathbf{c})$ and $(\mathbf{y}\mathbf{b})$. Then here again the point $\mathbf{x} \cdot \mathbf{y}$ is the point on $\ell$ at the intersection with $(\mathbf{a}\mathbf{d})$.

As pictured on Figure~\ref{fig:multiplication}, right, sending the line $\ell_\infty$ at infinity allows for a clean geometric picture where one can verify that $d(\mathbf{0},\mathbf{x\cdot y})=d(\mathbf{0},\mathbf{x})\cdot d(\mathbf{0},\mathbf{y})$ and thus $(x \cdot y)= x \cdot y$ as desired.

\subsection{The reduction}

We are now ready to prove Theorem~\ref{thm:OrderTypeGame}.

\begin{proof}[Proof of Theorem~\ref{thm:OrderTypeGame}]
  The reduction from \TOfotr into \OrderTypeGame first encodes the formula $\Phi$ as an order type using the von Staudt constructions. We fix one projective line $\ell$ with three distinguished points $\mathbf{0}$, $\mathbf{1}$ and $\II$. The constant $a_0=1$ corresponds to the $\mathbf{1}$ point. We take an arbitrary point $\mathbf{p}$ not on $\ell$, which defines with $\II$ the line $\ell_\infty$. Then, for each elementary equation $x+y=z$ and $x \cdot y=z$, we apply the corresponding von Staudt gadget. We always use the same line $\ell_\infty$ but the helper points $\mathbf{a}$ and $\mathbf{b}$ are different for each elementary equation. There is an ambiguity here, as different choices of $\mathbf{a}$ and $\mathbf{b}$, and thus of $\mathbf{c}$ and $\mathbf{d}$, could lead to different arrangements of lines in $\mathbb{R}^2$. Therefore this contruction does not uniquely define an order type. This ambiguity is resolved by taking, for each new von Staudt gadget being built, a pair of points $\mathbf{a}$ and $\mathbf{b}$ that are infinitely farther out on the $\ell_\infty$ line (but the two of them still close together) than all the previously used points. This choice ensures that in each gadget, the relative placement of the four new points compared to the already placed points is fully specified: intuitively these four new points are all clustered together very far away, with $\mathbf{a}$ and $\mathbf{b}$ on the $\ell_\infty$ line and $\mathbf{c}$ and $\mathbf{d}$ a tiny bit below it.

  This combinatorics of the line arrangement behind this construction can be directly computed from the \TOfotr instance and abstracted in the form of an order type. Such an order type needs to specify, for any three points, whether they are aligned, and if they are not it specifies their orientation. Since we are reducing from \TOfotr, the linear order $\prec$ of all the variables $a_i, b_i$ and $x_i$ is fixed by assumption. This order is enforced on the line $\ell$ via the chirotope values $\chi(p,\cdot,\cdot)$. Let us insist here that the boundaries of the scaffolding intervals $(a_i,b_i)$ are themselves variables, and thus are also points which are placed on the line $\ell$. Therefore our construction enforces that for all $i$, $a_i < x_i < b_i$. This linear order and our choice of sending $\mathbf{a}$ and $\mathbf{b}$ infinitely farther than all previous points, while being geometrically fuzzy, leads to a unique and well-defined order type (this was first observed by Mn\"ev~\cite[Paragraph~4.3]{mnev1988universality}, see also Shor~\cite[p. 549]{shor1991stretchability} and Richter-Gebert~\cite[Section~2.3]{richter1999universality}).

  Summarizing what we have done so far, we have defined an abstract order type that is realizable if and only if there exist $a_i,b_i$ and $x_i$ such that the formula $\Phi(a_0,a_1,b_1, \ldots a_n, b_n, x_1, \ldots, x_n)$ is true, $a_0=1$, all the other variables are greater than one and the order given by their real values matches the input linear order. We now explain how the alternation of quantifiers can be modeled as a game on a (slightly different) order type. In order to do this, we have some room for manoeuvre that we have not yet used: first we can add some dummy points in our order type which do not affect its realizability, and second we get to choose the order in which the points arrive, and thus which players is asked to place them.

  To explain the first point, let us define a \emph{dummy extension} of an order type $O$ as an order type $O'$ with one additional element $n$ with a unique possible realization. For instance, $\mathbf{0}$ is the unique intersection point of $\ell$ and $\ell_\infty$, so if $O$ has a pair of points on each of these lines, one can force $n$ to be realized exactly at $\mathbf{0}$ in any realization. In terms of chirotope, this is enforced by setting $\chi(n,x,y)=\chi(0,x,y)$ for any elements $x$ and $y$, where $0$ is the element represented by $\mathbf{0}$. It directly follows from this definition that $O$ is realizable if and only if $O'$ is. We use such dummy extensions as a way to force the \devil to pass a turn: if the \devil is given the new point of a dummy extension of the current order type, it is forced to place that point, which will have no influence on the outcome of the game.

  We can now conclude the proof. We start with an instance of \TOfotr. Each player is successively given points to place as follows (when the \human is given consecutive turns, the \devil is given dummy extensions inbetween). Initially, the human is given the points $\mathbf{0}, \mathbf{1}$, $\II$ and $\mathbf{p}$ to place. Then the human is given all the points $\mathbf{a_i}$ and $\mathbf{b_i}$ to place. Then the real game begins: following the alternations of the \TOfotr instance, the permutation $\pi$ gives the existentially quantified variables to the \human and the universally quantified variables to the \devil. Note that the constraint that players have to respect the order type of the points placed until now exactly enforces the correct scaffolding of the variables: this uses the fact that the ranges in \TOfotr are open (this is why we did not reduce from \fotrinv), and that no other variable lies within the scaffolding of a universally quantified variable. Once all these variables are placed, the \human is given all the remaining points to place (those are the points involved in the von Staudt gadgets). Finally, if the game is not over by now, the \devil is given one last impossible point to place, for example a point $d$ so that $\chi(d,\cdot,\cdot)$ is identically zero.

  We claim that the \human has a winning strategy if and only if \TOfotr is satisfiable. For the forward direction, if the \human wins the game, it means that the \devil could not place a point. By construction, the \devil can always place the points coming from variables of \TOfotr, since they are only constrained as being on $\ell$ and inbetween their scaffolding points, which can always be satisfied. So this means that we have reached the last point $d$, and all the other points have been placed. In particular, this means that all the equations in $\Phi$ are satisfied. Since the other conditions of \TOfotr are satisfied by construction and the universally quantified variables have been placed adversarially by the \devil, this implies that \TOfotr is satisfiable. For the reverse direction, the winning strategy for the human is to place all the helper points as per the von Staudt constructions and the variables on $\ell$ as per the satisfying values for the \TOfotr instance, taking into account the adversarial values on $\ell$ chosen by the devil. This forces the \devil to lose since the last point they are tasked with placing is always impossible to place. This concludes the proof.
\end{proof}

\bibliographystyle{alphaurl}
\bibliography{literature}
\end{document}